\documentclass[12pt]{article}
\usepackage{amsmath}
\usepackage{graphicx}
\usepackage{enumerate}
\usepackage{xr-hyper}
\usepackage{natbib}
\usepackage{url} 
\usepackage{subfig}
\usepackage{comment}
\usepackage{nicematrix}

\newcommand{\blind}{1}

\makeatletter
\newcommand*{\addFileDependency}[1]{
  \typeout{(#1)}
  \@addtofilelist{#1}
  \IfFileExists{#1}{}{\typeout{No file #1.}}
}
\makeatother

\usepackage{amssymb,amsfonts,mathrsfs,amsmath,amsthm,mathtools,bm,dsfont, thmtools,bbm}\allowdisplaybreaks
\usepackage[dvipsnames]{xcolor}
\usepackage{array, graphicx, graphics,float,multirow,tabularx,tikz,pgfplots, longtable,threeparttable}
\usepackage{setspace}
\usepackage[colorlinks,citecolor=blue,urlcolor=blue, runcolor=blue, filecolor=cyan]{hyperref}
\usepackage{footnotebackref}
\usepackage{bibentry, natbib} 
\usepackage{paralist}  
\usepackage{appendix} 
\usepackage{mathrsfs}
\usepackage{enumitem}
\usepackage{authblk}
\usepackage{caption}
\usepackage{titlesec}
\usepackage{textcase,relsize}
\usepackage{booktabs}

\makeatletter
\@ifundefined{newcounteralias}{}{%
  \renewcommand\thmt@autorefsetup{%
    \@xa\def\csname\thmt@envname autorefname\@xa\endcsname\@xa{\thmt@thmname}}%
}
\makeatother

\declaretheoremstyle[notefont=\bfseries,notebraces={}{},%
    headpunct={},postheadspace=1em]{mystyle}
\declaretheorem[style=mystyle,numbered=no,name=Assumption]{asmp-hand}
\declaretheorem[style=mystyle,numbered=no,name=Condition]{cond-hand}
\declaretheorem[style=mystyle,numbered=no,name=Example]{exmp-hand}

	\def \calA {\mathcal{A}}		
			
	\def \calC {\mathcal{C}}		
	\def \calD {\mathcal{D}}		
\def \bbE {\mathbb{E}}	\def \calE {\mathcal{E}}		
	\def \calF {\mathcal{F}}		
			
	\def \calH {\mathcal{H}}

	\def \calM {\mathcal{M}}		
\def \bbN {\mathbb{N}}	\def \calN {\mathcal{N}}		
	\def \calO {\mathcal{O}}		
	\def \calP {\mathcal{P}}		
			
\def \bbR {\mathbb{R}}	\def \calR {\mathcal{R}}		
	\def \calS {\mathcal{S}}		
	\def \calT {\mathcal{T}}

	\def \calW {\mathcal{W}}		
	\def \calX {\mathcal{X}}		
			
	\def \calZ {\mathcal{Z}}

\def \conP {\overset{\bbP}\longrightarrow}

\def \diag {\textrm{diag}}

\newcommand{\norm}[1]{\left\Vert#1\right\Vert}
\newcommand{\abs}[1]{\left\vert#1\right\vert}

\usepackage{subfig}

\theoremstyle{definition}
\newtheorem{remark}{Remark}

\newtheorem{assumption}{Assumption}[section]
\newtheorem*{assumption*}{Assumption}

\theoremstyle{plain}
\newtheorem{theorem}{Theorem}[section]

\newtheorem{lemma}[theorem]{Lemma}
\newtheorem{corollary}[theorem]{Corollary}

\def \conP {\to_p}

\def \bbE {\mathbb{E}}	\def \calE {\mathcal{E}}		
\def \calP {\mathcal{P}}

\def \calW {\mathcal{W}}

\allowdisplaybreaks

\usepackage{algorithm}
\usepackage{algpseudocode}

\makeatletter
\newenvironment{breakablealgorithm}
{
	\begin{center}
		\refstepcounter{algorithm}
		\hrule height.8pt depth0pt \kern2pt
		\renewcommand{\caption}[2][\relax]{
			{\raggedright\textbf{\fname@algorithm~\thealgorithm} ##2\par}%
			\ifx\relax##1\relax 
			\addcontentsline{loa}{algorithm}{\protect\numberline{\thealgorithm}##2}%
			\else 
			\addcontentsline{loa}{algorithm}{\protect\numberline{\thealgorithm}##1}%
			\fi
			\kern2pt\hrule\kern2pt
		}
	}{
		\kern2pt\hrule\relax
	\end{center}
}
\makeatother

\usepackage{minitoc}

\begin{document}

\doparttoc 
\faketableofcontents 

\def\spacingset#1{\renewcommand{\baselinestretch}%
{#1}\small\normalsize} \spacingset{1}


\if1\blind
{
	\title{Flexible Matrix Estimation with Side Information\thanks{This research was supported by NSF grant DMS-2052955.}
	}
	\author{Anish Agarwal,\qquad Jungjun Choi\footnote{Address for Correspondence: Computer Science and Statistics Department, Tyler Hall, 9 Greenhouse Road, Kingston, RI, 02881, Email: jungjun.choi@uri.edu},\qquad Ming Yuan\\
		Department of IEOR, Columbia University\\
            Department of CS \& Statistics, University of Rhode Island\\
        Department of Statistics, Columbia University}
	\maketitle
} \fi

\if0\blind
{
  \title{Flexible Matrix Estimation with Side Information}
	\maketitle
} \fi

\smallskip

\spacingset{1.12} 

\begin{abstract}
We introduce a flexible framework for high-dimensional matrix estimation to incorporate side information for both rows and columns. Existing approaches, such as inductive matrix completion, often impose restrictive structure---for example, an exact low-rank covariate interaction term, linear covariate effects, and limited ability to exploit components explained only by one side (row or column) or by neither---and frequently omit an explicit noise component. To address these limitations, we propose to decompose the underlying matrix as the sum of four complementary components: (possibly nonlinear) interaction between row and column characteristics; row characteristic-driven component, column characteristic-driven component, and residual low-rank structure unexplained by observed characteristics. By combining sieve-based projection with nuclear-norm penalization, each component can be estimated separately and these estimated components can then be aggregated to yield a final estimate. We derive Frobenius- and entrywise-norm convergence rates that demonstrate the robustness of our estimator across a broad range of model configurations with varying degrees of informativeness in the side information. We further extend the method to partially observed matrices under both missing-at-random and missing-not-at-random mechanisms, including block-missing patterns motivated by causal panel data. Simulations and a real-data application to tobacco sales show that leveraging side information improves imputation accuracy and can enhance treatment-effect estimation relative to standard low-rank and spectral-based alternatives.
\end{abstract}

\vfill

\part{} 


\spacingset{1.7}

\section{Introduction}

Recent technological progress has made it possible to gather and process high-volume data that are conveniently organized in matrix form, often with both dimensions scaling up rapidly. Accordingly, high-dimensional matrix estimation problems such as matrix denoising and matrix completion have attracted considerable attention, and many impressive results have been obtained from both statistical and computational perspectives. However, although side information is often available in addition to the target outcome data, traditional approaches typically use only the outcome data for matrix estimation. Incorporating side information can enrich the underlying model and improve estimation and prediction accuracy. As our ability to access auxiliary covariate data continues to grow, developing matrix estimation methods that effectively leverage side information has become an important and timely research direction.

A number of computational algorithms, along with their statistical properties, have been proposed recently. Arguably, the most popular model that incorporates additional information in matrix estimation is the Inductive Matrix Completion (IMC) model  \citep[e.g.,][]{xu2013speedup,jain2013provable,zhang2018fast}. The standard IMC model takes the form:
$$
Y = M = X L Z^\top
$$
where $Y = [y_{it}]_{i \leq N,, t \leq T}$ is the outcome matrix, $X = [x_1, \ldots, x_N]^\top$ is the $N \times d_1$ row-feature matrix, $Z = [z_1, \ldots, z_T]^\top$ is the $T \times d_2$ column-feature matrix, and $M$ is the target matrix. Here, $L$ is assumed to be a low-rank $d_1 \times d_2$ matrix. A typical estimation approach is to solve
$$
\min_{L \in \bbR^{d_1 \times d_2}} \norm{\calP_{\Omega} (XLZ^\top - Y)}_F^2 + \lambda \norm{L}_{*},
$$
where $||\cdot||_*$ denotes the nuclear norm, $\calP_{\Omega}(A) = \Omega \circ A$, and $\Omega$ is the $N \times T$ indicator matrix for observability in matrix completion.

Although IMC is a popular and useful approach to matrix estimation with side information, it has several important limitations. First, it requires that the features be present on both sides and also interact linearly. Moreover, it predicates upon the informativeness of both row and column features and can break down if features are weak or irrelevant. Several extensions of IMC have been proposed in recent years to address these shortcomings. \cite{ledent2023generalization} incorporate a noise component $E$ into the IMC model and derive bounds on the expected $\ell$-risk. \cite{zhong2019provable} allow a nonlinear relationship between $(X, Z)$ and $M$. \cite{wang2018high} consider settings in which the rank of $L$ can be large. Notably, \cite{chiang2015matrix} propose the so-called “dirty” IMC model, which augments the standard IMC formulation with an additional low-rank term $R$ and estimates $(L, R)$ by solving
$$
\min_{L ,R } \norm{\calP_{\Omega} (XLZ^\top + R - Y)}_F^2 + \lambda_1 \norm{L}_{*} + \lambda_2 \norm{R}_{*}.
$$
However, this model does not include a noise term $E$ and components explained only by one-sided characteristics. In addition, it does not allow a nonlinear relationship between $(X, Z)$ and $M$, and it still requires $L$ to be low-rank. Overall, each extension addresses only part of the limitations and still leaves other issues unresolved.

A closely related study is \cite{ledent2021orthogonal}, which proposes orthogonal inductive matrix completion (OIMC). OIMC decomposes the target matrix into mutually orthogonal components associated with the row and column side-information spaces and their orthogonal complements and applies component-specific nuclear-norm penalties. Its four-block specification is therefore closely related to our decomposition into \(M_1\), \(M_2\), \(M_3\), and \(M_4\) below. The two approaches, however, differ in several important respects. OIMC imposes orthogonality by construction to obtain a unique and interpretable componentwise decomposition, whereas our primary estimand is the aggregate matrix \(M\), which can be estimated without imposing orthogonality or separately identifying its four components. Moreover, our framework uses sieve approximations to accommodate potentially nonlinear effects of row and column characteristics and develops estimation procedures for fully observed, MAR (missing at random) allowing the observation probability to depend on side information, and blockwise MNAR (missing not at random) settings. The theoretical analyses also differ: OIMC derives out-of-sample generalization bounds under bounded Lipschitz losses, whereas we establish explicit Frobenius- and entrywise-norm estimation-error bounds. Although OIMC's excess risk under squared loss coincides with the normalized squared Frobenius error under uniform joint sampling, the two sets of bounds are not directly comparable because they are derived under different sampling, noise, and structural assumptions.

Another notable line of research on matrix completion with covariates includes \cite{mao2019matrix} and \cite{ma2025statistical}. These papers consider the model
$$
Y = X B^\top + R + E
$$
where $B$ is an unknown coefficient matrix and $R$ is a low-rank matrix. Because this approach does not incorporate column characteristics $Z$, it cannot capture interaction terms involving both $X$ and $Z$ (such as $XLZ^\top$) or components explained solely by $Z$. In addition, it does not allow for nonlinear effects. As a result, the model still has some limitations. 

Lastly, a related strand of work studies PCA or factor analysis in settings without missing data \citep[see, e.g.,][]{fan2016projected,chiang2016robust,niranjan2017provable,zhu2016personalized,xue2017side}. For example, \cite{fan2016projected} consider the model, 
$Y = (G(X) + \Gamma) F^\top + E = G(X) F^\top + \Gamma F^\top + E$, where $G(X)$ is a part of loading defined by an unknown function of $X$, and $F$ denotes unobserved factors. In contrast, \cite{zhu2016personalized} study the model, $Y = XB^\top +AZ^\top + R + E$, where $R$ is low-rank. The former framework cannot capture interaction terms involving both $X$ and $Z$ (such as $XLZ^\top$) or components explained solely by $Z$, whereas the latter does not include an interaction term between $X$ and $Z$. Moreover, this model primarily emphasizes linear relationships.

To overcome the limitations of existing approaches, we consider the following model:
\begin{gather}\label{eq:model}
Y = M + E, \qquad M = M_1 + M_2 + M_3 + M_4,\\
\nonumber M_1 = G_1(X) Q_1(Z)^\top , \quad M_2 = G_2(X) V_1^\top, \quad M_3 = W_1 Q_2(Z)^\top , \quad M_4 = W_2 V_2^\top,
\end{gather}
where $G_1, G_2, Q_1,$ and $Q_2$ are unknown matrix-valued functions, and $W_1, W_2, V_1,$ and $V_2$ are unobserved matrices. This model is more general and nests the above models. For example, the models in \cite{xu2013speedup,jain2013provable,wang2018high} correspond to the special case $M = M_1$, and the model in \cite{chiang2015matrix} corresponds to $M = M_1 + M_4$. In addition, the models in \cite{fan2016projected,mao2019matrix,ma2025statistical} can be viewed as special cases of $M = M_2 + M_4$. For instance, the model in \cite{fan2016projected} can be represented as $M = M_2 + M_4$ with $V_1 = V_2$. Hence, our estimation approach under this model is less likely to suffer from model misspecification. Moreover, if the data contain components that existing models do not account for, our estimator is expected to perform better than estimators based on those restricted models. As discussed in Section \ref{sec:asymp_obs}, the convergence rates of our estimator demonstrate the robustness of our method across models, and the simulation results in Section \ref{sec:simul} are consistent with these theoretical findings.

Our estimation is based on a sieve projection method. Using projection matrices constructed from sieve bases for $X$ and $Z$, we estimate each component of $M$ separately and then obtain an estimator of $M$ by summing these estimates. Thanks to the sieve projection, our method can accommodate potentially nonlinear effects of $X$ and $Z$ on $M$. In addition, estimating each component separately allows us to fully exploit the model structure in \eqref{eq:model}. Together, these features make our estimator more likely to outperform methods based on more restrictive models when the data contain components that those restrictive models do not account for.

Another important feature of our approach is the use of nuclear-norm penalization, which corresponds to a soft-thresholding procedure. Hence, if some of $M_2$, $M_3$, and $M_4$ are exactly zero, then our estimators for those components are also exactly zero with high probability. This property enhances the robustness of our estimator.

In contrast, if we use a spectral method to estimate each component, we must estimate the rank of each part, and existing rank estimators may produce incorrect (nonzero) estimates when the corresponding component is weak due to a low signal-to-noise ratio. As a result, spectral methods may perform poorly when some of $M_2$, $M_3$, and $M_4$ are zero or close to zero. By comparison, our nuclear-norm–penalized estimator does not require estimating the rank of each component or the signal strength of each component; therefore, small values of $M_2$, $M_3$, and $M_4$ do not pose a problem.

Another important contribution of this paper is that we also consider a setting with missing entries, where the missingness is not at random. While many papers use side information for imputation under MAR, to the best of our knowledge, no existing matrix completion work incorporates side information under blockwise MNAR. Since the seminal work of \cite{athey2021matrix}, which demonstrated that matrix completion techniques can be very useful for causal panel data models, matrix completion has become a popular tool for estimating unobserved potential outcomes under the untreated (control) condition. However, the potential-outcome matrix under the untreated condition usually exhibits a missingness pattern that does not follow random missingness. Consequently, matrix completion under blockwise MNAR—and its applications to causal inference—has been actively studied recently (see, e.g., \cite{athey2021matrix,bai2021matrix,agarwal2023causal,choi2024matrix,yan2024entrywise}). We propose a novel matrix completion method that leverages side information under MNAR with blockwise missingness. As shown in our real-data experiment in Section \ref{sec:real_data}, our method outperforms standard matrix-completion approaches in imputing unobserved potential outcomes and demonstrates its usefulness for treatment-effect estimation.

The remainder of this paper is organized as follows. Section \ref{sec:model_est} introduces our model and our estimation method, which uses sieve projection with nuclear norm penalization. Section \ref{sec:asymp_obs} presents asymptotic error bounds for the estimator and discusses the robustness of our method across different models. Section \ref{sec:ext_missing} extends our estimation strategy to the case in which the outcome matrix is partially observed. Importantly, we consider the MNAR setting with blockwise missingness as well as the MAR setting. Section \ref{sec:simul} presents numerical studies using simulated and real data to demonstrate the advantages of our method. The appendix to the main paper presents several extensions of the main framework, while the online appendix contains all proofs.

\section{Model and Estimation}\label{sec:model_est}

In this paper, we consider the following panel model:
$$
Y = M + E, \qquad M = M_1 + M_2 + M_3 + M_4,
$$
where \(Y = (y_{it})_{i\le N,\, t\le T}\) is the outcome matrix, \(E = (\epsilon_{it})_{i\le N,\, t\le T}\) is the noise matrix, and \(M = (m_{it})_{i\le N,\, t\le T}\) is the matrix of interest. We decompose \(M\) into four parts: (i) \(M_1\), a component well explained by both \(X\) and \(Z\); (ii) \(M_2\), a component explained by \(X\) but irrelevant to \(Z\); (iii) \(M_3\), a component explained by \(Z\) but irrelevant to \(X\); and (iv) \(M_4\), a component irrelevant to both \(X\) and \(Z\), where \(X = (x_i)_{i\le N}\) and \(Z = (z_t)_{t\le T}\) are observable characteristics corresponding to the row and column indices, respectively.

More specifically, each part can be represented as
\begin{gather}\label{eq:decomposition}
M_1 = G_1(X) Q_1(Z)^\top,\quad
M_2 = G_2(X) V_1^\top,\quad
M_3 = W_1 Q_2(Z)^\top,\quad
M_4 = W_2 V_2^\top,
\end{gather}
where \(G_1(X) = (g_{1,k}(x_i))_{i\le N,\, k\le K_1}\), \(G_2(X) = (g_{2,k}(x_i))_{i\le N,\, k\le K_2}\), \(Q_1(Z) = (q_{1,k}(z_t))_{t\le T,\, k\le K_1}\), and \(Q_2(Z) = (q_{2,k}(z_t))_{t\le T,\, k\le K_3}\) for some unknown functions \(g_{1,k}(\cdot)\), \(g_{2,k}(\cdot)\), \(q_{1,k}(\cdot)\), and \(q_{2,k}(\cdot)\). Here, \(W_1 = (w_{1,ik})_{i\le N,\, k\le K_3}\) and \(W_2 = (w_{2,ik})_{i\le N,\, k\le K_4}\) capture the components not explained by \(X\), while \(V_1 = (v_{1,tk})_{t\le T,\, k\le K_2}\) and \(V_2 = (v_{2,tk})_{t\le T,\, k\le K_4}\) capture the components not explained by \(Z\). This model is general and encompasses many existing models.

\paragraph{Estimation.}

To properly accommodate and exploit the structure of our model in \eqref{eq:decomposition}, we propose estimating \(M\) using a sieve projection method. For two sets of basis functions \(\{\phi_1(x), \ldots, \phi_{J}(x)\}\) and \(\{\psi_1(z), \ldots, \psi_{J}(z)\}\) (e.g., B-splines, Fourier series, wavelets, or polynomial series), define
\begin{align*}
\boldsymbol{\phi}(x_i)
&= \big[ \phi_1(x_{i1}), \ldots, \phi_{J}(x_{i1}), \ldots, \phi_1(x_{id_1}), \ldots, \phi_{J}(x_{id_1}) \big]^\top \in \mathbb{R}^{J d_1}, \\
\boldsymbol{\psi}(z_t)
&= \big[ \psi_1(z_{t1}), \ldots, \psi_{J}(z_{t1}), \ldots, \psi_1(z_{td_2}), \ldots, \psi_{J}(z_{td_2}) \big]^\top \in \mathbb{R}^{J d_2},
\end{align*}
where \(d_1\) and \(d_2\) are the dimensions of \(x_i\) and \(z_t\), respectively.
The corresponding projection matrices are
\[
P_X = \Phi(X)\big(\Phi(X)^{\top}\Phi(X)\big)^{-1}\Phi(X)^{\top},
\qquad
P_Z = \Psi(Z)\big(\Psi(Z)^{\top}\Psi(Z)\big)^{-1}\Psi(Z)^{\top},
\]
where \(\Phi(X) = \big[ \boldsymbol{\phi}(x_1), \ldots, \boldsymbol{\phi}(x_N) \big]^\top\) and
\(\Psi(Z) = \big[ \boldsymbol{\psi}(z_1), \ldots, \boldsymbol{\psi}(z_T) \big]^\top\). In addition, we define $P_X^\perp = I_N - P_X$ and $P_Z^\perp = I_T - P_Z$. Note that as long as \(g_{1,k}(\cdot)\), \(g_{2,k}(\cdot)\), \(q_{1,k}(\cdot)\), and \(q_{2,k}(\cdot)\) are sufficiently smooth, for any \(\iota \in (1,2)\) we have
\begin{gather}\label{eq:projection}
P_X G_\iota(X) \approx G_\iota(X), \qquad P_Z Q_\iota(Z) \approx Q_\iota(Z).
\end{gather}
Moreover, \(\|P_X E P_Z\|_F\) can be much smaller than \(\|E\|_F\) due to the orthogonality between \((X,Z)\) and \(E\). Leveraging this property, we propose estimating \(M\) as follows.

\begin{breakablealgorithm}
\caption{Estimation procedure}
\label{alg:estimation_obs}
\begin{algorithmic}
\noindent \textbf{Step 1:} Compute $\widehat{M}_1 = P_X Y P_Z$. \\

\textbf{Step 2:} Compute the following nuclear-norm-penalized estimators:
\begin{gather*}
\widehat{M}_2 \coloneqq \arg\min_{A \in \bbR^{N \times T}}
\big\| P_X Y P_Z^\perp- A \big\|_F^2 + \nu_2 \|A\|_*, \\
\widehat{M}_3 \coloneqq \arg\min_{A \in \bbR^{N \times T}}
\big\| P_X^\perp Y P_Z - A \big\|_F^2 + \nu_3 \|A\|_*, \\
\widehat{M}_4 \coloneqq \arg\min_{A \in \bbR^{N \times T}}
\big\| P_X^\perp Y P_Z^\perp- A \big\|_F^2 + \nu_4 \|A\|_*,
\end{gather*}
where $\nu_2 = C_2 \sqrt{T}$, $\nu_3 = C_3 \sqrt{N}$, and $\nu_4 = C_4 \sqrt{N+T}$ for some sufficiently large constants $C_2, C_3, C_4 > 0$. \\

\textbf{Step 3:} Form the final estimator $\widehat{M} = \widehat{M}_1 + \widehat{M}_2 + \widehat{M}_3 + \widehat{M}_4$.
\end{algorithmic}
\end{breakablealgorithm}

To understand how this estimator works, note that by \eqref{eq:projection} and basic properties of nuclear-norm penalization, we have
\begin{gather*}
 \widehat{M}_1 \approx G_1 Q_1^\top + G_2 V_1^\top P_Z + P_X W_1 Q_2^\top + P_X W_2 V_2^\top P_Z, \\
 \widehat{M}_2 \approx G_2 V_1^\top P_Z^\perp+ P_X W_2 V_2^\top P_Z^\perp, \quad
 \widehat{M}_3 \approx P_X^\perp W_1 Q_2^\top + P_X^\perp W_2 V_2^\top P_Z, \quad
 \widehat{M}_4 \approx P_X^\perp W_2 V_2^\top P_Z^\perp,
\end{gather*}
under suitable conditions on the noise and on the smoothness of \(g_{1,k}(\cdot)\), \(g_{2,k}(\cdot)\), \(q_{1,k}(\cdot)\), and \(q_{2,k}(\cdot)\). Importantly, the terms involving \(P_X W_\iota\) or \(P_Z V_\iota\) cancel out when we sum the four estimators. Hence, without imposing any orthogonality conditions between \(X\) and \(W\) (or between \(Z\) and \(V\)), our final estimator \(\widehat{M}\) can estimate \(M\) well.

In addition, because nuclear-norm penalization acts as a thresholding estimator, when some (or all) of \(M_2\), \(M_3\), and \(M_4\) are zero or sufficiently small, it helps us obtain a tighter bound.

\section{Asymptotic Results}\label{sec:asymp_obs}

In this section, we present the convergence rate of our estimator. We begin by imposing the following conditions.

\begin{assumption}[Noise]\label{asp:noise}
The random variables \((\epsilon_{it})_{i \leq N,\, t \leq T}\) are independent, mean-zero, sub-Gaussian, and satisfy
\(c < \mathbb{E}[\epsilon_{it}^2] \leq \sigma^2 \leq C_1\) and \(\mathbb{E}[\exp(s\epsilon_{it})] \leq \exp(C_2 s^2 \sigma^2)\) for all \(s \in \mathbb{R}\), for some constants \( C_1, C_2 > 0\). In addition, \((\epsilon_{it})_{i \leq N,\, t \leq T}\) are independent of \(X\) and \(Z\).
\end{assumption}

The independence and sub-Gaussianity assumptions are used to derive tight bounds for
\(\|P_X E P_Z\|\), \(\|P_X E\|\), and \(\|E P_Z\|\).
We can generalize this condition to weakly dependent noise with suitable moment conditions, at the cost of additional \(J\)-dependent terms in the bound in Theorem \ref{thm:convergence_rate_obs}.

\begin{assumption}[Basis functions]\label{asp:basis}
(i) There exist constants \(c, C > 0\) such that, with probability approaching one,
\begin{gather*}
c < \lambda_{\min}\!\big(N^{-1}\Phi(X)^\top \Phi(X)\big)
\le \lambda_{\max}\!\big(N^{-1}\Phi(X)^\top \Phi(X)\big) < C,\\
c < \lambda_{\min}\!\big(T^{-1}\Psi(Z)^\top \Psi(Z)\big)
\le \lambda_{\max}\!\big(T^{-1}\Psi(Z)^\top \Psi(Z)\big) < C,
\end{gather*}
where \(\lambda_{\max}(A)\) and \(\lambda_{\min}(A)\) denote the largest and smallest singular values of \(A\), respectively.
(ii) \(\max_{j \leq J,\, i \leq N,\, l \leq d_1} \mathbb{E}[\phi_{j}(x_{il})^4] < \infty\) and
\(\max_{j \leq J,\, t \leq T,\, l \leq d_2} \mathbb{E}[\psi_{j}(z_{tl})^4] < \infty\).
\end{assumption}

This condition is standard in the sieve-estimation literature (e.g., \cite{fan2016projected,chen2023semiparametric}). Because we focus on the case where \(J d_1 \ll N\) and \(J d_2 \ll T\), it follows from the law of large numbers and therefore is not overly restrictive.

\begin{assumption}[Sieve approximation]\label{asp:sieve}
(i) There exist constants \(\gamma_1^G,\gamma_2^G,\gamma_1^Q,\gamma_2^Q \geq 2\) such that, for some sieve coefficient vectors \(b_{1,k}, b_{2,k} \in \bbR^{J d_1}\) and \(a_{1,k}, a_{2,k} \in \bbR^{J d_2}\), the sieve approximations satisfy
\[
\sup_{x \in \calX} \big| g_{\iota,k}(x) - b_{\iota,k}^\top \boldsymbol{\phi}(x) \big|
= O\!\left(J^{-\gamma_\iota^G/d_1}\right), 
\qquad
\sup_{z \in \calZ} \big| q_{\iota,k}(z) - a_{\iota,k}^\top \boldsymbol{\psi}(z) \big|
= O\!\left(J^{-\gamma_\iota^Q/d_2}\right),
\]
where \(\calX\) and \(\calZ\) are the supports of \(x_i\) and \(z_t\), respectively. \\
(ii) The sieve dimension \(J\) satisfies
\begin{gather*}
\sqrt{T}\,K_1 / J^{(\gamma_1^G/d_1)} \rightarrow 0, \qquad 
\sqrt{N}\,K_1 / J^{(\gamma_1^Q/d_2)} \rightarrow 0, \qquad 
\max\{\sqrt{N},\sqrt{T}\}\, K_2 / J^{(\gamma_2^G/d_1)} \rightarrow 0, \\
\max\{\sqrt{N},\sqrt{T}\}\, K_3 / J^{(\gamma_2^Q/d_2)} \rightarrow 0 , \qquad
J d_1 \ll N, \qquad
J d_2 \ll T .
\end{gather*}
\end{assumption}

Condition (i) is a standard assumption in sieve estimation. For example, if \(g_{\iota,k}(\cdot)\) has an additive form $g_{\iota,k}(x_i) = \sum_{l=1}^{d_1} g_{\iota,kl}(x_{il})$
and each \(g_{\iota,kl}(\cdot)\) belongs to the H\"older class \(\calH(\rho_\iota^G, \tau_\iota^G)\), where
\[
\calH(\rho,\tau)
= \left\{ h : \big| h^{(\rho)}(s) - h^{(\rho)}(t) \big| \leq C |s-t|^\tau \right\}
\]
for some \(C>0\), then \(\gamma_\iota^G = \rho_\iota^G + \tau_\iota^G\) for typical choices of basis functions (see, e.g., \cite{chen2007large}). On the other hand, condition (ii) requires sufficient smoothness of the functions \(g_{\iota,k}(\cdot)\) and \(q_{\iota,k}(\cdot)\). Note that \(\gamma_\iota^G\) and \(\gamma_\iota^Q\) can be viewed as smoothness parameters for \(g_{\iota,k}(\cdot)\) and \(q_{\iota,k}(\cdot)\), respectively. Therefore, if \(g_{\iota,k}(\cdot)\) and \(q_{\iota,k}(\cdot)\) are sufficiently smooth, then \(\gamma_\iota^G\) and \(\gamma_\iota^Q\) will be large, and condition (ii) can be satisfied even when \(J\) increases slowly.

Lastly, we impose the following moment conditions.

\begin{assumption}[Moments]\label{asp:moment}
(i) For all \(i\) and \(t\), \(\mathbb{E}[m_{it}^4]\) is bounded. \\
(ii) There exists a constant \(C_1>0\) such that for all \(i,t,k\),
\[
\mathbb{E}\!\big[g_{1,k}^2(x_i)\big],\ \ 
\mathbb{E}\!\big[g_{2,k}^2(x_i)\big],\ \ 
\mathbb{E}\!\big[q_{1,k}^2(z_t)\big],\ \ 
\mathbb{E}\!\big[q_{2,k}^2(z_t)\big]
\le C_1.
\]
(iii) There exists a constant \(C_2>0\) such that for all \(i,t,k\),
\[
\mathbb{E}[w_{1,ik}^2],\ \ 
\mathbb{E}[w_{2,ik}^2],\ \ 
\mathbb{E}[v_{1,tk}^2],\ \ 
\mathbb{E}[v_{2,tk}^2]
\le C_2.
\]
\end{assumption}

We are now in a position to state the statistical properties of our estimators. The following theorem provides the convergence rate of the proposed estimator. Here, \(r_1\), \(r_2\), \(r_3\), and \(r_4\) denote the ranks of \(M_1\), \(M_2\), \(M_3\), and \(M_4\), respectively.

\begin{theorem}[Convergence rate]\label{thm:convergence_rate_obs}
Suppose that Assumptions \ref{asp:noise}--\ref{asp:moment} hold. Then,
\begin{align*}
\|\widehat{M} - M\|_F
&= O_p\Bigg(
J\sqrt{d_1d_2}
+ \sqrt{r_2 + r_4}\,\min \Big\{ \sqrt{T},\, \norm{M_2 P_Z^\perp +  P_X M_4 P_Z^\perp}_F \Big\} \\
&\qquad\quad
+ \sqrt{r_3 + r_4}\,\min \Big\{ \sqrt{N},\,  \norm{ P_X^\perp M_3  +  P_X^\perp M_4 P_Z }_F \Big\}\\
&\qquad\quad + \sqrt{r_4}\,\min \Big\{ \sqrt{N + T},\, \norm{ P_X^\perp M_4 P_Z^\perp }_F \Big\} \\
&\qquad\quad
+ \sqrt{NT}\left[ \frac{K_1}{J^{\gamma_1^G/d_1}} + \frac{K_1}{J^{\gamma_1^Q/d_2}} + \frac{K_2}{J^{\gamma_2^G/d_1}} + \frac{K_3}{J^{\gamma_2^Q/d_2}} \right]
\Bigg).
\end{align*}
\end{theorem}

Some immediate remarks are in order. First, note that the dominating part of the error bound for our estimator does not depend on \(r_1\). Thus, we can allow \(r_1\) to be large as long as \(g_{1,k}(\cdot)\) and \(q_{1,k}(\cdot)\) are sufficiently smooth. The last term,
\[
\sqrt{NT}\left[ \frac{K_1}{J^{\gamma_1^G/d_1}} + \frac{K_1}{J^{\gamma_1^Q/d_2}} + \frac{K_2}{J^{\gamma_2^G/d_1}} + \frac{K_3}{J^{\gamma_2^Q/d_2}} \right],
\]
arises from the sieve approximation (smoothing) error. When the functions \(g_{\iota,k}(\cdot)\) and \(q_{\iota,k}(\cdot)\) are sufficiently smooth (i.e., \(\gamma_\iota^G\) and \(\gamma_\iota^Q\) are large), this term is small and dominated by the other terms.

Theorem \ref{thm:convergence_rate_obs} illustrates the robustness of our estimator. First, consider the most favorable case, \(M=M_1\), where \(M\) is well explained by \(X\) and \(Z\). In this case, the convergence rate of our estimator is \(O_p(J\sqrt{d_1d_2})\) provided that \(g_{\iota,k}(\cdot)\) and \(q_{\iota,k}(\cdot)\) are sufficiently smooth. This matches the rate of the ``double projection'' estimator \(P_X Y P_Z\), which is the most suitable estimator when we know \emph{a priori} that \(M=M_1\). By contrast, if we estimate \(M\) using a standard low-rank method such as nuclear-norm penalization in this setting, the convergence rate would be \(O_p(\sqrt{r_1(N+T)})\), which is much larger than ours. Moreover, even when \(M\) contains an additional component \(M_2+M_3+M_4\) beyond \(M_1\), as long as this component is small (in the sense that \(\|M_2+M_3+M_4\|_F \ll \sqrt{N+T}\)), our rate \(O_p\big(J\sqrt{d_1d_2} + \sqrt{r^*}\,\|M_2+M_3+M_4\|_F\big)\) is smaller than that of the usual low-rank estimator \(O_p\big(\sqrt{r(N+T)}\big)\), where \(r^*=\max\{r_2,r_3,r_4\}\) and \(r\) is the rank of \(M\).

Next, consider the least favorable case, \(M=M_4\), where \(M\) is unrelated to \(X\) and \(Z\) and the side information is uninformative. In this case, the convergence rate of our estimator is \(O_p(\sqrt{r_4(N+T)})\), provided that \(g_{\iota,k}(\cdot)\) and \(q_{\iota,k}(\cdot)\) are sufficiently smooth and \(J\sqrt{d_1d_2} \ll \sqrt{r_4(N+T)}\). Note that this rate coincides with that of standard low-rank estimators. Hence, even in the least favorable case, the error bound for our estimator is comparable to that of a typical low-rank method. In contrast, the ``double projection'' estimator \(P_X Y P_Z\) is inconsistent in this case. Moreover, if \(M\) contains an additional small component \(M_1+M_2+M_3\) beyond \(M_4\), the convergence rate of our estimator becomes \(O_p(\sqrt{r^*(N+T)})\), whereas that of a typical low-rank estimator remains \(O_p(\sqrt{r(N+T)})\). Thus, when \(r_1\) is large, our method can yield a tighter bound.

Lastly, consider the case \(M=M_2\). In this case, the convergence rate of our estimator is \(O_p(\sqrt{r_2 T})\), provided that \(g_{\iota,k}(\cdot)\) and \(q_{\iota,k}(\cdot)\) are sufficiently smooth and \(J\sqrt{d_1d_2} \ll \sqrt{r_2 T}\). By comparison, the convergence rate of a typical low-rank estimator is \(O_p(\sqrt{r_2(N+T)})\). Hence, when \(N \gg T\), our method yields a tighter bound. In addition, if \(M\) contains an additional small component \(M_1+M_3+M_4\) beyond \(M_2\), the convergence rate of our estimator becomes \(O_p(\sqrt{(r_2+r_4)T})\), whereas that of a typical low-rank estimator is \(O_p(\sqrt{r(N+T)})\). Thus, when \(r_1\) is large or \(N \gg T\), our method can yield a better bound. A similar discussion applies to the case \(M=M_3\).

Table \ref{tab:comparison_obs} summarizes the convergence rates of the estimators in the cases discussed above. We can see that, in every case, the convergence rate of our estimator is at least as good as that of the other estimators, provided that \(g_{\iota,k}(\cdot)\) and \(q_{\iota,k}(\cdot)\) are sufficiently smooth.

\begin{table}[htbp]
  \centering
  \begin{tabular}{c|cccc}
    \hline\hline
          & \(M = M_1\) & \(M = M_2\) & \(M = M_3\) & \(M = M_4\) \\[3pt]
    \hline
    Our estimator & \(J\sqrt{d_1d_2}\) & \(\sqrt{r_2 T}\) & \(\sqrt{r_3 N}\) & \(\sqrt{r_4 (N+T)}\) \\[3pt]
    Double projection & \(J\sqrt{d_1d_2}\) & \(\sqrt{NT}\) & \(\sqrt{NT}\) & \(\sqrt{NT}\) \\[3pt]
    Low-rank estimation & \(\sqrt{r_1(N+T)}\) & \(\sqrt{r_2 (N+T)}\) & \(\sqrt{r_3 (N+T)}\) & \(\sqrt{r_4 (N+T)}\) \\[3pt]
    \hline
  \end{tabular}
  \caption{Convergence rates of matrix estimators}
  \label{tab:comparison_obs}
\end{table}

Lastly, as a corollary, we present convergence rates for the estimated singular vectors, since the factors and loadings are often of interest to researchers (see, e.g., \cite{bai2008large}). Let \(\widehat{U} \in \bbR^{N \times r}\) and \(\widehat{V} \in \bbR^{T \times r}\) be the left and right singular vectors of \(\widehat{M}\), respectively. Similarly, let \(U\) and \(V\) denote the left and right singular vectors of \(M\), respectively. For notational convenience, denote the upper bound on \(\|\widehat{M}-M\|_F\) in Theorem \ref{thm:convergence_rate_obs} by
\begin{align*}
\calR
&= 
J\sqrt{d_1d_2}
+ \sqrt{r_2 + r_4}\,\min \Big\{ \sqrt{T},\, \norm{M_2 P_Z^\perp +  P_X M_4 P_Z^\perp}_F \Big\} \\
&\quad
+ \sqrt{r_3 + r_4}\,\min \Big\{ \sqrt{N},\,  \norm{ P_X^\perp M_3  +  P_X^\perp M_4 P_Z }_F \Big\} + \sqrt{r_4}\,\min \Big\{ \sqrt{N + T},\, \norm{ P_X^\perp M_4 P_Z^\perp }_F \Big\} \\
&\quad
+ \sqrt{NT}\left[ \frac{K_1}{J^{\gamma_1^G/d_1}} + \frac{K_1}{J^{\gamma_1^Q/d_2}} + \frac{K_2}{J^{\gamma_2^G/d_1}} + \frac{K_3}{J^{\gamma_2^Q/d_2}} \right] .
\end{align*}
Then, we obtain the following convergence rates.

\begin{corollary}\label{cor:singular_vector}
Suppose that Assumptions \ref{asp:noise}--\ref{asp:moment} hold. In addition, assume that \(\calR/\lambda_{\min} \conP 0\), where \(\lambda_{\min}\) denotes the smallest nonzero singular value of \(M\). Then,
\[
\max \biggl\{
\min_{R \in \calO_{r \times r}} \|\widehat{U} - R U\|_F,\ \ 
\min_{R \in \calO_{r \times r}} \|\widehat{V} - R V\|_F
\biggr\}
= O_p\!\left( \frac{\calR}{\lambda_{\min}}\right).
\]
\end{corollary}

Additionally, Section~\ref{sec:entrywise} establishes an entrywise convergence
rate for our estimator.

\begin{remark}\label{rem:comparison}

In this remark, we compare our results with the existing generalization bounds after expressing their rank and dimension dependence on a common scale. In standard matrix completion, the classical trace- and max-norm literature establishes rank-adaptive generalization bounds under uniform and nonuniform sampling
(\cite{foygel2011concentration,foygel2011learning,foygel2012matrix,
shamir2011collaborative,shamir2014matrix,srebro2005rank}).
Suppressing logarithmic and error-tolerance factors, the relevant sample-complexity dependence for an \(n\times n\) rank-\(r\) matrix is of order \(nr\) under uniform sampling. Under arbitrary sampling, the ordinary trace-norm bound is generally of order \(n^{3/2}\sqrt r\), whereas suitably weighted or adjusted regularizers recover an \(nr\)-type dependence. Recent results for Schatten-\(p\) regularization provide analogous conclusions and reduce to the familiar \(n^{3/2}\sqrt r\) dependence when \(p=1\)
(\cite{ledent2024generalization}).

The earlier generalization bounds for inductive matrix completion have sample complexity of order \(rd^2\), where \(d\) is the dimension of the side information
(\cite{chiang2015matrix,chiang2018using}). More recent analysis improves this dependence to $\widetilde O\!\left(d^{3/2}\sqrt r\right)$ for the standard IMC regularizer under arbitrary sampling and to $\widetilde O(dr)$ under uniform sampling or with an adjusted trace-norm regularizer (\cite{ledent2021fine}). Orthogonal inductive matrix completion further considers a sum of orthogonal IMC components and is therefore closely related to our four-component decomposition (\cite{ledent2021orthogonal}).

To relate these results to our framework, conditional on the constructed sieve bases, \(M_2\) can be viewed as a one-sided IMC component with feature dimensions \(J d_1\) and \(T\), where the column-side feature matrix is the identity matrix. Similarly, \(M_3\) has effective feature dimensions \(N\) and \(J d_2\), while \(M_4\) corresponds to standard matrix completion with identity features on both sides. Their effective degrees of freedom are therefore of orders
$$
r_2(J d_1+T),\qquad
r_3(N+J d_2),\qquad
r_4(N+T),
$$
respectively. Under our maintained conditions \(J d_1\ll T\) and \(J d_2\ll N\), these simplify to
$$
r_2T,\qquad r_3N,\qquad r_4(N+T).
$$
Correspondingly, suppressing the common sampling, noise, and sieve-approximation terms, our revised Frobenius-error bounds have leading terms
$$
O_p\!\left(\sqrt{r_2T}\right),\qquad
O_p\!\left(\sqrt{r_3N}\right),\qquad
O_p\!\left(\sqrt{r_4(N+T)}\right).
$$
Thus, their squared rank–dimension dependence is of the \(dr\) type, consistent with the modern rank-adaptive bounds under uniform sampling, rather than of the earlier \(rd^2\) type.

Under uniform sampling and squared loss, the excess population risk satisfies
$$
\mathcal R(\widehat M)-\mathcal R(M)
=
\frac{1}{NT}\|\widehat M-M\|_F^2,
$$
so our Frobenius bounds can be translated into per-entry prediction-error bounds. Nevertheless, the comparison remains qualitative because the cited literature primarily studies bounded-loss generalization under i.i.d. sampling, whereas we derive explicit Frobenius- and entrywise-norm estimation bounds under additive noise and consider fully observed, MAR, and blockwise MNAR settings. In particular, the blockwise MNAR setting is not covered by these i.i.d. generalization results.

Finally, \cite{ledent2026generalization} studies a different semi-supervised setting in which the sampling distribution itself provides side information through unlabeled observations. Its assumptions and two-stage error bound are therefore not directly comparable with our setting based on observed row and column characteristics.

\end{remark}

\section{Extension to Missing Case}\label{sec:ext_missing}

Next, we extend our estimation strategy to the case where the outcome matrix is only partially observed. The base model is the same as in Section \ref{sec:model_est}, and we additionally assume that researchers observe \(\Omega \circ Y\) instead of \(Y\), where \(\Omega = (\omega_{it})_{i \leq N,\, t \leq T} \in \{0,1\}^{N \times T}\).

\subsection{Missing At Random Case}

In this section, we consider the case where outcome entries are missing at random. Specifically, we assume that \((\omega_{it})_{i \leq N,\, t \leq T}\) are i.i.d.\ Bernoulli random variables with mean \(p\), as is common in the matrix completion literature. In addition, we consider the case in which the observation probability is heterogeneous and depends on the side information in Appendix \ref{sec:mar_hetero}.

\paragraph{Estimation.}

Similarly to the fully observed case above, we use the projection method to exploit the structure of the model in \eqref{eq:decomposition}. However, when entries are missing, a key difficulty is that we cannot directly observe \(\Omega \circ \big(P_X Y P_Z^\perp\big)\) and \(\Omega \circ \big(P_X^\perp Y P_Z\big)\) when we aim to estimate \(P_X M P_Z^\perp\) or \(P_X^\perp M P_Z\) via nuclear-norm penalization. On the other hand, we can still estimate \(P_X M P_Z\) accurately using the projection estimator \(p^{-1} P_X (\Omega \circ Y) P_Z\).

Hence, in the presence of missing entries, we proceed as follows.

\begin{breakablealgorithm}
		\caption{Estimation procedure for MAR case}
		\label{alg:estimation_mar}
		\begin{algorithmic}
    	\noindent \textbf{Step 1:} Derive $\widehat{M}_1 = \widehat{p}^{-1} P_X (\Omega \circ Y) P_Z$ where $\widehat{p} = \frac{1}{NT} \sum_{it} \omega_{it}$. \\
			\textbf{Step 2:} Apply the nuclear norm penalization to $\Omega \circ (Y - \widehat{M}_1)$:
   \begin{gather*}
\widehat{M}_{rest} \coloneqq \arg\min_{A: \norm{A}_{\infty} \leq M_{\max}} \norm{\Omega \circ (Y - \widehat{M}_1 - A)}_F^2 + \nu \norm{A}_*, 
   \end{gather*}
   where $M_{\max} >0 $ is some large constant and $\nu = C p^{1/2} \sqrt{N+T}$ with a constant $C > 0$.\\
	\textbf{Step 3:} Get the final estimator, $\widehat{M} = \widehat{M}_1 + \widehat{M}_{rest}$.
		\end{algorithmic}
	\end{breakablealgorithm}

Note that, by the projection relation \eqref{eq:projection} and the usual properties of nuclear-norm penalization in matrix completion, we have
\begin{align*}
 \widehat{M}_1 &\approx G_1 Q_1^\top + G_2 V_1^\top P_Z + P_X W_1 Q_2^\top + P_X W_2 V_2^\top P_Z, \\
 \widehat{M}_{\mathrm{rest}} &\approx G_2 V_1^\top P_Z^\perp+ P_X^\perp W_1 Q_2^\top + W_2 V_2^\top - P_X W_2 V_2^\top P_Z,
\end{align*}
under conditions on the noise and on the smoothness of \(g_{1,k}(\cdot)\), \(g_{2,k}(\cdot)\), \(q_{1,k}(\cdot)\), and \(q_{2,k}(\cdot)\) similar to those in the previous section. As above, the terms involving \(P_X W_\iota\) or \(P_Z V_\iota\) cancel out when we add the two estimators. Hence, our final estimator \(\widehat{M}\) can estimate \(M\) well. In particular, because we estimate the \(M_1\) component using the projection method rather than a low-rank estimator, our approach can have advantages when \(r_1\) is large or when \(M_1\) is large relative to \(M_2\), \(M_3\), and \(M_4\).

\paragraph{Asymptotic result.}

We now present the convergence rate of our estimator. We begin by introducing several additional assumptions.

\begin{assumption}[Random missing]\label{asp:mar}
The random variables \((\omega_{it})_{i \leq N,\, t \leq T}\) are i.i.d.\ Bernoulli with \(\mathbb{E}[\omega_{it}] = p\). In addition, \(\Omega\) is independent of \(E\), \(X\), \(Z\), and \(M\).
\end{assumption}

In addition to Assumption \ref{asp:mar}, we require a slightly different condition on the sieve approximation error than in the fully observed case.

\begin{assumption}[Sieve approximation]\label{asp:sieve_mar}
(i) Assumption \ref{asp:sieve} (i) holds. \\
(ii) The sieve approximation satisfies
\begin{gather*}
\frac{\min\{\sqrt{N},\sqrt{T}\}}{\sqrt{p}}
\left(
\frac{K_1}{J^{\gamma_1^G/d_1}} + \frac{K_1}{J^{\gamma_1^Q/d_2}} + \frac{K_2}{J^{\gamma_2^G/d_1}} + \frac{K_3}{J^{\gamma_2^Q/d_2}}
\right) \rightarrow 0.
\end{gather*}
\end{assumption}

The following theorem provides the convergence rate of our estimator.

\begin{theorem}[Convergence rate for the MAR case]\label{thm:convergence_rate_mar}
Suppose that Assumptions \ref{asp:noise}, \ref{asp:basis}, \ref{asp:moment}, \ref{asp:mar}, and \ref{asp:sieve_mar} hold. In addition, assume that \(J \sqrt{d_1 d_2 } \ll p \sqrt{N+T}\). Then, if
\[
M_{\max} \geq \big\| M_2 P_Z^\perp+ P_X^\perp M_3 + M_4 - P_X M_4 P_Z \big\|_{\infty},
\]
we have
\begin{align*}
\|\widehat{M} - M\|_F
&= O_p\Bigg(
\frac{J\sqrt{d_1d_2}}{\sqrt{p}}
+ \sqrt{r^*}\,
\min \left\{
\frac{\sqrt{N + T}\left(1 + M_{\max} \right)}{\sqrt{p}},
\, \norm{M_2 P_Z^\perp+ P_X^\perp M_3 + M_4 - P_X M_4 P_Z}_F
\right\} \\
&\qquad\quad
+ \sqrt{NT}
\left[
\frac{K_1}{J^{\gamma_1^G/d_1}} + \frac{K_1}{J^{\gamma_1^Q/d_2}} + \frac{K_2}{J^{\gamma_2^G/d_1}} + \frac{K_3}{J^{\gamma_2^Q/d_2}}
\right]
\Bigg),
\end{align*}
where \(r^* = \max\{r_2, r_3, r_4\}\).
\end{theorem}

Similar discussions to the fully observed case apply here. The error bound for our estimator does not depend on \(r_1\) as long as \(g_{1,k}(\cdot)\) and \(q_{1,k}(\cdot)\) are sufficiently smooth, because we estimate \(M_1\) using the projection method rather than a low-rank estimator. Hence, we can allow \(r_1\) to be large.

In addition, the estimator enjoys a robustness property. For simplicity, assume that \(M_{\max}\) and \(\|M\|_{\infty}\) are bounded. First, consider the most favorable case, \(M=M_1\). In this case, the convergence rate of our estimator is \(O_p(J\sqrt{d_1d_2}/\sqrt{p})\) provided that \(g_{\iota,k}(\cdot)\) and \(q_{\iota,k}(\cdot)\) are sufficiently smooth. By contrast, if we estimate \(M\) using standard low-rank completion methods (e.g., nuclear-norm penalization), the convergence rate would be \(O_p(\sqrt{r_1(N+T)/p})\), which is much larger than ours. Moreover, even when \(M\) contains an additional component \(M_2+M_3+M_4\) beyond \(M_1\), as long as this component is small (i.e., \(\|M_2+M_3+M_4\|_F \ll \sqrt{(N+T)/p}\)), our rate \(O_p\!\Big( J\sqrt{d_1d_2}/\sqrt{p} + \sqrt{r^*}\,\|M_2+M_3+M_4\|_F \Big) \) is smaller than that of the usual low-rank completion methods \(O_p\!\big(\sqrt{r(N+T)/p}\big)\), where \(r\) is the rank of \(M\) and \(r^*=\max\{r_2,r_3,r_4\}\).

Next, consider the least favorable case, \(M=M_4\), where \(M\) is unrelated to \(X\) and \(Z\). In this case, the convergence rate of our estimator is \(O_p(\sqrt{r_4(N+T)/p})\), provided that \(g_{\iota,k}(\cdot)\) and \(q_{\iota,k}(\cdot)\) are sufficiently smooth. Note that this rate coincides with that of standard low-rank completion methods. Hence, even in the least favorable case, our estimator is comparable to typical low-rank completion methods. Moreover, if \(M\) contains an additional small component \(M_1+M_2+M_3\) beyond \(M_4\), the convergence rate of our estimator becomes \(O_p(\sqrt{r^*(N+T)/p})\), whereas that of a typical low-rank completion method is \(O_p(\sqrt{r(N+T)/p})\). Thus, when \(r_1\) is large, our method can yield a tighter bound. Similar discussions apply to the cases \(M=M_2\) and \(M=M_3\): our estimator attains the same rate as standard low-rank completion methods, and it can yield a better bound when \(M\) also contains an \(M_1\) component with large \(r_1\).

\subsection{Missing Not At Random Case: Blockwise Missingness}

In causal inference, a central challenge is to estimate the unobserved counterfactual outcome for a treated unit. Consider a binary treatment, and let \(y_{it}^{(1)}\) and \(y_{it}^{(0)}\) denote the potential outcomes of unit \(i\) at time \(t\) under treatment and control, respectively. The treatment effect is then defined as
$$
\tau_{it} = y_{it}^{(1)} - y_{it}^{(0)}.
$$
The fundamental difficulty is that, for each unit-time pair \((i,t)\), only one of these two potential outcomes is observed. Let \(d_{it}\in\{0,1\}\) denote the treatment indicator. The observed outcome can then be written as
$$
d_{it} y_{it}^{(1)}
+
(1-d_{it})y_{it}^{(0)}.
$$
Thus, when \(d_{it}=1\), the observed outcome is \(y_{it}^{(1)}\), while the counterfactual outcome \(y_{it}^{(0)}\) is unobserved. Recovering these missing counterfactual outcomes is therefore essential for estimating treatment effects. The proposed matrix completion method can be applied to the untreated potential-outcome matrix
$$
Y^{(0)}
=
\left(y_{it}^{(0)}\right)_{1\leq i\leq N,\,1\leq t\leq T}
$$
to recover the unobserved entries \(y_{it}^{(0)}\) corresponding to treated unit-time pairs.

However, in contrast to the basic matrix completion model discussed above, the
missingness pattern of \(Y^{(0)}\) in causal panel-data applications is usually
not generated by independent Bernoulli sampling. For instance, consider the
evaluation of a program that is introduced after time \(T_0\) for the last
\(N-N_0\) units. Then the untreated potential outcomes are observed for all
control units and for treated units only before treatment. Thus,
\(y_{it}^{(0)}\) is unobserved whenever \(i>N_0\) and \(t>T_0\), or equivalently, $\omega_{it} = \mathbf{1}\{t \leq T_0 \ \text{or} \ i \leq N_0\}$, and this leads to the block missingness pattern. A more general empirically common setting is staggered adoption, in which different units become exposed to treatment at different times. Such a design induces a structured missingness pattern in the untreated potential outcome matrix \(Y^{(0)}\).

\begin{figure}[h]
	\centering
 \caption{Missing pattern in MNAR case}
	\includegraphics[width=0.7\textwidth]{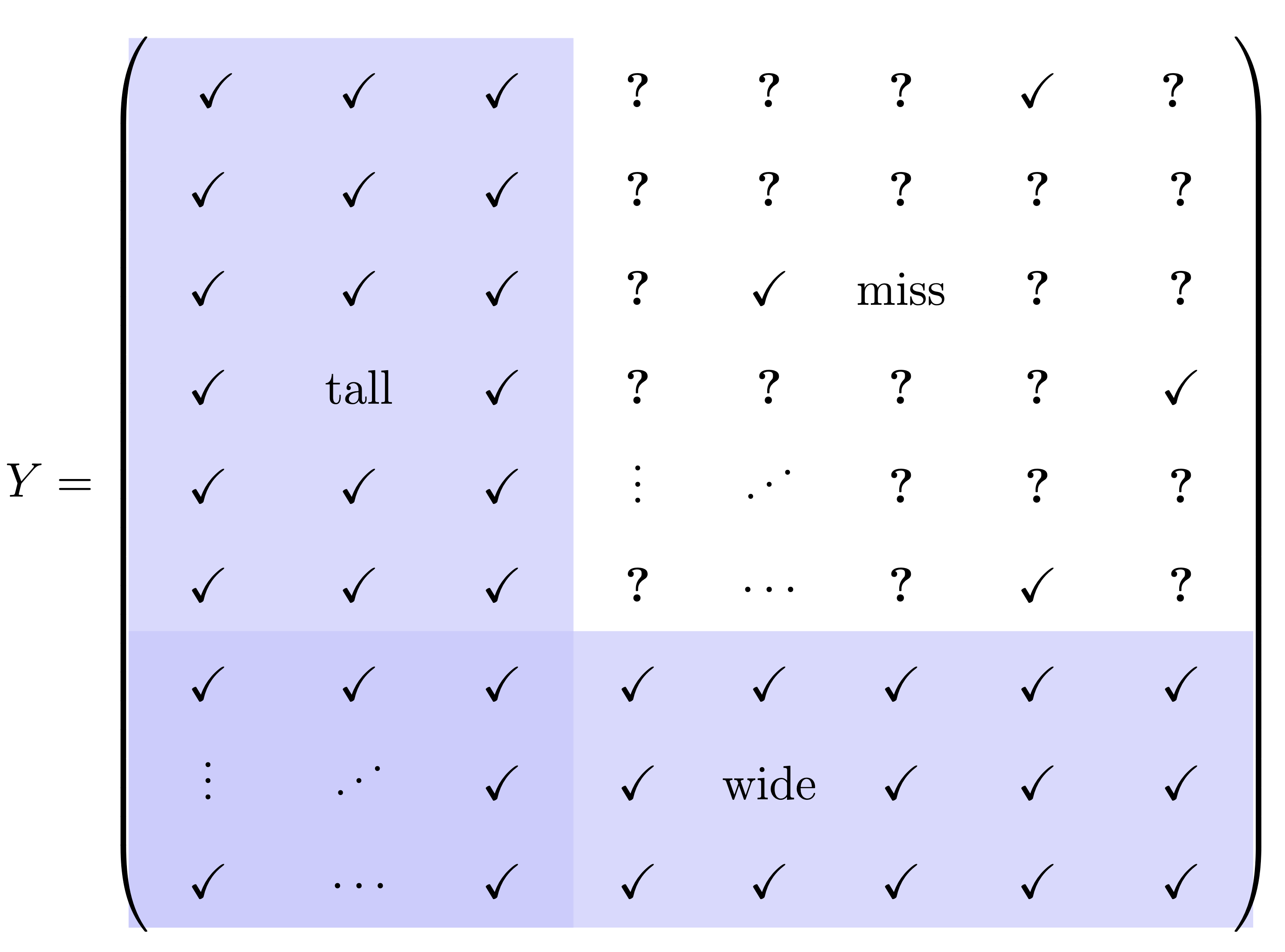}
	\label{fig:MNAR}
\end{figure}

Following the literature on matrix completion under missing not at random (MNAR), e.g., \cite{bai2021matrix,choi2024matrix,yan2024entrywise}, we assume that the missingness pattern takes the form shown in Figure \ref{fig:MNAR}. In this setting, all (or some) entries in the ``miss'' submatrix are unobserved, while all entries in the ``tall'' and ``wide'' submatrices are observed. This pattern is prevalent in causal panel data: the ``wide'' submatrix corresponds to observations for the control group over all time periods, and the ``tall'' submatrix corresponds to observations for all units in the pre-treatment period, where the outcome is the potential outcome under control.

\paragraph{Estimation.}

Note that the entries in the ``tall'' and ``wide'' submatrices are fully observed. Hence, we can apply Algorithm \ref{alg:estimation_obs} to the ``tall'' and ``wide'' submatrices to estimate
\(M_{\mathrm{tall}} = (m_{it})_{i\leq N,\, t \leq T_0}\) and
\(M_{\mathrm{wide}} = (m_{it})_{i\leq N_0,\, t \leq T}\).
Then, as noted in Corollary \ref{cor:singular_vector}, we can estimate the left and right singular vectors of \(M_{\mathrm{tall}}\) and \(M_{\mathrm{wide}}\), respectively. Importantly, the left singular vectors of \(M_{\mathrm{tall}}\) and \(M\) span the same space. Similarly, the right singular vectors of \(M_{\mathrm{wide}}\) and \(M\) span the same space. Hence, by combining the estimator of the left singular vectors of \(M_{\mathrm{tall}}\) with the estimator of the right singular vectors of \(M_{\mathrm{wide}}\), with an appropriate rotation adjustment, we can estimate \(M\).

Specifically, we estimate \(M\) as follows.

\begin{breakablealgorithm}
\caption{Estimation procedure for the MNAR case}
\label{alg:estimation_mnar}
\begin{algorithmic}
\noindent \textbf{Step 1:} From the ``tall'' submatrix \(Y_{\mathrm{tall}} = (y_{it})_{i\leq N,\, t \leq T_0}\), obtain \(\widehat{M}_{\mathrm{tall}}\) using Algorithm \ref{alg:estimation_obs} and compute its left singular vectors \(\widehat{U}_{\mathrm{tall}} \in \bbR^{N \times r}\). \\

\textbf{Step 2:} From the ``wide'' submatrix \(Y_{\mathrm{wide}} = (y_{it})_{i\leq N_0,\, t \leq T}\), obtain \(\widehat{M}_{\mathrm{wide}}\) using Algorithm \ref{alg:estimation_obs} and compute its left and right singular vectors \(\widehat{U}_{\mathrm{wide}} \in \bbR^{N_0 \times r}\) and \(\widehat{V}_{\mathrm{wide}} \in \bbR^{T \times r}\), along with the corresponding singular values \(\widehat{D}_{\mathrm{wide}} \in \bbR^{r \times r}\). \\

\textbf{Step 3:} Obtain the rotation matrix \(\widehat{H}_{\mathrm{adj}} \in \bbR^{r \times r}\) by regressing \(\widehat{U}_{\mathrm{wide}}\) on the submatrix of \(\widehat{U}_{\mathrm{tall}}\) corresponding to \(i \leq N_0\). \\

\textbf{Step 4:} Form the final estimator \(\widehat{M} = \widehat{U}_{\mathrm{tall}} \widehat{H}_{\mathrm{adj}} \widehat{D}_{\mathrm{wide}} \widehat{V}_{\mathrm{wide}}^\top\).
\end{algorithmic}
\end{breakablealgorithm}

Because this estimator is built on Algorithm \ref{alg:estimation_obs} for the fully observed case, we expect it to share similar advantages to those discussed in Section \ref{sec:asymp_obs}.

\paragraph{Asymptotic result.} 

To make this point more precise, we present the convergence rate of our estimator. We begin with an additional assumption.

\begin{assumption}[Block incoherence]\label{asp:incoherence}
Denote the \(i\)-th column of \(U^\top\) and the \(t\)-th column of \(V^\top\) by \(u_i\) and \(v_t\), respectively. Then there exist constants \(c_1,c_2>0\) such that, with probability approaching one,
\begin{gather*}
c_1 \leq \lambda_{\min} \left( \frac{N}{N_0} \sum_{i \leq N_0} u_i u_i^\top \right)
\leq  \lambda_{\max} \left( \frac{N}{N_0} \sum_{i \leq N_0} u_i u_i^\top \right) \leq c_2, \\ 
c_1 \leq \lambda_{\min} \left( \frac{T}{T_0} \sum_{t \leq T_0} v_t v_t^\top \right)
\leq  \lambda_{\max} \left( \frac{T}{T_0} \sum_{t \leq T_0} v_t v_t^\top \right) \leq c_2.
\end{gather*}
\end{assumption}

This assumption can be viewed as an incoherence condition ensuring that the left singular vectors of \(M\) are not dominated by either treated or untreated units, and that the right singular vectors are not dominated by either pre-treatment or post-treatment periods. It is common in the MNAR matrix completion literature (e.g., Assumption D in \cite{bai2021matrix} and Theorem 3.1(v) in \cite{choi2024matrix}) and allows us to relate the properties of the submatrices \(M_{\mathrm{tall}}\) and \(M_{\mathrm{wide}}\) to those of \(M\). For example, if \(\{u_i\}_{i \leq N}\) is stationary, then
\[
\frac{N}{N_0} \sum_{i \leq N_0} u_i u_i^\top \approx \frac{1}{N}\sum_{i \le N} u_i u_i^\top = I_K,
\]
and the condition is satisfied.

We are now in a position to state the convergence rate of the estimator. Let \(\calR_{\mathrm{tall}}\) and \(\calR_{\mathrm{wide}}\) denote the upper bounds on \(\|\widehat{M}_{\mathrm{tall}} - M_{\mathrm{tall}}\|_F\) and \(\|\widehat{M}_{\mathrm{wide}} - M_{\mathrm{wide}}\|_F\), respectively, as given by Theorem \ref{thm:convergence_rate_obs}:
\begin{align*}
\calR_{\mathrm{tall}}
&=  J\sqrt{d_1d_2} + \sqrt{r_2 + r_4}   \min \left\{ \sqrt{T_0} , \norm{M_{2,\mathrm{tall}} P_{Z,\mathrm{sub}}^\perp +  P_X M_{4,\mathrm{tall}} P_{Z,\mathrm{sub}}^\perp}_F \right\}\\
&\quad + \sqrt{r_3 + r_4}   \min \left\{ \sqrt{N} , \norm{ P_X^\perp M_{3,\mathrm{tall}}  +  P_X^\perp M_{4,\mathrm{tall}} P_{Z,\mathrm{sub}} }_F \right\} \\
&\quad + \sqrt{r_4}   \min \left\{ \sqrt{N + T_0} , \| P_X^\perp M_{4,\mathrm{tall}} P_{Z,\mathrm{sub}}^\perp\|_F  \right\}
+ \sqrt{NT_0} \left( \frac{K_1}{J^{\gamma_1^G/d_1}} + \frac{K_1}{J^{\gamma_1^Q/d_2}} + \frac{K_2}{J^{\gamma_2^G/d_1}} + \frac{K_3}{J^{\gamma_2^Q/d_2}} \right),\\
\calR_{\mathrm{wide}}
&=  J\sqrt{d_1d_2} + \sqrt{r_2 + r_4}   \min \left\{ \sqrt{T} , \norm{M_{2,\mathrm{wide}} P_{Z}^\perp +  P_{X,\mathrm{sub}} M_{4,\mathrm{wide}} P_{Z}^\perp}_F \right\}\\
&\quad + \sqrt{r_3 + r_4}   \min \left\{ \sqrt{N_0} , \norm{ P_{X,\mathrm{sub}}^\perp M_{3,\mathrm{wide}}  +  P_{X,\mathrm{sub}}^\perp M_{4,\mathrm{wide}} P_{Z} }_F \right\} \\
&\quad + \sqrt{r_4}   \min \left\{ \sqrt{N_0 + T} , \| P_{X,\mathrm{sub}}^\perp M_{4,\mathrm{wide}} P_Z^\perp\|_F  \right\}
+ \sqrt{N_0 T} \left( \frac{K_1}{J^{\gamma_1^G/d_1}} + \frac{K_1}{J^{\gamma_1^Q/d_2}} + \frac{K_2}{J^{\gamma_2^G/d_1}} + \frac{K_3}{J^{\gamma_2^Q/d_2}} \right),
\end{align*}
where
\[
P_{X,\mathrm{sub}} = \Phi_{\mathrm{sub}}(X)\big( \Phi_{\mathrm{sub}}(X)^{\top}\Phi_{\mathrm{sub}}(X)\big)^{-1} \Phi_{\mathrm{sub}}(X)^{\top},
\quad
P_{Z,\mathrm{sub}} = \Psi_{\mathrm{sub}}(Z)\big( \Psi_{\mathrm{sub}}(Z)^{\top}\Psi_{\mathrm{sub}}(Z)\big)^{-1} \Psi_{\mathrm{sub}}(Z)^{\top},
\]
\(\Phi_{\mathrm{sub}}(X) = \big[ \boldsymbol{\phi}(x_1), \ldots, \boldsymbol{\phi}(x_{N_0}) \big]^\top\), and
\(\Psi_{\mathrm{sub}}(Z) = \big[ \boldsymbol{\psi}(z_1), \ldots, \boldsymbol{\psi}(z_{T_0}) \big]^\top\).
In addition, let \(\delta_N = N_0/N\) and \(\delta_T = T_0/T\). The following theorem provides the convergence rate of our estimator.

\begin{theorem}[Convergence rate for the MNAR case]\label{thm:convergence_rate_mnar}
Suppose that Assumptions \ref{asp:noise}--\ref{asp:moment} hold for the submatrices \(Y_{\mathrm{tall}}\) and \(Y_{\mathrm{wide}}\), and that Assumption \ref{asp:incoherence} holds. In addition, assume that
\[
\frac{\max\{\calR_{\mathrm{wide}}, \calR_{\mathrm{tall}}\}}{\lambda_{\min}\sqrt{\delta_N \delta_T}} \conP 0,
\]
where \(\lambda_{\min}\) is the smallest nonzero singular value of \(M\). Then,
\[
\|\widehat{M} - M\|_F
= O_p\left( \frac{\kappa\,\max\{\calR_{\mathrm{wide}}, \calR_{\mathrm{tall}}\}}{\sqrt{\delta_N \delta_T}} \right),
\]
where \(\kappa = \lambda_{\max}/\lambda_{\min}\).
\end{theorem}

Theorem \ref{thm:convergence_rate_mnar} highlights the advantage of our estimator. A discussion similar to that in Section \ref{sec:asymp_obs} applies. For simplicity, consider a typical case in which \(\kappa\) is bounded and \(N_0 \geq cN\) and \(T_0 \geq cT\) for some \(c>0\). First, consider the most favorable case, \(M=M_1\). In this case, the convergence rate of our estimator is \(O_p(J\sqrt{d_1 d_2})\) provided that \(g_{\iota,k}(\cdot)\) and \(q_{\iota,k}(\cdot)\) are sufficiently smooth. By contrast, if we estimate the submatrices \(M_{\mathrm{tall}}\) and \(M_{\mathrm{wide}}\) using an MNAR low-rank method (e.g., \cite{bai2021matrix}), the convergence rate would be \(O_p(\sqrt{r_1(N+T)})\), which is much larger than ours.

Next, consider the least favorable case, \(M=M_4\). In this case, the convergence rate of our estimator is \(O_p(\sqrt{r_4(N+T)})\), provided that \(g_{\iota,k}(\cdot)\) and \(q_{\iota,k}(\cdot)\) are sufficiently smooth and \(J\sqrt{d_1 d_2} \ll \sqrt{r_4(N+T)}\). This rate coincides with that obtained by applying an MNAR low-rank estimator to the submatrices \(M_{\mathrm{tall}}\) and \(M_{\mathrm{wide}}\).

In addition, if \(M=M_2\), the convergence rate of our estimator is \(O_p(\sqrt{r_2 T})\) provided that \(g_{\iota,k}(\cdot)\) and \(q_{\iota,k}(\cdot)\) are sufficiently smooth and \(J \sqrt{d_1 d_2} \ll \sqrt{r_2 T}\), whereas a low-rank approach yields the rate \(O_p(\sqrt{r_2(N+T)})\). Similarly, if \(M=M_3\), the convergence rate of our estimator is \(O_p(\sqrt{r_3 N})\), whereas a low-rank approach yields \(O_p(\sqrt{r_3(N+T)})\). Therefore, when \(M\) is at least partially explained by observable characteristics \(X\) and \(Z\), incorporating this information can substantially improve estimation accuracy.

\section{Simulated Experiments}\label{sec:simul}

To demonstrate the practical merits and finite-sample performance of our methodology, we conducted several sets of simulation experiments.

\subsection{Change in the relative size of each part}\label{sec:simul_alpha}

First, to study how the relative advantage of our estimator over existing methods varies with the contribution of each component, we change the component weights and compare the estimation performance across estimators. Specifically, we consider the model
\[
M = \alpha_1 M_1 + \alpha_2 M_2 + \alpha_3 M_3 + \alpha_4 M_4,
\]
where \(\sum_{q=1}^4 \alpha_q = 1\) and \(\|M_q\|_F = 2\sqrt{NT}\) for all \(1 \le q \le 4\). We vary the values of \(\alpha_q\) and evaluate the mean squared error of the estimators.

\paragraph{Data generating process.} 

We consider eight characteristics, with \(x_i = [x_{1,i}, x_{2,i}, x_{3,i}, x_{4,i}]^\top\) and
\(z_t = [z_{1,t}, z_{2,t}, z_{3,t}, z_{4,t}]^\top\).
We draw \(x_{1,i} \overset{\mathrm{i.i.d.}}{\sim} \mathrm{Unif}[-1,1]\),
\(x_{2,i} \overset{\mathrm{i.i.d.}}{\sim} \mathrm{Unif}[-0.5,0.5]\),
\(x_{3,i} \overset{\mathrm{i.i.d.}}{\sim} \calN(0,0.2^2)\), and
\(x_{4,i} \overset{\mathrm{i.i.d.}}{\sim} \calN(0,0.3^2)\).
We generate \(z_{1,t}, z_{2,t}, z_{3,t}, z_{4,t}\) in the same way.

For the matrix \(M_1\), we set
\begin{gather*}
g_{1,k}(x_i)
= b^{(1,k)}_{0}
+ \sum_{d=1}^4 \Big( b^{(1,k)}_{d,1} x_{d,i} + b^{(1,k)}_{d,2} x_{d,i}^2 + b^{(1,k)}_{d,3} x_{d,i}^3 + b^{(1,k)}_{d,4} x_{d,i}^4 \Big),
\qquad k \leq r_1 = 17,
\end{gather*}
and draw the coefficients \(b^{(1,k)}_{0}\) and \(b^{(1,k)}_{d,j}\) from the standard normal distribution.
Similarly, we set
\begin{gather*}
q_{1,k}(z_t)
= a^{(1,k)}_{0}
+ \sum_{d=1}^4 \Big( a^{(1,k)}_{d,1} z_{d,t} + a^{(1,k)}_{d,2} z_{d,t}^2 + a^{(1,k)}_{d,3} z_{d,t}^3 + a^{(1,k)}_{d,4} z_{d,t}^4 \Big),
\qquad k \leq r_1 = 17,
\end{gather*}
and draw the coefficients \(a^{(1,k)}_{0}\) and \(a^{(1,k)}_{d,j}\) from the standard normal distribution.

For the matrix \(M_2\), we set \(r_2=3\) and generate \(G_2(X)\) using the same specification as above. In addition, we generate \(v_{1,t} \in \bbR^{3}\) i.i.d.\ from \(\calN\!\big(0,\diag(0.5,1,1.5)\big)\), and stack them into \(V_1 = [v_{1,1}, \ldots, v_{1,T}]^\top\).

For the matrix \(M_3\), we set \(r_3=3\) and generate \(Q_2(Z)\) using the same specification as above. We generate \(w_{1,i} \in \bbR^{3}\) i.i.d.\ from \(\calN\!\big(0,\diag(0.5,1,1.5)\big)\), and stack them into \(W_1 = [w_{1,1}, \ldots, w_{1,N}]^\top\).

For the matrix \(M_4\), we draw \(w_{2,i} \in \bbR^{3}\) i.i.d.\ from \(\calN\!\big(0,\diag(0.5,1,1.5)\big)\) and \(v_{2,t} \in \bbR^{3}\) i.i.d.\ from \(\calN\!\big(0,1.5^2 I_3\big)\).
Lastly, we normalize all four matrices so that \(\|M_q\|_F = 2\sqrt{NT}\) for \(q=1,2,3,4\). We generate the noise entries i.i.d.\ from \(\calN(0,0.5^2)\).

For the fully observed case, we set \(N=T=200\). For the MAR (missing at random) case, we set \(N=T=400\) and the observation probability \(p=0.6\). For the MNAR (missing not at random) case, we set \(N=T=400\) and \(N_0=T_0=200\).

\paragraph{Results.} 

Here, we use a polynomial sieve with \(J=5\), and we set the number of iterations to \(100\). We vary \(\alpha_q\) under the restrictions \(\sum_{q=1}^4 \alpha_q = 1\) and \(\alpha_2 = \alpha_3\). In addition, to keep the rank of \(M\) constant, we restrict attention to cases in which \(\alpha_q \geq 0.01\) for all \(q=1,2,3,4\).

\begin{figure}[h!]
	\centering
    \caption{AMSE under different values of \(\alpha_r\)}
	\includegraphics[width=\textwidth]{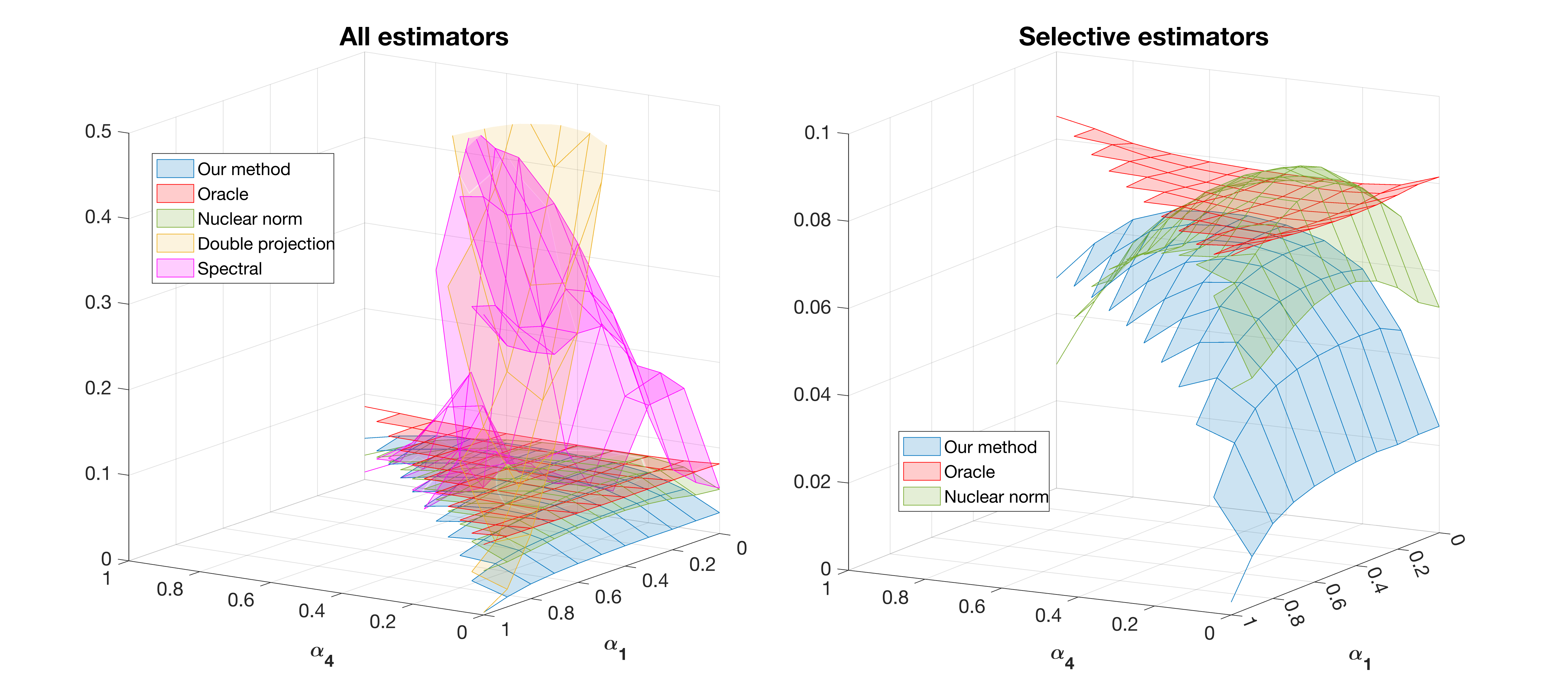}
    \\ \noindent {\small Footnote: We vary \(\alpha_q\) under the restrictions \(\sum_{q=1}^4 \alpha_q = 1\) and \(\alpha_2 = \alpha_3\).}
    \raggedright
	\label{fig:Mesh_obs_alpha}
\end{figure}

We first study the fully observed case. Figure \ref{fig:Mesh_obs_alpha} reports the AMSE (average mean squared error) of the estimators. Here, ``Oracle'' denotes the spectral estimator with known \(r\); ``Nuclear norm'' denotes the plain nuclear-norm-penalized estimator; ``Double projection'' denotes \(P_X Y P_Z\); and ``Spectral'' denotes the spectral estimator with an estimated \(r\). For rank estimation, we use the eigenvalue-ratio method of \cite{ahn2013eigenvalue}.

From the left panel, we see that, in general, the double projection estimator and the spectral estimator with an estimated rank perform poorly relative to the other estimators. From the right panel, we find that our method performs better than the spectral estimator with known \(r\). The AMSE of the oracle estimator is quite stable and is not sensitive to changes in \(\alpha_q\). In contrast, the AMSEs of our method and the nuclear-norm-penalized estimator are strongly affected by \(\alpha_q\). Overall, our method outperforms nuclear-norm penalization except when \(\alpha_1\) is very small and \(\alpha_4\) is large. When \(\alpha_1\) is large and \(\alpha_4\) is small, our estimator performs particularly well.

\begin{figure}[h!]
	\centering
    \caption{\((AMSE_{\mathrm{other}} - AMSE_{\mathrm{our}})/AMSE_{\mathrm{our}}\) under different values of \(\alpha_q\)}
	\includegraphics[width=\textwidth]{Comp_obs_alpha.png}
    \\ \noindent {\small Footnote: In the left panel, the value at \(\alpha_1 = 1\) and \(\alpha_4 = 0.01\) is \(16.66\). In the right panel, the value at \(\alpha_1 = 1\) and \(\alpha_4 = 0.01\) is \(27.12\).}
    \raggedright
	\label{fig:Comp_obs_alpha}
\end{figure}

To assess the relative advantage of our estimator over others, Figure \ref{fig:Comp_obs_alpha} plots \((AMSE_{\mathrm{other}} - AMSE_{\mathrm{our}})/AMSE_{\mathrm{our}}\). Relative to nuclear-norm penalization, the advantage of our estimator increases as \(\alpha_1\) increases. Roughly speaking, the advantage also becomes larger as \(\alpha_4\) decreases. In particular, when \(\alpha_1\) is close to zero (e.g., \(\alpha_1 = 0.01\)), the relative performance improves as \(\alpha_2=\alpha_3\) increases and \(\alpha_4\) decreases. In the right panel, which compares our estimator with the oracle estimator, the dependence on \(\alpha_4\) is less clear; nevertheless, we still observe that the relative advantage increases with \(\alpha_1\).

\begin{figure}[h!]
	\centering
    \caption{Performance comparison in the MAR case}
	\includegraphics[width=\textwidth]{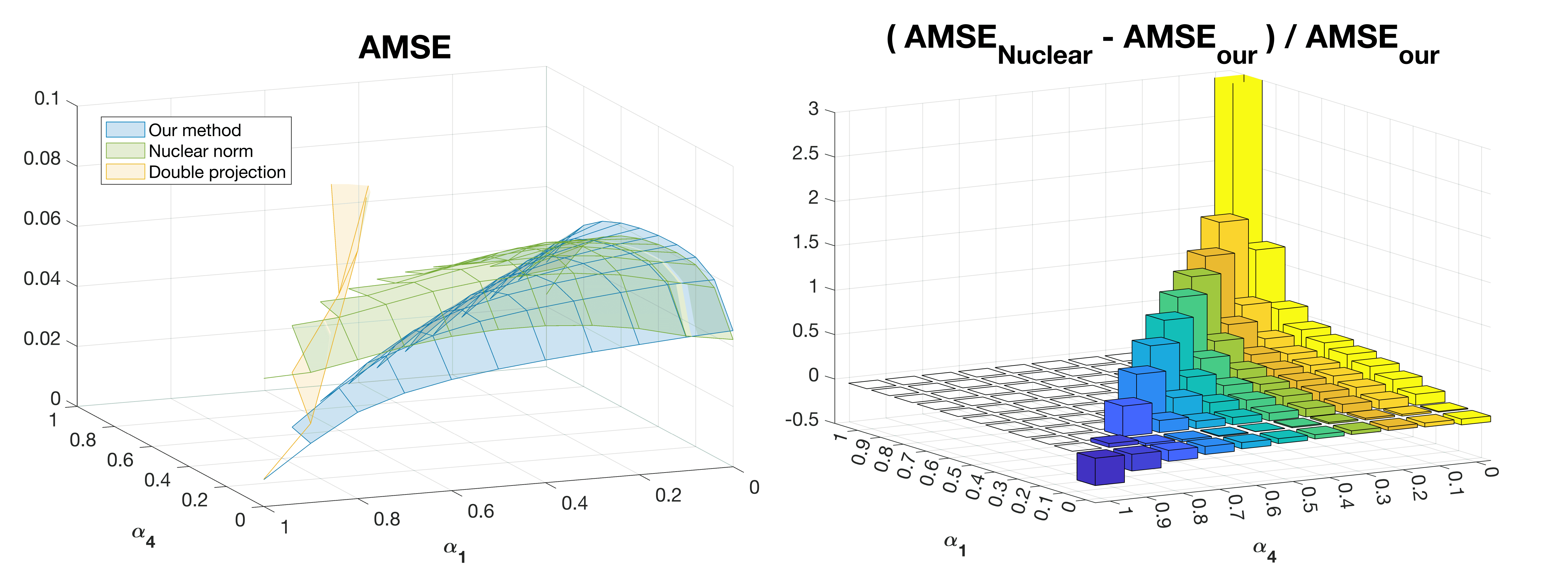}
	\label{fig:graph_miss_alpha}
    \\ \noindent {\small Footnote: In the right panel, the value at \(\alpha_1 = 1\) and \(\alpha_4 = 0.01\) is \(3.77\).}
    \raggedright
\end{figure}

Next, we study the MAR (missing at random) case. Figure \ref{fig:graph_miss_alpha} reports the AMSE (average mean squared error) of the estimators as well as their relative performance. Here, we include the double projection method \(\big(p^{-1} P_X(\Omega \circ Y) P_Z\big)\) and the nuclear-norm-penalized estimator, which is a standard approach in the MAR setting. We find that the double projection method performs very poorly except when \(\alpha_1 = 1\). Relative to nuclear-norm penalization, the advantage of our method increases as \(\alpha_1\) increases. However, the pattern with respect to \(\alpha_4\) is less clear than in the fully observed case. This may be because, in the MAR setting, we cannot separately estimate \(M_2\), \(M_3\), and \(M_4\). When \(\alpha_1\) is very small and \(\alpha_4\) is relatively large, nuclear-norm penalization performs better than our method.

\begin{figure}[h!]
	\centering
    \caption{Performance comparison in the MNAR case}
	\includegraphics[width=\textwidth]{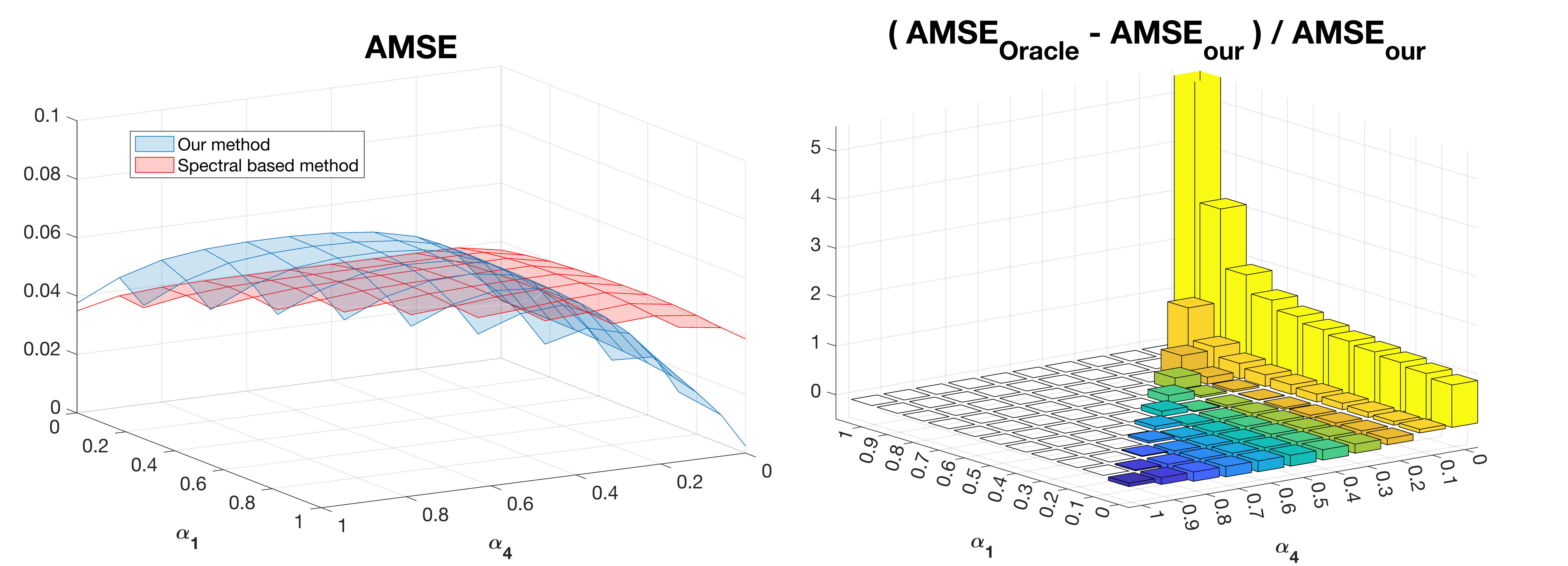}
	\label{fig:graph_mnar_alpha}
    \\ \noindent {\small Footnote: In the right panel, the value at \(\alpha_1 = 1\) and \(\alpha_4 = 0.01\) is \(15.58\).}
    \raggedright
\end{figure}

Lastly, we study the MNAR (missing not at random) case. Figure \ref{fig:graph_mnar_alpha} reports the AMSE of the estimators and the relative advantage of our estimator. Here, we assume the rank is known and compare our method with the standard spectral-based estimator for the MNAR setting in \cite{bai2021matrix,yan2024entrywise}. We find that when \(\alpha_1\) is large and/or \(\alpha_4\) is small, our method generally outperforms the spectral-based estimator. Conversely, when \(\alpha_1\) is small and/or \(\alpha_4\) is large, the spectral-based estimator typically performs better. However, the relative advantage when \(\alpha_1\) is large (or \(\alpha_4\) is small) is substantially greater than the relative disadvantage when \(\alpha_1\) is small (or \(\alpha_4\) is large). Moreover, even when \(\alpha_1\) is close to zero, our method can still perform better when \(\alpha_2\) and \(\alpha_3\) are large.

In summary, across all settings, our method performs substantially better than the competing estimators when \(\alpha_1\) is large and/or \(\alpha_4\) is small. When \(\alpha_1\) is small and/or \(\alpha_4\) is large, the disadvantage of our method is relatively mild compared to the gains achieved when \(\alpha_1\) is large and/or \(\alpha_4\) is small. 

Additionally, to examine how our estimator’s relative advantage over existing methods varies with the rank of each component, we vary the ranks and compare estimation performance. Overall, our method performs markedly better than competing estimators when $r_1$ is large. In contrast, when $r_1$ is small and $r_4$ is large, the advantage is more modest and performance is comparable to that of other estimators. For details, please refer to the online appendix.

\subsection{Simulated tobacco sales experiment}\label{sec:real_data}

In this section, we conduct a real-data experiment using the tobacco sales data in \cite{abadie2010synthetic}, which is widely used in the literature. In 1988, California introduced the first major anti-tobacco legislation in the United States (Proposition 99). To study the effect of this legislation on tobacco sales, \cite{abadie2010synthetic} used per-capita cigarette sales data collected from 1970 to 2000 across 38 U.S.\ states with no anti-tobacco legislation prior to 2000 (\(N=38, T=31\)). We encode these data into a \(38 \times 31\) matrix \(Y\), where the entry \(y_{it}\) represents the ``potential'' outcome of per-capita cigarette sales (in packs) for state \(i\) in year \(t\) under ``control,'' i.e., in the absence of any intervention. To generate missing entries, we artificially introduce interventions (i.e., missingness) for a subset of states: in each iteration, we randomly select 8 states to adopt an intervention (e.g., a tobacco control program) starting from period \(T_0+1\). After rearranging the matrix, this yields the block-missing pattern shown in Figure \ref{fig:MNAR}, with an \(8 \times (T-T_0)\) missing submatrix.

For state-level characteristics, we use the time-averaged retail price of cigarettes, log per-capita state personal income, the percentage of the population aged 18--24, the percentage of adults completing four years of college or more, and per-capita beer consumption. Most of these variables are averaged over the 1970--2000 period. In addition, as a proxy for a state's general preference for tobacco, we use per-capita cigarette sales in 2001. For year-level characteristics, we use log per-capita real GDP and the state-average retail price of cigarettes in each year. We also use the average per-capita cigarette sales of Florida and Michigan as a proxy for general tobacco preference in each year. Although Florida and Michigan are not included among the above 38 states because they had interventions before 2000, the effects of those interventions were relatively mild compared to other treated states. Appendix Section \ref{ref:data_desc} provides additional details on the construction of these characteristics.

We compare the performance of our estimator with that of the spectral-based estimator in \cite{bai2021matrix,yan2024entrywise}, which is a standard method for block-missing patterns. This approach estimates the ``tall'' and ``wide'' submatrices using a spectral estimator. For rank estimation, we use the eigenvalue-ratio method of \cite{ahn2013eigenvalue}. For the projection step in our method, we use a second-order polynomial sieve (\(J=2\)). We set the number of iterations to 100.

We first compare the AMSE (average mean squared error) over the missing entries. Specifically, in each iteration, we sum the squared estimation errors over all \(8(T-T_0)\) missing entries and divide by \(8(T-T_0)\), and then average this quantity across iterations. We compute the estimation error as the difference between the estimated value and the observed per-capita cigarette sales for each missing entry. The first two rows of Table \ref{tab:tobacco_1} report the results for different adoption times \(T_0\). We find that our method outperforms the spectral-based estimator in all cases. In particular, when the number of observed periods is relatively small (i.e., \(T_0\) is small), the performance gap is larger.
\begin{table}[htb]
  \centering
  \caption{Average mean squared errors}
  \begin{tabular}{c|c|cccc}
    \hline\hline
    Target parameter & Method & \(T_0 = 10\) & \(15\) & \(20\) & \(25\) \\[0.5ex]
    \hline
    \multirow{2}[1]{*}{Each missing element} 
      & Ours     & 239.28 & 219.69 & 217.05 & 181.76 \\[0.5ex]
      & Spectral & 268.28 & 238.57 & 233.12 & 194.19 \\[0.5ex]
    \hline
    \multirow{2}[1]{*}{Average of missing elements in each year} 
      & Ours     & 43.56 & 43.67 & 32.19 & 27.01 \\[0.5ex]
      & Spectral & 50.45 & 48.19 & 34.91 & 27.91 \\[0.5ex]
    \hline
    \multirow{2}[2]{*}{Average of all missing elements} 
      & Ours     & 32.89 & 35.77 & 25.28 & 23.75 \\[0.5ex]
      & Spectral & 41.35 & 40.41 & 28.69 & 25.47 \\[0.5ex]
    \hline
  \end{tabular}
  \label{tab:tobacco_1}
\end{table}

Moreover, we consider the AMSE of (i) the average of the missing entries in each year and (ii) the average of all missing entries. For the AMSE of the year-by-year averages, in each iteration we compute the average of the missing entries in each post-intervention year, compute the squared estimation error for each such average, sum these squared errors, and divide by \(T-T_0\). We then average this quantity across iterations. For the AMSE of the overall average, in each iteration we compute the average of all missing entries, compute its squared estimation error, and then average it across iterations. These average-type targets have the advantage that the noise in outcomes is averaged out; therefore, averages of \(y_{it}\) are close to averages of \(m_{it}\). The last four rows of Table \ref{tab:tobacco_1} report the results for different adoption times \(T_0\). We find that our method outperforms the spectral-based estimator, and the performance gap increases when the number of observed periods is small (i.e., when \(T_0\) is small).

As an alternative, we also consider different proxies for tobacco preference. For the proxy of each state's general preference for tobacco, we use the time average of per-capita cigarette sales over the full sample period when estimating the ``wide'' submatrix, and we use the average over the pre-intervention period \(1,\ldots,T_0\) when estimating the ``tall'' submatrix. Similarly, for the proxy of each year's general preference for tobacco, we use the average per-capita cigarette sales across the 30 control states in each year when estimating the ``wide'' submatrix, and we use the average across all 38 states in each year when estimating the ``tall'' submatrix.

\begin{table}[htb]
  \centering
  \caption{Average mean squared errors}
  \begin{tabular}{c|c|cccc}
    \hline\hline
    Target parameter & Method & \(T_0 = 10\) & \(15\) & \(20\) & \(25\) \\[0.5ex]
    \hline
    \multirow{2}[1]{*}{Each missing element} 
      & Ours     & 218.32 & 211.43 & 212.12 & 175.56 \\[0.5ex]
      & Spectral & 268.28 & 238.57 & 233.12 & 194.19 \\[0.5ex]
    \hline
    \multirow{2}[1]{*}{Average of missing elements in each year} 
      & Ours     & 40.26 & 42.67 & 32.02 & 25.35 \\[0.5ex]
      & Spectral & 50.45 & 48.19 & 34.91 & 27.91 \\[0.5ex]
    \hline
    \multirow{2}[2]{*}{Average of all missing elements} 
      & Ours     & 30.78 & 34.83 & 25.82 & 22.84 \\[0.5ex]
      & Spectral & 41.35 & 40.41 & 28.69 & 25.47 \\[0.5ex]
    \hline
  \end{tabular}
  \label{tab:tobacco_2}
\end{table}

Table \ref{tab:tobacco_2} reports the results when we use these alternative proxies as characteristics. We find that the performance of our method improves in most cases, and its relative advantage becomes larger. This type of proxy is not fully consistent with the theory because it may violate the exogeneity condition; however, because the proxy averages out outcome noise, the resulting endogeneity may be negligible in practice. In our experiment, the results using these proxies are indeed favorable.

In summary, the empirical results suggest that incorporating side information can improve estimation of the control potential outcomes in causal panel settings, relative to typical low-rank methods.

\section{Concluding Remarks}

This paper proposes a flexible framework for high-dimensional matrix estimation that systematically incorporates rich side information on both rows and columns. By decomposing the signal into components explained jointly by $(X, Z)$, by $X$ alone, by $Z$ alone, and by neither, and by estimating these components using sieve projection combined with nuclear-norm penalization, our approach accommodates nonlinear covariate effects, avoids explicit rank selection for each component, and automatically thresholds weak or negligible signals. We establish convergence rates that demonstrate robustness across diverse model configurations, matching specialized procedures in favorable settings while remaining competitive when side information is uninformative. We further extend the method to partially observed matrices under both MAR and MNAR mechanisms, including block-missing patterns motivated by causal panel data, and show through simulations and a tobacco-sales application that leveraging side information can substantially improve imputation accuracy and enhance treatment-effect estimation.

\newpage

\bibliographystyle{apalike}
\bibliography{side_info}


\newpage

\appendix

{\LARGE 
\begin{center}
    APPENDIX
\end{center}
}

\section{Data Descriptions}\label{ref:data_desc}

In this section, we describe the data used in our experiment and provide sources.

\begin{itemize}

\item per capita cigarette sales (in packs). Source: The Tax Burden on Tobacco by Orzechowski and Walker from Centers for Disease Control and Prevention (CDC).

\item time average retail price of cigarettes (in dollars): For each state, we derive the average of (annual) retail price of cigarettes over the 1970-2000 period. Here, the retail price includes the average cost and sales tax in the data of `The Tax Burden on Tobacco' (Orzechowski and Walker). We additionally converted it to 2000 dollars using the Consumer Price Index. Source: The Tax Burden on Tobacco by Orzechowski and Walker from Centers for Disease Control and Prevention (CDC).
  
\item per capita state personal income (logged): For each state, we derive the average of (annual) logged per capita state personal income over the 1970-2000 period. We converted the data of U.S. Bureau of Economic Analysis to 2000 dollars using the Consumer Price Index and changed it to the logged value. Source: U.S. Bureau of Economic Analysis (BEA).
  
\item percentage of the population age 18-24: For each state, we derive the average of the percentage of the population age 18-24 in 1970, 1980, 1990, 2000 U.S. Census. Source: Integrated Public Use Microdata Series (IPUMS USA).

\item percentage of adults completing four years of college or higher: For each state, we derive the average of the percentage of adults completing four years of college or higher in 1970, 1980, 1990, 2000 U.S. Census. Source: USDA, Economic Research Service.

\item per capita beer consumption (in gallons): For each state, we derive the average of (annual) per capita beer consumption over the 1977-2000 period because the data start from 1977. Source: Surveillance report \#121: `Apparent per capita alcohol consumption: national, state, and regional trends, 1977-2022' by National Institute on Alcohol Abuse and Alcoholism in NIH.

\item per capita real GDP (logged): We converted the data of World Bank Open Data to the logged value. Source: World Bank Open Data.

\item state average retail price of cigarettes (in dollars): For each year, we derive the average of retail price of cigarettes over 38 states. Here, the retail price includes the average cost and sales tax in the data of `The Tax Burden on Tobacco' (Orzechowski and Walker). We additionally converted it to 2000 dollars using the Consumer Price Index. Source: The Tax Burden on Tobacco by Orzechowski and Walker from Centers for Disease Control and Prevention (CDC).

\end{itemize}

\section{Additional Simulation: Change in the size of rank of each part}\label{sec:add_sim}

To study how the relative advantage of our estimator over existing methods varies with the rank of each component, we vary the ranks and compare the estimation performance. Specifically, we vary \(K_q\) subject to the constraints \(\sum_{q=1}^4 K_q = 15\), \(K_2 = K_3\), and \(K_q \ge 1\) for all \(q\) (If \(15 - (K_1 + K_4)\) is odd, we set \(K_2 = K_3 + 1\)). Here, $r_q = K_q$. In addition, we fix \(\alpha_q\) such that \(\alpha_1 = \cdots = \alpha_4 = 0.25\).

\paragraph{Data generating process.} 

We generate the eight characteristics in the same way as in Section 5.1. For the matrix \(M_1\), we set
\begin{gather*}
g_{1,k}(x_i)
= b^{(1,k)}_{0}
+ \sum_{d=1}^4 \Big( b^{(1,k)}_{d,1} x_{d,i}
+ b^{(1,k)}_{d,2} x_{d,i}^2
+ b^{(1,k)}_{d,3} x_{d,i}^3 \Big),
\qquad \text{for } k \leq K_1,
\end{gather*}
and draw the coefficients \(b^{(1,k)}_{0}\) and \(b^{(1,k)}_{d,j}\) from the standard normal distribution. Similarly, we set
\begin{gather*}
q_{1,k}(z_t)
= a^{(1,k)}_{0}
+ \sum_{d=1}^4 \Big( a^{(1,k)}_{d,1} z_{d,t}
+ a^{(1,k)}_{d,2} z_{d,t}^2
+ a^{(1,k)}_{d,3} z_{d,t}^3 \Big),
\qquad \text{for } k \leq K_1,
\end{gather*}
and draw the coefficients \(a^{(1,k)}_{0}\) and \(a^{(1,k)}_{d,j}\) from the standard normal distribution.

For the matrix \(M_2\), we generate \(G_2(X)\) using the same specification as above. For \(V_1 = [v_{1,1}, \ldots, v_{1,T}]^\top\), we draw \(v_{1,t} \in \bbR^{K_2}\) i.i.d.\ from the first \(K_2\) coordinates of
\begin{gather}\label{eq:v_1}
\calN\!\big(0,\diag(1, 0.75^2, 1.25^2, 0.5^2, 1.5^2, 0.25^2, 1.75^2)\big).
\end{gather}
Similarly, for the matrix \(M_3\), we generate \(Q_2(Z)\) using the same specification as above and draw \(w_{1,i} \in \bbR^{K_3}\) i.i.d.\ from the first \(K_3\) coordinates of \eqref{eq:v_1}.

For the matrix \(M_4\), we draw \(w_{2,i} \in \bbR^{K_4}\) i.i.d.\ from the first \(K_4\) coordinates of
\[
\calN\!\Big(0,\diag(1, 0.75^2, 1.25^2, 0.75^2, 1.25^2, 0.5^2, 1.5^2, 0.5^2, 1.5^2, 0.25^2, 1.75^2, 0.25^2)\Big),
\]
and draw \(v_{2,t} \in \bbR^{K_4}\) i.i.d.\ from \(\calN(0,1.5^2 I_{K_4})\). Lastly, we normalize all matrices so that \(\|M_q\|_F = 2\sqrt{NT}\) for \(q=1,2,3,4\). We generate the noise entries i.i.d.\ from \(\calN(0,1.5^2)\).

\paragraph{Results.} 

We use a polynomial sieve with \(J=4\). The number of iterations and the sample size are the same as in Section 5.1. We first study the fully observed case. Figure \ref{fig:Mesh_obs_rank} reports the AMSE (average mean squared error) of the estimators. We consider the same set of estimators as in Section 5.1. The oracle estimator is the spectral estimator with known \(r\).

\begin{figure}[h!]
	\centering
    \caption{AMSE under different values of \(K_q\)}
	\includegraphics[width=\textwidth]{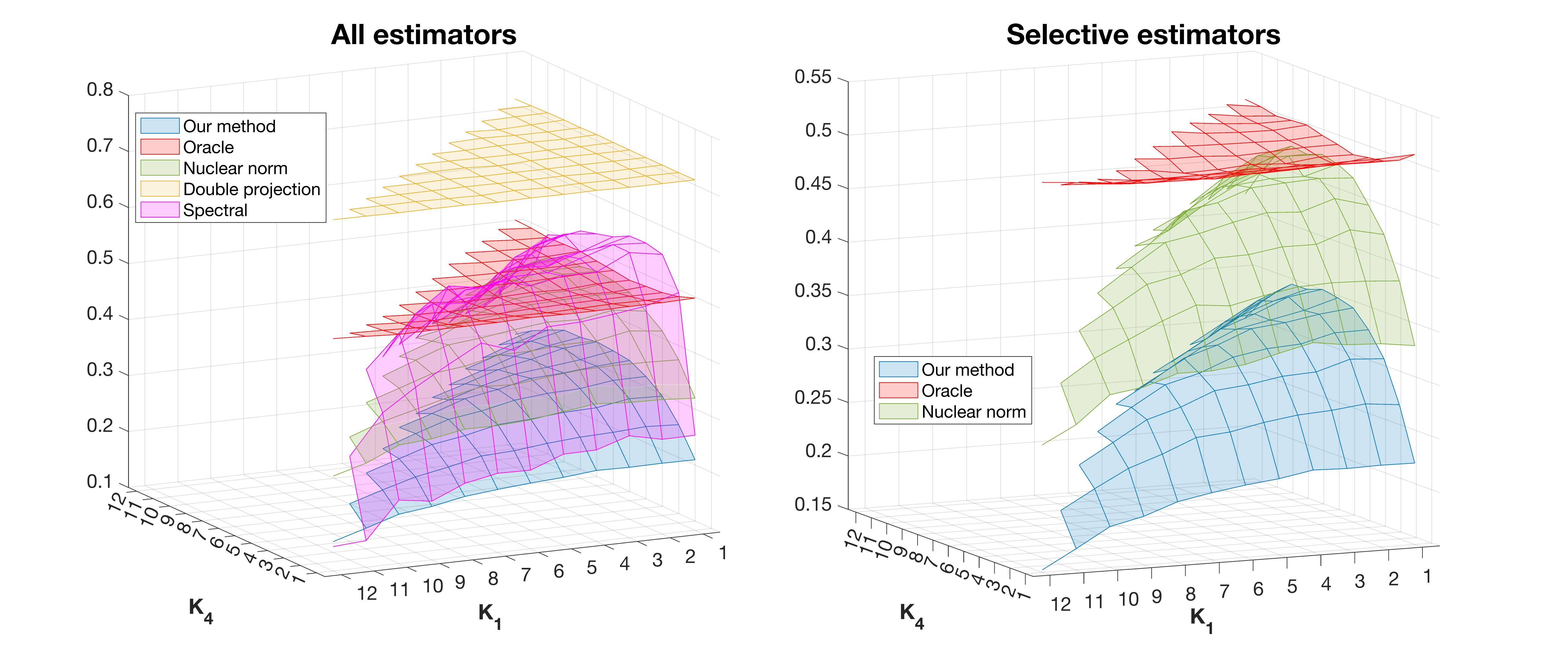}
    \\ \noindent {\small Footnote: We vary \(K_q\) subject to \(\sum_{q=1}^4 K_q = 15\) and \(K_2 = K_3\).}
    \raggedright
	\label{fig:Mesh_obs_rank}
\end{figure}

\begin{figure}[h!]
	\centering
    \caption{\((AMSE_{\mathrm{other}} - AMSE_{\mathrm{our}})/AMSE_{\mathrm{our}}\) under different values of \(K_q\)}
	\includegraphics[width=\textwidth]{Comp_obs_rank.png}
    \\ \noindent {\small Footnote: We vary \(K_q\) subject to \(\sum_{q=1}^4 K_q = 15\) and \(K_2 = K_3\).}
    \raggedright
    \label{fig:Comp_obs_rank}
\end{figure}

From the left panel, we see that the double projection estimator performs poorly and that the spectral estimator with an estimated rank behaves unstably. From the right panel, we find that our method outperforms the other estimators. In addition, the AMSE of the oracle estimator is quite stable and is not sensitive to changes in \(K_q\), whereas the AMSEs of our method and nuclear-norm penalization vary with \(K_q\). Overall, as \(K_1\) increases, our method tends to perform better.

To further assess the relative advantage of our estimator, Figure \ref{fig:Comp_obs_rank} plots \((AMSE_{\mathrm{other}} - AMSE_{\mathrm{our}})/AMSE_{\mathrm{our}}\). Relative to nuclear-norm penalization, the advantage of our estimator becomes larger as \(K_1\) increases. Roughly speaking, the advantage also increases as \(K_4\) decreases. In the right panel, which compares our estimator with the oracle estimator, the dependence on \(K_4\) is less clear; nevertheless, we still observe that the relative advantage increases with \(K_1\).

\begin{figure}[h!]
	\centering
    \caption{Performance comparison in the MAR case}
	\includegraphics[width=\textwidth]{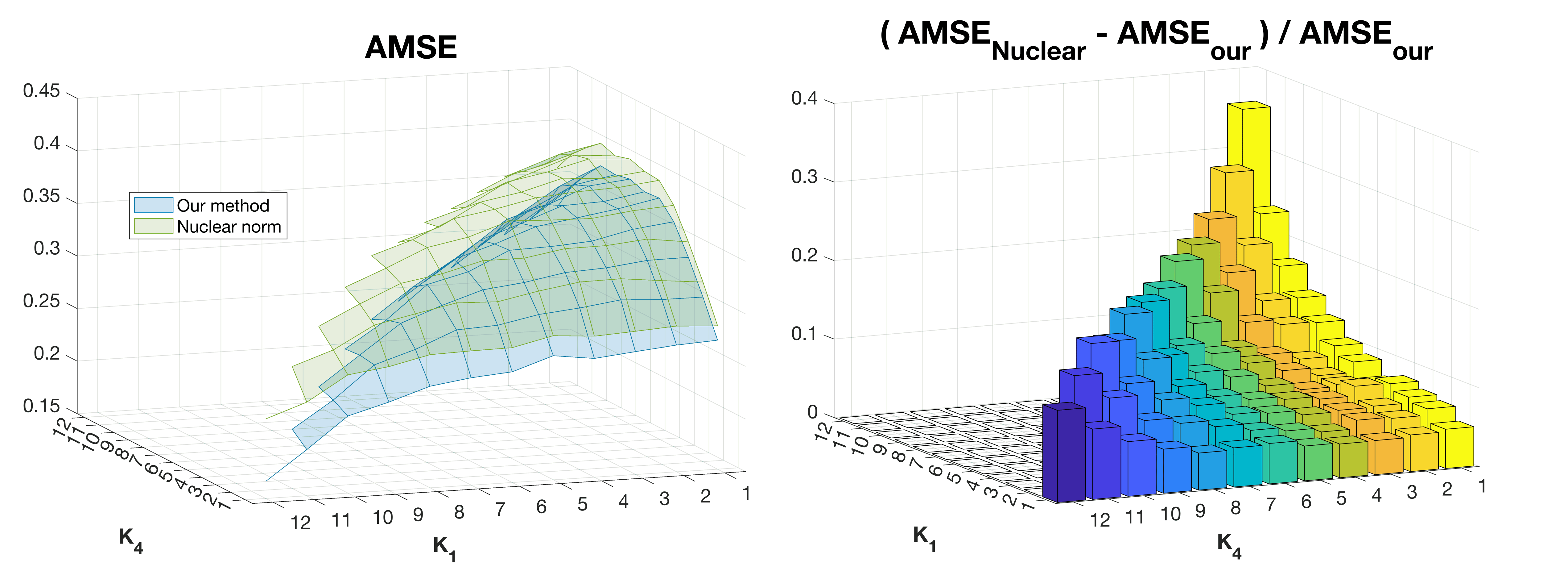}
	\label{fig:graph_miss_rank}
\end{figure}

Next, we study the MAR (missing at random) case. Figure \ref{fig:graph_miss_rank} reports the AMSE (average mean squared error) and the ratio between the AMSEs of our method and nuclear-norm penalization, which is a standard estimator in the MAR setting. Relative to nuclear-norm penalization, our method performs better, and its advantage increases as \(K_1\) increases. However, the pattern with respect to \(K_4\) is less clear.

\begin{figure}[h!]
	\centering
    \caption{Performance comparison in the MNAR case}
	\includegraphics[width=\textwidth]{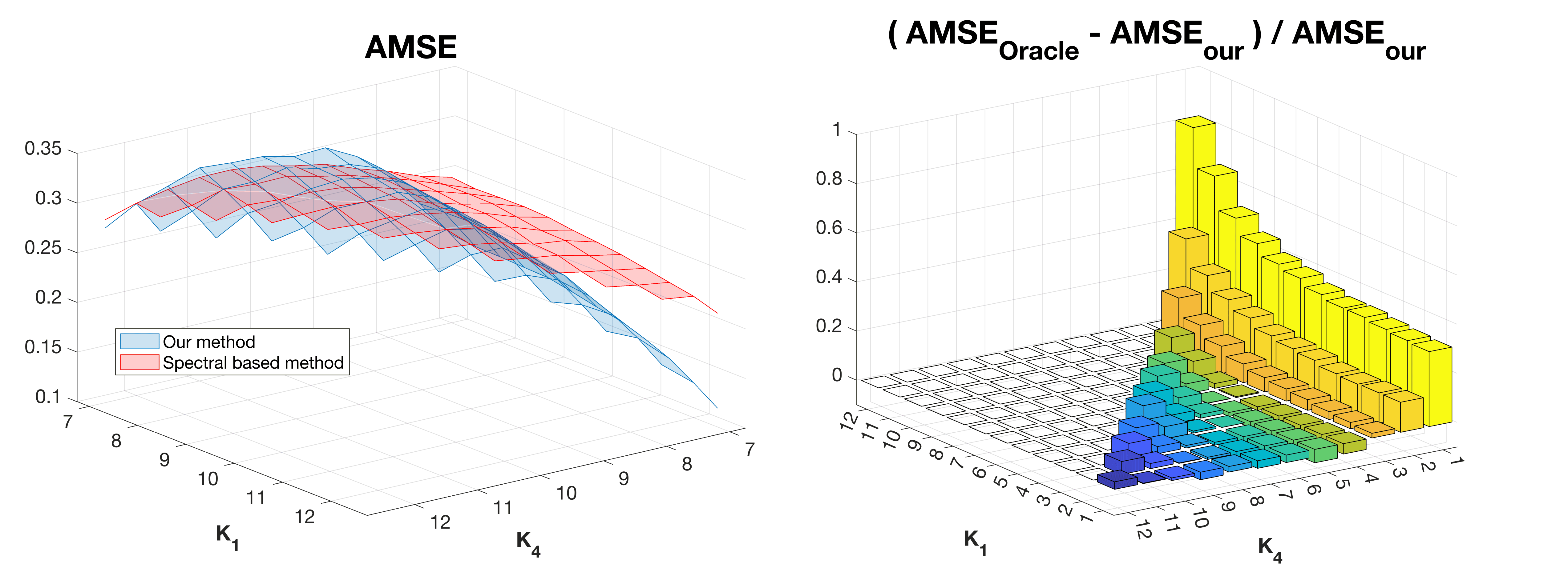}
	\label{fig:graph_mnar_rank}
    \\ \noindent {\small Footnote: In the right panel, the value at \(\alpha_1 = 1\) and \(\alpha_4 = 0.01\) is \(15.58\).}
    \raggedright
\end{figure}

Lastly, we study the MNAR (missing not at random) case. Figure \ref{fig:graph_mnar_rank} reports the AMSE and the relative advantage of our estimator. Here, we consider the same estimators as in Section 5.1. The spectral-based method refers to the approach that uses the oracle estimator to estimate \(M_{\mathrm{tall}}\) and \(M_{\mathrm{wide}}\).

When \(K_1\) is large and/or \(K_4\) is small, our method generally outperforms the spectral-based method. Conversely, when \(K_1\) is small and/or \(K_4\) is large, the spectral-based method performs better. However, the relative advantage in the former case is substantially larger than the relative disadvantage in the latter case. In summary, across all observation patterns, our method performs markedly better than the competing estimators when \(K_1\) is large. When \(K_1\) is small, the disadvantage of our method is relatively mild compared to the gains observed when \(K_1\) is large.

\section{Choosing the Tuning Parameters}\label{sec:tuning}

\subsection{Theory based method}

In practice, the regularization parameters \(\nu_2\), \(\nu_3\), \(\nu_4\),
and \(\nu\) must be selected. We adopt a theory-guided simulation procedure
that is widely used in the literature; see, for example,
\cite{chernozhukov2023inference,choi2024inference,chernozhukov2026inference,choi2026inference}. As shown in the proofs of
Lemmas~\ref{lem:nuclear} and \ref{lem:nuclear_missing}, the key requirements
on the regularization parameters are
\[
2\norm{P_XEP_Z^\perp}
\leq (1-c)\nu_2,
\quad
2\norm{P_X^\perp EP_Z}
\leq (1-c)\nu_3,
\quad
2\norm{P_X^\perp EP_Z^\perp}
\leq (1-c)\nu_4,
\quad
2\norm{\Omega \circ E}\leq (1 - c) \nu,
\]
with high probability for some small constant \(c>0\).
Suppose that
\[
\epsilon_{it}\overset{\mathrm{i.i.d.}}{\sim}
\mathcal{N}(0,\sigma^2).
\]
Let \(U=(u_{it})\in\mathbb{R}^{N\times T}\)
be an independent Gaussian matrix whose entries satisfy
\[
u_{it}\overset{\mathrm{i.i.d.}}{\sim}\mathcal{N}(0,\sigma^2).
\]
Then,
\[
\begin{split}
\big(
\norm{P_XEP_Z^\perp},
\norm{P_X^\perp EP_Z},
\norm{P_X^\perp EP_Z^\perp},
\norm{\Omega \circ E}
\big)
\overset{d}{=}\,
\big(
\norm{P_XUP_Z^\perp},
\norm{P_X^\perp UP_Z},
\norm{P_X^\perp UP_Z^\perp},
\norm{\Omega \circ U}
\big).
\end{split}
\]
Let \(\overline{Q}(A;b)\) denote the \(b\)-quantile of a random variable
\(A\). For \(\delta_{NT}\in(0,1)\), with \(\delta_{NT}=o(1)\), we set
\begin{align*}
\nu_2
&=
2(1+c_1)
\overline{Q}\left(
\norm{P_XUP_Z^\perp};
1-\delta_{NT}
\right),
\\
\nu_3
&=
2(1+c_1)
\overline{Q}\left(
\norm{P_X^\perp UP_Z};
1-\delta_{NT}
\right),
\\
\nu_4
&=
2(1+c_1)
\overline{Q}\left(
\norm{P_X^\perp UP_Z^\perp};
1-\delta_{NT}
\right),
\\
\nu
&=
2(1+c_1)
\overline{Q}\left(
\norm{\Omega \circ U};
1-\delta_{NT}
\right),
\end{align*}
where \(c_1>0\) is a small constant. Then, we have each of 
\begin{align*}
2\norm{P_XEP_Z^\perp}
&\leq
\left(1-\frac{c_1}{1+c_1}\right)\nu_2,
&
2\norm{P_X^\perp EP_Z}
&\leq
\left(1-\frac{c_1}{1+c_1}\right)\nu_3,
\\
2\norm{P_X^\perp EP_Z^\perp}
&\leq
\left(1-\frac{c_1}{1+c_1}\right)\nu_4,
&
2\norm{\Omega \circ E}
&\leq
\left(1-\frac{c_1}{1+c_1}\right)\nu,
\end{align*}
with probability at least \(1-\delta_{NT}\). Thus, the required conditions hold with $c=\frac{c_1}{1+c_1}$. In our simulations, we set \(c_1=1/7\) and \(\delta_{NT}=0.05\). The
quantiles above are approximated by repeatedly generating independent
Gaussian matrices \(U^{(1)},\ldots,U^{(B)}\) and computing the corresponding
empirical quantiles.

Implementation of this procedure requires an estimate of \(\sigma^2\).
We first obtain an initial estimate using a conventional low-rank estimator.
Specifically, let \(\widehat{M}_{\mathrm{init}}\), with entries
\(\widehat{m}_{\mathrm{init},it}\), denote the initial estimator, and define
\[
\widehat{\epsilon}_{\mathrm{init},it}
=
y_{it}-\widehat{m}_{\mathrm{init},it},
\qquad
\widehat{\sigma}_{\mathrm{init}}^2
=
\frac{1}{NT}
\sum_{i=1}^N\sum_{t=1}^T
\widehat{\epsilon}_{\mathrm{init},it}^{\,2}.
\]
Using \(\widehat{\sigma}_{\mathrm{init}}^2\), we generate the Gaussian
matrices and select the regularization parameters according to the procedure
described above. Let \(\widetilde{M}_{\mathrm{init}}\), with entries
\(\widetilde{m}_{\mathrm{init},it}\), denote the resulting estimator. We then
update the residuals and the variance estimate according to
\[
\widetilde{\epsilon}_{\mathrm{init},it}
=
y_{it}-\widetilde{m}_{\mathrm{init},it},
\qquad
\widetilde{\sigma}_{\mathrm{init}}^2
=
\frac{1}{NT}
\sum_{i=1}^N\sum_{t=1}^T
\widetilde{\epsilon}_{\mathrm{init},it}^{\,2}.
\]
We subsequently regenerate the Gaussian matrices using the updated variance
estimate, recompute the regularization parameters, and re-estimate \(M\).
This procedure is repeated until the variance estimate and regularization
parameters converge.

\subsection{Cross validation method}

Additionally, we propose a cross-validation procedure for selecting the
regularization parameter. For simplicity, we describe the procedure for the MAR
setting; the similar procedure can be easily applied to the fully observed
setting. Let
\[
\Omega=\{(i,t):\omega_{it}=1\}
\]
denote the set of observed entries, and partition \(\Omega\) into \(K\)
approximately equal-sized disjoint folds:
\[
\Omega=\Omega_1\cup\cdots\cup\Omega_K.
\]
For each fold \(k\), define the training set and its observation frequency as
\[
\Omega_{-k}=\Omega\setminus\Omega_k,
\qquad
p_{-k}=\frac{|\Omega_{-k}|}{NT}.
\]
Using only the observations in \(\Omega_{-k}\), we first construct
\(\widehat{M}_1^{(-k)}\). For each candidate constant \(c\in\mathcal{C}\), we
then estimate the remaining component by
\[
\widehat{M}_{\mathrm{rest}}^{(-k)}(c)
=
\arg\min_A
\left\{
\left\|
\mathcal{P}_{\Omega_{-k}}
\left(Y-\widehat{M}_1^{(-k)}-A\right)
\right\|_F^2
+
\nu_{-k}(c)\norm{A}_*
\right\},
\]
where
\[
\nu_{-k}(c)
=
c\sqrt{p_{-k}(N+T)}.
\]
The corresponding fold-specific estimator of \(M\) is
\[
\widehat{M}^{(-k)}(c)
=
\widehat{M}_1^{(-k)}
+
\widehat{M}_{\mathrm{rest}}^{(-k)}(c).
\]
We evaluate its prediction error on the held-out fold according to
\[
\operatorname{CV}_k(c)
=
\frac{1}{|\Omega_k|}
\left\|
\mathcal{P}_{\Omega_k}
\left(
Y-\widehat{M}^{(-k)}(c)
\right)
\right\|_F^2.
\]
The overall cross-validation criterion is
\[
\operatorname{CV}(c)
=
\frac{1}{K}
\sum_{k=1}^K
\operatorname{CV}_k(c),
\]
and the tuning constant is selected as
\[
\widehat{c}
=
\arg\min_{c\in\mathcal{C}}
\operatorname{CV}(c).
\]
Finally, we re-estimate \(M_1\) using all observed entries and denote the
resulting estimator by \(\widehat{M}_1\). We then set
\[
\widehat{\nu}
=
\widehat{c} p^{1/2}
\sqrt{N+T},
\qquad
p=\frac{|\Omega|}{NT},
\]
and refit the remaining component using
\[
\widehat{M}_{\mathrm{rest}}
=
\arg\min_A
\left\{
\left\|
\mathcal{P}_{\Omega}
\left(Y-\widehat{M}_1-A\right)
\right\|_F^2
+
\widehat{\nu}\norm{A}_*
\right\}.
\]
The final cross-validated estimator is
\[
\widehat{M}
=
\widehat{M}_1+\widehat{M}_{\mathrm{rest}}.
\]

\section{Heterogeneous Observation Probability Case}\label{sec:mar_hetero}

In this section, we allow the observation probability to vary across units. Specifically, for each unit \(i\), we define the unit-specific observation probability as $p_i := \bbE\left[\omega_{it} \mid x_i, w_{1,i}, w_{2,i}\right]$. Thus, we allow the observation probability to depend on \(x_i\), \(w_{1,i}\), and \(w_{2,i}\).

The estimation procedure proceeds as follows.

\begin{breakablealgorithm}
		\caption{Estimation procedure for MAR case}
		\label{alg:estimation_mar_hetro}
		\begin{algorithmic}
    	\noindent \textbf{Step 1:} Derive $\widehat{M}_1 = P_X \widehat{\Pi}^{-1} (\Omega \circ Y) P_Z$ where $\widehat{\Pi} = \diag(\widehat{p}_1, \dots, \widehat{p}_N )$ and $\widehat{p}_i = \frac{1}{T} \sum_{t} \omega_{it}$. \\
			\textbf{Step 2:} Apply the nuclear norm penalization to $\Omega \circ (Y - \widehat{M}_1)$:
   \begin{gather*}
\widehat{M}_{rest} \coloneqq \arg\min_{A: \norm{A}_{\infty} \leq M_{\max}} \norm{\Omega \circ (Y - \widehat{M}_1 - A)}_F^2 + \nu \norm{A}_*, 
   \end{gather*}
   where $M_{\max} >0 $ is some large constant and $\nu = C p_{\max}^{1/2} \sqrt{N+T}$ with a constant $C > 0$.\\
	\textbf{Step 3:} Get the final estimator, $\widehat{M} = \widehat{M}_1 + \widehat{M}_{rest}$.
		\end{algorithmic}
	\end{breakablealgorithm}

We introduce several additional assumptions.

\begin{assumption}[Random missingness]\label{asp:mar_hetero}
Conditional on $\{X,Z,M\}$, the random variables
\(\{\omega_{it}: i \leq N,\, t \leq T\}\) are independent across $i$ and $t$. Moreover, we assume that $p_i = \bbE\left[\omega_{it} \mid X,Z,M\right]$ and there exist $0 < p_{\min},p_{\max} \leq 1$ such that $p_{\min} \leq p_i \leq p_{\max}$ with probability converging to one.
\end{assumption}

\begin{assumption}[Sieve approximation]\label{asp:sieve_mar_hetero}
(i) Assumption \ref{asp:sieve} (i) holds. \\
(ii) The sieve approximation satisfies
\begin{gather*}
\frac{\min\{\sqrt{N},\sqrt{T}\}}{\sqrt{p_{\max}}}
\left(
\frac{K_1}{J^{\gamma_1^G/d_1}} + \frac{K_1}{J^{\gamma_1^Q/d_2}} + \frac{K_2}{J^{\gamma_2^G/d_1}} + \frac{K_3}{J^{\gamma_2^Q/d_2}}
\right) \rightarrow 0.
\end{gather*}
(iii) We have
$$
 J \sqrt{d_1 d_2 } \ll  \sqrt{p_{\max}p_{\min}} \sqrt{N+T}    \min \left\{ 1, 
 \sqrt{p_{\min}} \frac{\sqrt{T}}{\sqrt{N}} \right\}.
$$
\end{assumption}

The following theorem provides the convergence rate of our estimator.

\begin{theorem}[Convergence rate for the MAR case]\label{thm:convergence_rate_mar_hetero}
Suppose that Assumptions \ref{asp:noise}, \ref{asp:basis}, \ref{asp:moment}, \ref{asp:mar_hetero}, and \ref{asp:sieve_mar_hetero} hold. Then, if
\[
M_{\max} \geq \big\| M_2 P_Z^\perp+ P_X^\perp M_3 + M_4 - P_X M_4 P_Z \big\|_{\infty},
\]
we have
\begin{align*}
\|\widehat{M} - M\|_F
&= O_p\Bigg(
 \frac{J\sqrt{d_1d_2}}{\sqrt{p_{\min}}}  +  \frac{\sqrt{N}}{\sqrt{T}} \frac{J \sqrt{d_1 d_2}}{ p_{\min}} \\
&\qquad\quad + \sqrt{r^*}\,
\min \left\{
\frac{\sqrt{p_{\max}} \sqrt{N + T}\left(1 + M_{\max} \right)}{p_{\min}},
\, \norm{M_2 P_Z^\perp+ P_X^\perp M_3 + M_4 - P_X M_4 P_Z}_F
\right\} \\
&\qquad\quad
+ \sqrt{NT}
\left[
\frac{K_1}{J^{\gamma_1^G/d_1}} + \frac{K_1}{J^{\gamma_1^Q/d_2}} + \frac{K_2}{J^{\gamma_2^G/d_1}} + \frac{K_3}{J^{\gamma_2^Q/d_2}}
\right]
\Bigg),
\end{align*}
where \(r^* = \max\{r_2, r_3, r_4\}\).
\end{theorem}

The result is similar to its counterpart in the homogeneous case. The key
difference is the presence of the additional term, $p_{\min}^{-1}\sqrt{\frac{N}{T}}\,J\sqrt{d_1d_2}$. This term arises from estimating \(\Pi\), which is more challenging than
estimating the scalar \(p\) in the homogeneous case. If \(N/(Tp_{\min}) = O(1)\), a
conclusion analogous to that in the homogeneous case can be obtained.

\section{Additional Assumptions for Identification}\label{sec:identification}

To identify each component of \(M\), we can impose the following additional
orthogonality restrictions:
\begin{gather}\label{eq:identification}
 P_XW_1=P_XW_2=0,
\qquad
P_ZV_1=P_ZV_2=0.   
\end{gather}
Under these restrictions, \(M_1\), \(M_2\), \(M_3\), and \(M_4\) are
identified up to the sieve approximation error:
\[
M_1 \approx P_XMP_Z,
\qquad
M_2 \approx P_XMP_Z^\perp,
\qquad
M_3 \approx P_X^\perp MP_Z,
\qquad
M_4 \approx P_X^\perp MP_Z^\perp.
\]

Because our parameter of interest is \(M\), rather than its individual
components, their separate identification is not required for our analysis.
Nevertheless, these additional identification restrictions simplify the
representation of the results.

The following corollaries present simplified versions of the preceding results.

\begin{corollary}[Simplification of Theorem \ref{thm:convergence_rate_obs}]\label{corr:convergence_rate_obs_iden}
Suppose that the assumptions of Theorem \ref{thm:convergence_rate_obs} hold. In addition, assume that the restriction \eqref{eq:identification} holds. Then,
\begin{align*}
\|\widehat{M} - M\|_F
&= O_p\Bigg(
J\sqrt{d_1d_2}
+ \sqrt{r_2}\,\min \Big\{ \sqrt{T},\, \norm{M_2 }_F \Big\} + \sqrt{r_3}\,\min \Big\{ \sqrt{N},\,  \norm{ M_3  }_F \Big\}\\
&\qquad\quad + \sqrt{r_4}\,\min \Big\{ \sqrt{N + T},\, \norm{ M_4 }_F \Big\} 
+ \sqrt{NT}\left[ \frac{K_1}{J^{\gamma_1^G/d_1}} + \frac{K_1}{J^{\gamma_1^Q/d_2}}  \right]
\Bigg).
\end{align*}
\end{corollary}

\begin{corollary}[Simplification of Theorem \ref{thm:convergence_rate_mar}]\label{corr:convergence_rate_mar_iden}
Suppose that the assumptions of Theorem \ref{thm:convergence_rate_mar} hold. In addition, assume that the restriction \eqref{eq:identification} holds. Then, if
\[
M_{\max} \geq \big\| M_2 + M_3 + M_4 \big\|_{\infty},
\]
we have
\begin{align*}
\|\widehat{M} - M\|_F
&= O_p\Bigg(
\frac{J\sqrt{d_1d_2}}{\sqrt{p}}
+ \sqrt{r^*}\,
\min \left\{
\frac{\sqrt{N + T}\left(1 + M_{\max} \right)}{\sqrt{p}},
\, \norm{M_2 + M_3 + M_4}_F
\right\} \\
&\qquad\quad
+ \sqrt{NT}
\left[
\frac{K_1}{J^{\gamma_1^G/d_1}} + \frac{K_1}{J^{\gamma_1^Q/d_2}} 
\right]
\Bigg),
\end{align*}
where \(r^* = \max\{r_2, r_3, r_4\}\).
\end{corollary}

Since a simplified version of Theorem \ref{thm:convergence_rate_mnar}
follows directly from Corollary \ref{corr:convergence_rate_obs_iden},
we do not state it separately.

\section{Entrywise Convergence Rate}\label{sec:entrywise}

In this section, we derive an entrywise convergence rate for our estimator.
To simplify the notation, we impose the identification condition
\eqref{eq:identification}. We first consider the MAR setting.

\begin{assumption}[Sample complexity and signal strength]
\label{asp:entrywise_mar}
Suppose that
\[
\max_i \norm{\boldsymbol{\phi}(x_i)}
\leq \sqrt{Jd_1\mu_{\phi}},
\qquad
\max_t \norm{\boldsymbol{\psi}(z_t)}
\leq \sqrt{Jd_2\mu_{\psi}},
\]
for some positive parameters \(\mu_{\phi}\) and \(\mu_{\psi}\). Let $M_{rest}=M_2+M_3+M_4$, and let \(\kappa_{rest}\), \(\mu_{rest}\), \(r_{rest}\), and
\(\Upsilon_{rest,\min}\) denote its condition number, incoherence parameter,
rank, and smallest nonzero singular value, respectively. In particular,
\(\kappa_{rest}\) is the ratio of the largest singular value of \(M_{rest}\)
to \(\Upsilon_{rest,\min}\). The incoherence parameter \(\mu_{rest}\) is
chosen such that
\[
\norm{U_{rest}}_{2,\infty}
\leq
\sqrt{\frac{\mu_{rest}r_{rest}}{N}},
\qquad
\norm{V_{rest}}_{2,\infty}
\leq
\sqrt{\frac{\mu_{rest}r_{rest}}{T}},
\]
where \(U_{rest}\) and \(V_{rest}\) contain the left and right singular
vectors of \(M_{rest}\), respectively.
We assume that the following conditions hold with probability tending to one:
\begin{align*}
\textnormal{(i)}\quad
&
r_{rest}\kappa_{rest}^2\mu_{rest}^{1/2}
\max\left\{
N\sqrt{\log N},
T\sqrt{\log T}
\right\}
\ll
\sqrt{p} \Upsilon_{rest,\min}
\min\left\{
\sqrt{N},
\sqrt{T}
\right\},
\\
\textnormal{(ii)}\quad
&
r_{rest}^2\kappa_{rest}^4\mu_{rest}^2
\max\left\{
N\log^3 N,
T\log^3 T
\right\}
\ll
p\min\left\{
N^2,
T^2
\right\},
\\
\textnormal{(iii)}\quad
&
(1+M_{\max})J\sqrt{d_1d_2\mu_{\phi}\mu_{\psi}}
\ll
\sqrt{p}
\min\left\{
\sqrt{N},
\sqrt{T}
\right\},\\
\textnormal{(iv)}\quad 
&
\max\{\sqrt{N},\sqrt{T}\} K_1
 \ll 
J^{\gamma_1^G/d_1} + J^{\gamma_1^Q/d_2}
.
\end{align*}
\end{assumption}

Then, by following the typical argument for the entrywise convergence rate of the nuclear norm penalization in \cite{chen2020noisy}, we obtain the following entrywise convergence rate.

\begin{theorem}[MAR case]\label{thm:entry_convergence_rate_mar}
Suppose that Assumptions \ref{asp:noise}, \ref{asp:basis}, \ref{asp:sieve} (i), \ref{asp:moment}, \ref{asp:mar}, and \ref{asp:entrywise_mar} hold. Then, we have
\begin{align*}
\|\widehat{M} - M\|_{\infty}
&= O_p\Bigg(
\frac{ (1 + \norm{M_{rest}}_{\infty}) J \sqrt{d_1 d_2 \mu_\phi \mu_\psi}}{\sqrt{p}\sqrt{NT}}\,
+ \frac{\kappa_{rest}^2 \mu_{rest} r_{rest}^{3/2} \max \{\sqrt{\log N}, \sqrt{\log T} \}}{\sqrt{p} \min\{\sqrt{N}, \sqrt{T} \}} \\
&\qquad\quad
+ 
\left[
\frac{K_1}{J^{\gamma_1^G/d_1}} + \frac{K_1}{J^{\gamma_1^Q/d_2}} 
\right]
\Bigg).
\end{align*}
\end{theorem}

The first term arises from estimating \(M_1\), the second from estimating
\(M_{\mathrm{rest}}\), and the final term represents the sieve approximation
error. One limitation of the entrywise convergence rate is the signal-strength
condition in Assumption~\ref{asp:entrywise_mar}(i), which requires the smallest
nonzero singular value of \(M_{\mathrm{rest}}\) to be sufficiently large.
Consequently, when \(M_{\mathrm{rest}}\) is weak in the sense that its smallest
nonzero singular value is too small, this entrywise bound may not apply. Such a signal-strength condition is standard in the literature on entrywise convergence rates for nuclear-norm-penalized estimator.

The following theorem partially addresses the limitation of the preceding theorem.

\begin{theorem}[MAR case with weak signals]\label{thm:entry_convergence_rate_mar_weak_signal}
Suppose that assumptions of Theorem \ref{thm:convergence_rate_mar} hold. In addition, assume that the following condition hold with probability tending to one:
\begin{gather}\label{asp:entrywise_mar_weak_signal}
\Upsilon_{rest,\min} \kappa_{rest} r_{rest}^{1/2} \lesssim \sqrt{p} \sqrt{N+T} .    
\end{gather}
Then, we have
\begin{align*}
\|\widehat{M} - M\|_{\infty}
= O_p\Bigg(
\frac{ (1 + \norm{M_{rest}}_{\infty}) J \sqrt{d_1 d_2 \mu_\phi \mu_\psi}}{\sqrt{p}\sqrt{NT}}\,
+ \norm{M_{rest}}_{\infty} + 
\left[
\frac{K_1}{J^{\gamma_1^G/d_1}} + \frac{K_1}{J^{\gamma_1^Q/d_2}} 
\right]
\Bigg).
\end{align*}
\end{theorem}

When the nonzero singular values of \(M_{\mathrm{rest}}\) are sufficiently small, as specified in Assumption~\eqref{asp:entrywise_mar_weak_signal}, this theorem provides a useful bound. In particular, if \(M_{\mathrm{rest}}=\boldsymbol{0}_{N\times T}\), then
\begin{align*}
\norm{\widehat{M}-M}_{\infty} = 
O_p\left(
\frac{J\sqrt{d_1d_2\mu_\phi\mu_\psi}}{\sqrt{pNT}}
+
\frac{K_1}{J^{\gamma_1^G/d_1}}
+
\frac{K_1}{J^{\gamma_1^Q/d_2}}
\right),
\end{align*}
which can be substantially faster than the rate obtained using standard low-rank matrix estimation. However, if the signal strength \(\Upsilon_{\mathrm{rest},\min}\) falls into neither the strong-signal regime specified in Assumption~\ref{asp:entrywise_mar}(i) nor the weak-signal regime specified in Assumption~\eqref{asp:entrywise_mar_weak_signal}, our results do not provide an entrywise error bound.

In the fully observed setting, applying the standard
argument for establishing the entrywise convergence rate of a
nuclear-norm-penalized estimator, as developed in \cite{chen2020noisy}, poses a
challenge. This argument constructs auxiliary leave-one-out estimators and
requires the noise entries to be independent across both the unit and time
dimensions. For example, after projecting \(Y\) as \(P_XYP_Z^\perp\), the
corresponding noise term becomes \(P_XEP_Z^\perp\), whose entries are generally
not independent across \(i\) and \(t\). Consequently, the standard leave-one-out
argument may not be applied. We therefore use a different approach
based on the entrywise eigenvector perturbation bounds established in
\cite{fan2018l}.

\begin{assumption}[Sample complexity and signal strength]
\label{asp:entrywise_observed}
Let \(\kappa_{\iota}\), \(\mu_{\iota}\), \(r_{\iota}\), \(\Upsilon_{\iota,r}\), and
\(\Upsilon_{\iota,\min}\) denote its condition number, incoherence parameter,
rank, $r$-th largest singular value, and smallest nonzero singular value of $M_{\iota}$, respectively. We assume that the following conditions hold with probability tending to one:
\begin{align*}
(i) \ \ \Upsilon_{2,\min} \gg \sqrt{T} , \quad   
\min_{1 \leq r \leq r_2} \frac{\abs{\Upsilon_{2,r} - \Upsilon_{2,r+1}}  }{\Upsilon_{2,\min}} > c , \quad 
(ii) \ \ \Upsilon_{3,\min} \gg \sqrt{N},\quad
\min_{1 \leq r \leq r_3} \frac{\abs{\Upsilon_{3,r} - \Upsilon_{3,r+1}}  }{\Upsilon_{3,\min}} > c
\end{align*}
for some constant $c > 0$, and
\begin{align*}
(iii) \ \ & J^{1/2} \max\{ \sqrt{d_1 \mu_{\phi}},\sqrt{d_2 \mu_{\psi}} \} r_4 \kappa_4^2 \mu_4^{1/2} \max \{ N \sqrt{\log N}, T \sqrt{\log T} \} \ll \Upsilon_{4,\min} \min \{ \sqrt{N} , \sqrt{T} \} ,      \\
& r_4^2 \kappa_4^4 \mu_4^2 \max \{ N \log^3 N, T \log^3 T \} \ll  \min \{ N^2 , T^2 \}.
\end{align*}
\end{assumption}

Then, we obtain the following entrywise convergence rate.

\begin{theorem}[Fully observed case]\label{thm:entry_convergence_rate_observed}
Suppose that Assumptions \ref{asp:noise}--\ref{asp:moment} and \ref{asp:entrywise_observed} hold. Then,
\begin{align*}
&\norm{\widehat{M} - M}_{\infty}  \\
&= O_p\left( \frac{  J \sqrt{d_1 d_2 \mu_\phi \mu_\psi}}{\sqrt{NT}} +\frac{r_2^{9/2} \mu_2^{5/2} \kappa_2 \sqrt{J d_1 \mu_{\phi}}  \sqrt{  \log (N + T)}}{\sqrt{N}} +  \frac{r_3^{9/2} \mu_3^{5/2} \kappa_3 \sqrt{J d_2 \mu_{\psi}}  \sqrt{  \log (N + T)}}{\sqrt{T}} \right.  \\
& \qquad  \quad  
+ \frac{J^{1/2} \max\{ \sqrt{d_1 \mu_{\phi}},\sqrt{d_2 \mu_{\psi}} \}\kappa_4^2 \mu_4 r_4 \theta \max \{\sqrt{\log N}, \sqrt{\log T} \}}{ \min\{\sqrt{N}, \sqrt{T} \}}  \\
& \left. \qquad  \quad + \left[ \frac{K_1}{J^{\gamma_1^G/d_1}} + \frac{K_1}{J^{\gamma_1^Q/d_2}} \right]  \right) ,
\end{align*}
where $\theta = \max \{ N , T \}/\min \{ N , T \}$.
\end{theorem}

The first, second, third, and fourth terms arise from estimating \(M_1\),
\(M_2\), \(M_3\), and \(M_4\), respectively, while the final term represents
the sieve approximation error. As in the preceding case, a key limitation of
this entrywise convergence rate is the signal-strength condition in
Assumption~\ref{asp:entrywise_observed}, which requires the smallest nonzero
singular value of each of \(M_2\), \(M_3\), and \(M_4\) to be sufficiently
large. As noted above, such signal-strength conditions are standard in the
literature on entrywise convergence rates for nuclear-norm-penalized
estimators and are generally difficult to avoid. 

The following theorem partially addresses the limitation of the preceding theorem.

\begin{theorem}[Fully observed case with weak signals]\label{thm:entry_convergence_rate_observed_weak_signal}
Suppose that Assumptions \ref{asp:noise}--\ref{asp:moment}. If $\kappa_2 \Upsilon_{2,\min} \lesssim \sqrt{T} $ holds instead of Assumption \ref{asp:entrywise_observed} (i), then we can replace the bound 
$$
\frac{r_2^{9/2} \mu_2^{5/2} \kappa_2 \sqrt{J d_1 \mu_{\phi}}  \sqrt{  \log (N + T)}}{\sqrt{N}} 
$$
in Theorem \ref{thm:entry_convergence_rate_observed} to $\norm{M_2}_\infty$, and if $\kappa_3 \Upsilon_{3,\min} \lesssim \sqrt{N} $ holds instead of Assumption \ref{asp:entrywise_observed} (ii), then we can replace the bound 
$$
\frac{r_3^{9/2} \mu_3^{5/2} \kappa_3 \sqrt{J d_2 \mu_{\psi}}  \sqrt{  \log (N + T)}}{\sqrt{T}}
$$
in Theorem \ref{thm:entry_convergence_rate_observed} to $\norm{M_3}_\infty$. In addition, if $\kappa_4 \Upsilon_{4,\min} \lesssim \sqrt{N+T} $ holds instead of Assumption \ref{asp:entrywise_observed} (iii), then we can replace the bound 
$$
\frac{J^{1/2} \max\{ \sqrt{d_1 \mu_{\phi}},\sqrt{d_2 \mu_{\psi}} \}\kappa_4^2 \mu_4 r_4 \theta \max \{\sqrt{\log N}, \sqrt{\log T} \}}{ \min\{\sqrt{N}, \sqrt{T} \}} 
$$
in Theorem \ref{thm:entry_convergence_rate_observed} to $\norm{M_4}_\infty$.
\end{theorem}
If all $M_{2}$, $M_3$, and $M_4$ are sufficiently weak, we have
$$
\norm{\widehat{M} - M}_{\infty}  
= O_p\left(  \frac{  J \sqrt{d_1 d_2 \mu_\phi \mu_\psi}}{\sqrt{NT}} + \norm{M_2}_\infty +  \norm{M_3}_\infty + \norm{M_4}_\infty  + \left[ \frac{K_1}{J^{\gamma_1^G/d_1}} + \frac{K_1}{J^{\gamma_1^Q/d_2}} \right]  \right) ,
$$
which can be faster than the rate obtained using standard low-rank matrix estimation. However, if the signal strength \(\Upsilon_{\iota,\min}\) falls into neither the strong-signal regime specified in Assumption~\ref{asp:entrywise_observed} nor the weak-signal regime specified in Theorem \ref{thm:entry_convergence_rate_observed_weak_signal}, our results do not provide an entrywise error bound. The entrywise convergence rate for the MNAR setting follows from Theorems~\ref{thm:entry_convergence_rate_observed} and \ref{thm:entry_convergence_rate_observed_weak_signal}, and is therefore omitted.

\section{Diminishing Noise Regime}\label{sec:diminishing_noise}

To study the estimation error of our estimator in the diminishing-noise regime,
where \(\sigma \to 0\), we make the dependence of the error bound on \(\sigma\)
explicit. To obtain a theoretically sharp bound when \(\sigma \to 0\), we set the regularization parameters as follows:
\begin{align*}
\nu_2
&=
C_2'
\left(
\sigma\sqrt{T}
+
\frac{\sqrt{NT}K_1}{J^{\gamma_1^Q/d_2}}
+
\frac{\sqrt{NT}K_2}{J^{\gamma_2^G/d_1}}
+
\frac{\sqrt{NT}K_3}{J^{\gamma_2^Q/d_2}}
\right), \\
\nu_3
&=
C_3'
\left(
\sigma\sqrt{N}
+
\frac{\sqrt{NT}K_1}{J^{\gamma_1^Q/d_2}}
+
\frac{\sqrt{NT}K_2}{J^{\gamma_2^G/d_1}}
+
\frac{\sqrt{NT}K_3}{J^{\gamma_2^Q/d_2}}
\right), \\
\nu_4
&=
C_4'
\left(
\sigma\sqrt{N+T}
+
\frac{\sqrt{NT}K_1}
     {J^{\gamma_1^G/d_1+\gamma_1^Q/d_2}}
+
\frac{\sqrt{NT}K_2}{J^{\gamma_2^G/d_1}}
+
\frac{\sqrt{NT}K_3}{J^{\gamma_2^Q/d_2}}
\right),
\end{align*}
where \(C_2',C_3',C_4'>0\) are some large constants. The terms
involving \(J^{-\gamma_1^G/d_1}\), \(J^{-\gamma_2^G/d_1}\),
\(J^{-\gamma_1^Q/d_2}\), and \(J^{-\gamma_2^Q/d_2}\) arise from sieve
approximation errors. Although the quantity in parentheses in the definition
of \(\nu_2\) can be bounded by the regularization level
\(\nu_2 = C_2\sqrt{T}\) used in Section~\ref{sec:model_est}, we retain the
noise-related and sieve approximation terms separately to highlight the
effect of diminishing noise. We apply the same refinement to \(\nu_3\) and
\(\nu_4\).

The following corollary provides an estimation error bound for the fully observed case.

\begin{corollary}[Extension of Theorem \ref{thm:convergence_rate_obs}]\label{corr:convergence_rate_obs_sigma}
Suppose that Assumptions \ref{asp:noise}, \ref{asp:basis}, \ref{asp:sieve} (i), and \ref{asp:moment} hold. Then,
\begin{align*}
&\|\widehat{M} - M\|_F \\
&= O_p\Bigg(
\sigma J \sqrt{d_1d_2}
+ \sqrt{r_2 + r_4}\,\min \Big\{ \sigma \sqrt{T} + \frac{\sqrt{NT}K_1}{J^{\gamma_1^Q/d_2}}
+
\frac{\sqrt{NT}K_2}{J^{\gamma_2^G/d_1}}
+
\frac{\sqrt{NT}K_3}{J^{\gamma_2^Q/d_2}}
,\, \norm{M_2 P_Z^\perp +  P_X M_4 P_Z^\perp}_F \Big\} \\
&\qquad\quad
+ \sqrt{r_3 + r_4}\,\min \Big\{ \sigma\sqrt{N}
+
\frac{\sqrt{NT}K_1}{J^{\gamma_1^Q/d_2}}
+
\frac{\sqrt{NT}K_2}{J^{\gamma_2^G/d_1}}
+
\frac{\sqrt{NT}K_3}{J^{\gamma_2^Q/d_2}} ,\,  \norm{ P_X^\perp M_3  +  P_X^\perp M_4 P_Z }_F \Big\}\\
&\qquad\quad + \sqrt{r_4}\,\min \Big\{ \sigma\sqrt{N+T}
+
\frac{\sqrt{NT}K_1}
     {J^{\gamma_1^G/d_1+\gamma_1^Q/d_2}}
+
\frac{\sqrt{NT}K_2}{J^{\gamma_2^G/d_1}}
+
\frac{\sqrt{NT}K_3}{J^{\gamma_2^Q/d_2}},\, \norm{ P_X^\perp M_4 P_Z^\perp }_F \Big\} \\
&\qquad\quad
+ \sqrt{NT}\left[ \frac{K_1}{J^{\gamma_1^G/d_1}} + \frac{K_1}{J^{\gamma_1^Q/d_2}} + \frac{K_2}{J^{\gamma_2^G/d_1}} + \frac{K_3}{J^{\gamma_2^Q/d_2}} \right]
\Bigg).
\end{align*}
\end{corollary}

The bound above shows that \(\sigma = 0\) alone is not sufficient to
guarantee exact recovery, because the sieve approximation errors may exist. However, if \(G\) and \(Q\) are sufficiently smooth, then the contribution of the sieve approximation errors will be very small.

Next, in the missing case, to obtain a theoretically sharp bound when \(\sigma \rightarrow 0\), we set the regularization parameter as follows:
$$
\nu = C'\left( \sigma p^{1/2} \sqrt{N+T} + \frac{(\log J)^{\frac{1}{10}}J\sqrt{d_1d_2}}{\sqrt{p}} +  \frac{\sqrt{NT} K_1}{J^{\gamma_1^G/d_1}} + \frac{\sqrt{NT} K_1}{J^{\gamma_1^Q/d_2}} + \frac{\sqrt{NT} K_2}{J^{\gamma_2^G/d_1}} + \frac{\sqrt{NT} K_3}{J^{\gamma_2^Q/d_2}} \right),
$$
for some large constant $C' > 0$. Then, the following corollary provides an estimation error bound for the MAR case.

\begin{corollary}[Extension of Theorem \ref{thm:convergence_rate_mar}]\label{corr:convergence_rate_mar_sigma}
Suppose that Assumptions \ref{asp:noise}, \ref{asp:basis}, \ref{asp:sieve} (i), \ref{asp:moment}, and \ref{asp:mar} hold. Then, if
\[
M_{\max} \geq \big\| M_{remain} \big\|_{\infty} \quad \text{ and } \quad  J \sqrt{d_1 d_2 } \ll p \sqrt{N+T},
\]
where $M_{remain} = M_2 P_Z^\perp+ P_X^\perp M_3 + M_4 - P_X M_4 P_Z$, we have
\begin{align*}
&\|\widehat{M} - M\|_F \\
&= O_p\Bigg(\sigma \frac{J\sqrt{d_1d_2}}{\sqrt{p}} + \frac{J\sqrt{d_1d_2}}{\sqrt{p}} \\
&+ \sqrt{r^*}\,
\min \left\{
 \frac{(\sigma + M_{\max}) \sqrt{N+T}}{\sqrt{p}} + \frac{(\log J)^{\frac{1}{10}}J\sqrt{d_1d_2}}{p^{3/2}} +  \frac{\sqrt{NT} K_1}{p J^{\gamma_1^G/d_1}} + \frac{\sqrt{NT} K_1}{p J^{\gamma_1^Q/d_2}} + \frac{\sqrt{NT} K_2}{p J^{\gamma_2^G/d_1}} + \frac{\sqrt{NT} K_3}{p J^{\gamma_2^Q/d_2}}, \right. \\
& \qquad  \qquad \qquad  \norm{M_{remain}}_F \Biggr\} \\
&
+ \sqrt{NT}
\left[
\frac{K_1}{J^{\gamma_1^G/d_1}} + \frac{K_1}{J^{\gamma_1^Q/d_2}} + \frac{K_2}{J^{\gamma_2^G/d_1}} + \frac{K_3}{J^{\gamma_2^Q/d_2}}
\right]
\Bigg),
\end{align*}
where \(r^* = \max\{r_2, r_3, r_4\}\).
\end{corollary}

As can be seen from the bound above, setting \(\sigma=0\) does not by itself
guarantee exact recovery. Unlike in the fully observed case, the bound
contains, in addition to the sieve approximation terms, an extra term of
order \( (1 + (\log J)^{\frac{1}{10}}\sqrt{r^\star}/p)\frac{J\sqrt{d_1d_2}}{\sqrt{p}} + \frac{ M_{\max} \sqrt{r^\star(N+T)}}{\sqrt{p}} \). This term mainly arises from the sampling fluctuation component \( \frac{1}{p} P_X\left( \left(\Omega-p\boldsymbol{1}\boldsymbol{1}^{\top}\right)\circ M \right)P_Z\), which is induced by the random observation pattern. Consequently, the estimation error of our method is not zero when $\sigma = 0$.

Lastly, we derive different estimation error bounds with the regularization parameters:
$$
\nu_2 = C_2'' \sigma \sqrt{T}, \qquad \nu_3 = C_3'' \sigma \sqrt{N}, \qquad \nu_4 = C_4'' \sigma \sqrt{N + T}, \qquad \nu = C'' \sigma p^{1/2} \sqrt{N + T},
$$
where \(C_2'',C_3'',C_4'',C''>0\) are some large constants.

\begin{corollary}[Fully observed case]\label{corr:convergence_rate_obs_sigma_2}
Suppose that the assumptions of Corollary \ref{corr:convergence_rate_obs_sigma} hold. Then,
\begin{align*}
&\|\widehat{M} - M\|_F \\
&= O_p\Bigg(
\sigma J \sqrt{d_1d_2}
+ \sqrt{r_2 + J}\,\min \Big\{ \sigma \sqrt{T} 
,\, \norm{M_2 P_Z^\perp +  P_X M_4 P_Z^\perp}_F + \frac{\sqrt{NT}K_1}{J^{\gamma_1^Q/d_2}}
+
\frac{\sqrt{NT}K_2}{J^{\gamma_2^G/d_1}}
+
\frac{\sqrt{NT}K_3}{J^{\gamma_2^Q/d_2}}\Big\} \\
&\qquad\quad
+ \sqrt{r_3 + J}\,\min \Big\{ \sigma\sqrt{N}
 ,\,  \norm{ P_X^\perp M_3  +  P_X^\perp M_4 P_Z }_F  +
\frac{\sqrt{NT}K_1}{J^{\gamma_1^Q/d_2}}
+
\frac{\sqrt{NT}K_2}{J^{\gamma_2^G/d_1}}
+
\frac{\sqrt{NT}K_3}{J^{\gamma_2^Q/d_2}} \Big\}\\
&\qquad\quad + \sqrt{r}\,\min \Big\{ \sigma\sqrt{N+T}
,\, \norm{ P_X^\perp M_4 P_Z^\perp }_F +
\frac{\sqrt{NT}K_1}
     {J^{\gamma_1^G/d_1+\gamma_1^Q/d_2}}
+
\frac{\sqrt{NT}K_2}{J^{\gamma_2^G/d_1}}
+
\frac{\sqrt{NT}K_3}{J^{\gamma_2^Q/d_2}}  \Big\} \\
&\qquad\quad
+ \sqrt{NT}\left[ \frac{K_1}{J^{\gamma_1^G/d_1}} + \frac{K_1}{J^{\gamma_1^Q/d_2}} + \frac{K}{J^{\gamma_2^G/d_1}} + \frac{K_3}{J^{\gamma_2^Q/d_2}} \right]
\Bigg).
\end{align*}
\end{corollary}
As we can see, when $\sigma = 0$, the sieve approximation term $\sqrt{NT}\left[ \frac{K_1}{J^{\gamma_1^G/d_1}} + \frac{K_1}{J^{\gamma_1^Q/d_2}} + \frac{K}{J^{\gamma_2^G/d_1}} + \frac{K_3}{J^{\gamma_2^Q/d_2}} \right]$ may remain.

\begin{corollary}[MAR case]\label{corr:convergence_rate_mar_sigma_2}
Suppose that the assumptions of Corollary \ref{corr:convergence_rate_mar_sigma} hold. Then, 
\begin{align*}
&\|\widehat{M} - M\|_F \\
&= O_p\Bigg(\sigma \frac{J\sqrt{d_1d_2}}{\sqrt{p}} + \frac{J\sqrt{d_1d_2}}{\sqrt{p}} \\
&+ \sqrt{r^* + J}\\
& \times \min \left\{
 \sigma p^{1/2} \sqrt{N+T} ,
\, \norm{M_{remain}}_F  + (\sigma + 1) \frac{J\sqrt{d_1d_2}}{\sqrt{p}}  +  \frac{\sqrt{NT} K_1}{J^{\gamma_1^G/d_1}} + \frac{\sqrt{NT} K_1}{J^{\gamma_1^Q/d_2}} + \frac{\sqrt{NT} K_2}{J^{\gamma_2^G/d_1}} + \frac{\sqrt{NT} K_3}{J^{\gamma_2^Q/d_2}}
\right\} \\
&
+ \sqrt{NT}
\left[
\frac{K_1}{J^{\gamma_1^G/d_1}} + \frac{K_1}{J^{\gamma_1^Q/d_2}} + \frac{K_2}{J^{\gamma_2^G/d_1}} + \frac{K_3}{J^{\gamma_2^Q/d_2}}
\right]
\Bigg),
\end{align*}
where \(r^* = \max\{r_2, r_3, r_4\}\).
\end{corollary}
When \(\sigma=0\), the error bound contains, in addition to the sieve
approximation terms, an extra term of order \(\frac{J\sqrt{d_1d_2}}{\sqrt{p}}\). Accordingly, we obtain conclusions analogous to those above. Corollaries
\ref{corr:convergence_rate_obs_sigma_2} and
\ref{corr:convergence_rate_mar_sigma_2} are primarily relevant to the noiseless case and provide little additional information when \(\sigma\) is bounded away from zero. In the more empirically relevant case where \(\sigma\geq c\) for some constant \(c>0\), the estimation error bounds derived in Sections~\ref{sec:asymp_obs} and~\ref{sec:ext_missing} suffice.

\section{Proof}

\subsection{Proofs of main results}

\subsubsection{Proof of Theorem \ref{thm:convergence_rate_obs}}

For $\iota \in \{ 1, 2\}$, define the sieve approximation error matrices as $R_{G_\iota} = G_\iota(X) - \Phi(X) B_\iota$ and $R_{Q_\iota} = Q_\iota(Z) - \Psi(Z) A_\iota$ where $A_\iota$ and $B_\iota$ are the sieve coefficient matrices consisting of $a_{\iota,k}$ and $b_{\iota,k}$ in Assumption \ref{asp:sieve}, respectively. First,
we note that
\begin{align}\label{eq:part_1}
\widehat{M}_1 &= P_X G_1 Q_1^\top P_Z + P_X G_2 V_1^\top P_Z + P_X W_1 Q_2^\top P_Z + P_X W_2 V_2^\top P_Z + P_X E P_Z \\
\nonumber&= G_1 Q_1^\top + G_2 V_1^\top P_Z + P_X W_1 Q_2^\top +  P_X W_2 V_2^\top P_Z  + P_X E P_Z + \textit{smoothing error\_1},
\end{align}
where
\begin{align*}
\textit{smoothing error\_1} &= -P_X^\perp R_{G_1} Q_1^\top - G_1 R_{Q_1}^\top P_Z^\perp  + P_X^\perp R_{G_1}R_{Q_1}^\top P_Z^\perp  -P_X^\perp R_{G_2} V_1^\top P_Z - P_X W_1 R_{Q_2}^\top P_Z^\perp .
\end{align*}
Then, because, by Assumption \ref{asp:sieve}, we have $\norm{R_{G_1}}_F = O_p \left( \frac{\sqrt{N K_1}}{J^{\gamma_1^G/d_1}}  \right)$, $\norm{R_{G_2}}_F = O_p \left( \frac{\sqrt{N K_2}}{J^{\gamma_1^G/d_1}}  \right)$, $\norm{R_{Q_1}}_F = O_p \left( \frac{\sqrt{T K_1}}{J^{\gamma_1^Q/d_2}}  \right)$, and $\norm{R_{Q_2}}_F = O_p \left( \frac{\sqrt{T K_3}}{J^{\gamma_2^Q/d_2}}  \right)$, we know
\begin{gather}\label{eq:smoothing_1}
\norm{\textit{smoothing error\_1}}_F = O_p\left( \frac{\sqrt{NT} K_1}{J^{\gamma_1^G/d_1}} + \frac{\sqrt{NT} K_1}{J^{\gamma_1^Q/d_2}} + \frac{\sqrt{NT} K_2}{J^{\gamma_2^G/d_1}} + \frac{\sqrt{NT} K_3}{J^{\gamma_2^Q/d_2}} \right) .    
\end{gather}
Next, for $\widehat{M}_2$, note that
\begin{align*}
P_X Y P_Z^\perp&= P_X G_1 Q_1^\top P_Z^\perp+ P_X G_2 V_1^\top P_Z^\perp+ P_X W_1 Q_2^\top P_Z^\perp\\
&\ \ + P_X W_2 V_2^\top P_Z^\perp+ P_X E P_Z^\perp\\
& =  G_2 V_1^\top  P_Z^\perp+  P_X W_2 V_2^\top P_Z^\perp + P_X E P_Z^\perp+ \textit{smoothing error\_2},
\end{align*}
where
$$
\textit{smoothing error\_2} = P_X G_1 R_{Q_1}^\top P_Z^\perp+ -P_X^\perp R_{G_2} V_1^\top P_Z^\perp+ P_X W_1 R_{Q_2}^\top P_Z^\perp.
$$
By Assumption \ref{asp:sieve}, we have
$$
\norm{\textit{smoothing error\_2}} = O_p\left( \frac{\sqrt{NT} K_1}{J^{\gamma_1^Q/d_2}}  + \frac{\sqrt{NT} K_2}{J^{\gamma_2^G/d_1}} + \frac{\sqrt{NT} K_3}{J^{\gamma_2^Q/d_2}} \right) = o_p\left(\sqrt{T} \right).
$$
Then, because $\norm{P_X E P_Z^\perp } \lesssim \sigma \sqrt{T}$ with high probability by Lemma \ref{lem:error_oper_norm}, we have 
$$
2 C_\lambda \norm{P_X E P_Z^\perp  + \textit{smoothing error\_2}} \leq \nu_2 = C_2 \sqrt{T}
$$
for some large $C_2>0$ with high probability, where $C_\lambda > 1$ is in Lemma \ref{lem:nuclear}. Hence, by setting $S = P_X E P_Z^\perp  + \textit{smoothing error\_2}$ and $L = M_2  P_Z^\perp+  P_X M_4 P_Z^\perp$, we can get by Lemma \ref{lem:nuclear} that
\begin{align}\label{eq:part_2}
\norm{\widehat{M}_2 - M_2  P_Z^\perp-  P_X M_4 P_Z^\perp}_F  = O_p\left( \sqrt{r_2 + r_4} \min\{ \nu_2, \norm{M_2  P_Z^\perp+  P_X M_4 P_Z^\perp}_F \} \right) .
\end{align}
For $\widehat{M}_3$, note that
\begin{align*}
P_X^\perp Y P_Z &=  P_X^\perp W_1 Q_2^\top  +  P_X^\perp W_2 V_2^\top P_Z  + P_X^\perp E P_Z + \textit{smoothing error\_3},
\end{align*}
where
$$
\textit{smoothing error\_3} = P_X^\perp R_{G_1} Q_1^\top P_Z + P_X^\perp R_{G_2} V_1^\top P_Z + P_X^\perp W_1 R_{Q_2}^\top P_Z^\perp .
$$
By Assumption \ref{asp:sieve}, we have
$$
\norm{\textit{smoothing error\_3}} = O_p\left(  \frac{\sqrt{NT} K_1}{J^{\gamma_1^Q/d_2}} + \frac{\sqrt{NT} K_2}{J^{\gamma_2^G/d_1}} + \frac{\sqrt{NT} K_3}{J^{\gamma_2^Q/d_2}} \right) = o_p\left(\sqrt{N} \right).
$$
Then, because $\norm{P_X^\perp E P_Z} \lesssim \sigma \sqrt{N}$ with high probability by Lemma \ref{lem:error_oper_norm}, we have 
$$
2 C_\lambda \norm{P_X^\perp E P_Z + \textit{smoothing error\_3}} \leq \nu_3 = C_3 \sqrt{N}
$$
for some large $C_3>0$ with high probability. Hence, by setting $S = P_X^\perp E P_Z + \textit{smoothing error\_3}$ and $L = P_X^\perp M_3  +  P_X^\perp M_4 P_Z $, we can derive by Lemma \ref{lem:nuclear} that
\begin{align}\label{eq:part_3}
\norm{\widehat{M}_3 - P_X^\perp M_3  -  P_X^\perp M_4 P_Z}_F  = O_p\left( \sqrt{r_3 + r_4} \min\{ \nu_3 , \norm{P_X^\perp M_3  +  P_X^\perp M_4 P_Z}_F \} \right) .
\end{align}
Lastly, for $\widehat{M}_4$, note that
\begin{align*}
P_X^\perp Y P_Z^\perp&=  P_X^\perp W_2 V_2^\top P_Z^\perp + P_X^\perp E P_Z^\perp+ \textit{smoothing error\_4},
\end{align*}
where
\begin{align*}
 \textit{smoothing error\_4} = P_X^\perp R_{G_1} R_{Q_1}^\top P_Z^\perp+ P_X^\perp R_{G_2} V_1^\top P_Z^\perp  + P_X^\perp W_1 R_{Q_2}^\top P_Z^\perp.   
\end{align*}
Then, by Assumption \ref{asp:sieve}, we have
$$
\norm{\textit{smoothing error\_4}} = O_p\left(  \frac{\sqrt{NT} K_1}{J^{(\gamma_1^G/d_1 + \gamma_1^Q/d_2)}} + \frac{\sqrt{NT} K_2}{J^{\gamma_2^G/d_1}} + \frac{\sqrt{NT} K_3}{J^{\gamma_2^Q/d_2}} \right) = o_p\left(\sqrt{N} + \sqrt{T} \right).
$$
Then, because $\norm{P_X^\perp E P_Z^\perp} \lesssim \sigma (\sqrt{N} + \sqrt{T})$ with high probability by Lemma \ref{lem:error_oper_norm}, we have $2 C_\lambda \norm{P_X^\perp E P_Z^\perp+ \textit{smoothing error\_4}} \leq \nu_4 = C_4 \sqrt{N+T}$ for some large $C_4>0$ with high probability. Hence, by setting $S = P_X^\perp E P_Z^\perp+ \textit{smoothing error\_4}$ and $L = P_X^\perp M_4 P_Z^\perp$, we can derive by Lemma \ref{lem:nuclear} that
\begin{align}\label{eq:part_4}
\norm{\widehat{M}_4 - P_X^\perp M_4 P_Z^\perp}_F = O_p\left( \sqrt{r_4} \min\{ \nu_4, \norm{ P_X^\perp M_4 P_Z^\perp}_F \} \right) .
\end{align}
From the relation \eqref{eq:part_1}, we have
\begin{align*}
\widehat{M} - M &= \widehat{M}_1 + \widehat{M}_2 + \widehat{M}_3 + \widehat{M}_4 - \left( M_1 + M_2 + M_3 + M_4 \right) \\
&= \widehat{M}_2 - M_2  P_Z^\perp-  P_X M_4 P_Z^\perp
+ \widehat{M}_3 - P_X^\perp M_3  -  P_X^\perp M_4 P_Z \\
&\ \ + \widehat{M}_4 - P_X^\perp M_4 P_Z^\perp
+ P_X E P_Z + \textit{smoothing error\_1}.
\end{align*}
Hence, we have from the bounds \eqref{eq:smoothing_1}, \eqref{eq:part_2}, \eqref{eq:part_3}, and \eqref{eq:part_4} with Lemma \ref{lem:error_oper_norm} that
\begin{align*}
\norm{\widehat{M} - M}_F  
&= O_p\left( \sigma J\sqrt{d_1d_2} + \sqrt{r_2 + r_4}   \min \left\{ \nu_2 , \norm{M_2 P_Z^\perp +  P_X M_4 P_Z^\perp}_F \right\} \right. \\
& \qquad  \quad   + \sqrt{r_3 + r_4}   \min \left\{ \nu_3 ,   \norm{ P_X^\perp M_3  +  P_X^\perp M_4 P_Z }_F \right\}  
+ \sqrt{r_4}   \min \left\{ \nu_4 , \norm{ P_X^\perp M_4 P_Z^\perp }_F  \right\} \\
& \left. \qquad  \quad + \sqrt{NT} \left[ \frac{K_1}{J^{\gamma_1^G/d_1}} + \frac{K_1}{J^{\gamma_1^Q/d_2}} + \frac{K_2}{J^{\gamma_2^G/d_1}} + \frac{K_3}{J^{\gamma_2^Q/d_2}} \right]  \right) .
\end{align*}
Here, we derive
$$
\norm{P_X E P_Z}_F \leq \sqrt{rank(P_X E P_Z)} \norm{P_X E P_Z} \lesssim \sqrt{ J \min\{d_1,d_2 \} } \sigma \sqrt{J\max\{d_1,d_2 \}}
$$
by Lemma \ref{lem:error_oper_norm}. $\square$

\subsubsection{Proof of Corollary \ref{cor:singular_vector}}

It is easily derived from the Davis-Kahan theorem (see, e.g., Corollary 2.8 and Theorem 2.9 of \cite{chen2021spectral}). $\square$

\subsubsection{Proof of Theorem \ref{thm:convergence_rate_mar}}

First, note that
\begin{align}\label{eq:M_1_est}
\widehat{M}_1 &= \frac{1}{\widehat{p}} P_X (\Omega \circ M) P_Z + \frac{1}{\widehat{p}} P_X (\Omega \circ E) P_Z\\
\nonumber& = P_X  M P_Z + \left( \frac{p}{\widehat{p}} - 1\right) P_X  M P_Z + \frac{1}{\widehat{p}} P_X (( \Omega - p \boldsymbol{1}\boldsymbol{1}^\top )\circ M) P_Z + \frac{1}{\widehat{p}} P_X (\Omega \circ E) P_Z \\
\nonumber& = G_1 Q_1^\top + G_2 V_1^\top P_Z + P_X W_1 Q_2^\top +  P_X W_2 V_2^\top P_Z + \textit{smoothing error\_1}\\
\nonumber&\ \ + \left( \frac{p}{\widehat{p}} - 1\right) P_X  M P_Z  + \frac{1}{p} P_X (( \Omega - p \boldsymbol{1}\boldsymbol{1}^\top )\circ M) P_Z + \frac{1}{\widehat{p}} P_X (\Omega \circ E) P_Z  ,
\end{align}
where $\textit{smoothing error\_1}$ is defined in \eqref{eq:part_1}. Because
$$
1 - \frac{p}{\widehat{p}}  = \widehat{p}^{-1} \left( \frac{1}{NT} \sum_{it} \omega_{it} - p \right) = O_p \left( p^{-1} \frac{p^{1/2}}{\sqrt{NT}} \right) = O_p \left( \frac{1}{\sqrt{pNT}} \right)
$$
and $\norm{M}_F = O_p (\sqrt{NT} )$, we have $\norm{\left( \frac{p}{\widehat{p}} - 1\right) P_X  M P_Z }_F =  O_p \left( \frac{1}{\sqrt{p}} \right)$.
Note that
$$
\norm{ \frac{1}{p} P_X (( \Omega - p \boldsymbol{1}\boldsymbol{1}^\top )\circ M) P_Z }_F \leq  \frac{1}{p} \norm{\Phi (\Phi^\top \Phi)^{-1}} \norm{\Phi^\top (( \Omega - p \boldsymbol{1}\boldsymbol{1}^\top )\circ M) \Psi}_F \norm{\Psi (\Psi^\top \Psi)^{-1}}.
$$
By Assumption \ref{asp:basis}, we know $\norm{\Phi (\Phi^\top \Phi)^{-1}} \lesssim \frac{1}{\sqrt{N}}$ and $\norm{\Psi (\Psi^\top \Psi)^{-1}} \lesssim \frac{1}{\sqrt{T}}$. In addition, because
\begin{align*}
\bbE\left[ \left. \norm{\Phi^\top (( \Omega - p \boldsymbol{1}\boldsymbol{1}^\top )\circ M) \Psi}_F^2 \right|  M, X, Z \right] &= 
 \sum_{j_1,j_2} \bbE\left[ \left. \left(\sum_{it}(\omega_{it} - p) m_{it} \phi_{i,j_1} \psi_{t,j_2}\right)^2  \right|  M, X, Z \right] \\
&= \sum_{j_1,j_2} \sum_{it}\bbE[(\omega_{it} - p)^2 ] m^2_{it} \phi^2_{i,j_1} \psi^2_{t,j_2} \\
& \leq p \sum_{j_1,j_2} \sum_{it} m^2_{it} \phi^2_{i,j_1} \psi^2_{t,j_2}
 = p \sum_{it} m^2_{it} \norm{\phi_i}^2 \norm{\psi_t}^2 \\
& = O_p(p N T J^2 d_1 d_2)
\end{align*}
by the assumption that $\bbE[m_{it}^4]$, $\bbE[\phi_{ij}^4]$, and $\bbE[\psi_{tj}^4]$ are bounded, we have $ \norm{\Phi^\top (( \Omega - p \boldsymbol{1}\boldsymbol{1}^\top )\circ M) \Psi}_F = O_p( \sqrt{p N T} J\sqrt{d_1d_2})$. Hence, we have 
$$
\norm{ \frac{1}{p} P_X (( \Omega - p \boldsymbol{1}\boldsymbol{1}^\top )\circ M) P_Z }_F = O_p\left( \frac{J\sqrt{d_1d_2}}{\sqrt{p}} \right) .
$$
Moreover, note that
$$
\norm{ \frac{1}{\widehat{p}} P_X (\Omega \circ E)  P_Z }_F \leq  \frac{1}{\widehat{p}} \norm{\Phi (\Phi^\top \Phi)^{-1}} \norm{\Phi^\top (\Omega \circ E) \Psi}_F \norm{\Psi (\Psi^\top \Psi)^{-1}}.
$$
By Assumption \ref{asp:basis}, we know $\norm{\Phi (\Phi^\top \Phi)^{-1}} \lesssim \frac{1}{\sqrt{N}}$ and $\norm{\Psi (\Psi^\top \Psi)^{-1}} \lesssim \frac{1}{\sqrt{T}}$. In addition, because
\begin{align*}
\bbE\left[ \left. \norm{\Phi^\top (\Omega \circ E) \Psi}_F^2 \right|  M, X, Z \right] &= 
 \sum_{j_1,j_2} \bbE\left[ \left. \left(\sum_{it}\omega_{it}\epsilon_{it}  \phi_{i,j_1} \psi_{t,j_2}\right)^2  \right|  M, X, Z \right] \\
&= \sum_{j_1,j_2} \sum_{it}\bbE[(\omega_{it}\epsilon_{it})^2 ] \phi^2_{i,j_1} \psi^2_{t,j_2} \\
& \leq p \sigma^2 \sum_{j_1,j_2} \sum_{it} \phi^2_{i,j_1} \psi^2_{t,j_2}
 = p  \sigma^2 \sum_{it} \norm{\phi_i}^2 \norm{\psi_t}^2 \\
& = O_p(p\sigma^2  N T J^2 d_1 d_2)
\end{align*}
by the assumption that $\bbE[\phi_{ij}^4]$, and $\bbE[\psi_{tj}^4]$ are bounded, we have $ \norm{\Phi^\top (\Omega \circ E) \Psi}_F = O_p( \sigma \sqrt{p N T} J\sqrt{d_1d_2})$ and 
$$
\norm{ \frac{1}{\widehat{p}} P_X (\Omega \circ E)  P_Z }_F = O_p\left( \frac{ \sigma  J\sqrt{d_1 d_2}}{\sqrt{p}} \right).
$$
In addition, from \eqref{eq:M_1_est}, we have
\begin{align*}
Y - \widehat{M}_1 &= G_2 V_1^\top P_Z^\perp+ P_X^\perp W_1 V_2^\top + W_2 V_2^\top - P_X W_2 V_2^\top P_Z+ E\\
& \ \  - \left( \frac{p}{\widehat{p}} - 1\right) P_X  M P_Z- \frac{1}{p} P_X (( \Omega - p \boldsymbol{1}\boldsymbol{1}^\top )\circ M) P_Z - \frac{1}{\widehat{p}} P_X (\Omega \circ E) P_Z  - \textit{smoothing error\_1}.
\end{align*}
Note that, under our assumptions, $\norm{\left( \frac{p}{\widehat{p}} - 1\right) P_X  M P_Z }_F$, $\norm{ \frac{1}{p} P_X (( \Omega - p\boldsymbol{1}\boldsymbol{1}^\top )\circ M) P_Z }_F$, $\norm{ \frac{1}{\widehat{p}} P_X (\Omega \circ E)  P_Z }_F$, and $\norm{\textit{smoothing error\_1}}_F$ are $o_p(p^{1/2} \sqrt{N+T})$. Hence, we have
$$
\norm{\Omega \circ S'}  \leq \norm{\Omega \circ S'}_F \leq \norm{S'}_F = o_p\left(p^{1/2} \sqrt{N+T}\right)
$$
where $S' = \left( \frac{p}{\widehat{p}} - 1\right) P_X  M P_Z + \frac{1}{p} P_X (( \Omega - p \boldsymbol{1}\boldsymbol{1}^\top )\circ M) P_Z + \frac{1}{\widehat{p}} P_X (\Omega \circ E) P_Z  + \textit{smoothing error\_1}$. In addition, we have $\norm{\Omega \circ E} \lesssim \sigma p^{1/2}  \sqrt{N+T}$ with high probability by Lemma \ref{lem:orthogonal_oper_norm} since $\norm{\omega_{it} \epsilon_{it}}_{\psi_2} \leq p^{1/2} \sigma$. Hence, $2 C_\lambda \norm{\Omega \circ (E - S')} \leq \nu = C p^{1/2}  \sqrt{N+T}$ for some large $C > 0$ with high probability, where $C_\lambda > 1$ is in Lemma \ref{lem:nuclear_missing}. Then, by setting $S = E - S'$ and $L = M_2 P_Z^\perp+ P_X^\perp M_3 + M_4 - P_X M_4 P_Z$, we can get from Lemma \ref{lem:nuclear_missing} that
\begin{align*}
&\norm{\widehat{M}_{rest} - (M_2 P_Z^\perp+ P_X^\perp M_3 + M_4 - P_X M_4 P_Z)}_F  \\
&= O_p\left( \sqrt{r^*}   \min \left\{ \frac{\nu}{p} + \frac{\sqrt{N + T} M_{\max}}{\sqrt{p}} , \norm{M_2 P_Z^\perp+ P_X^\perp M_3 + M_4 - P_X M_4 P_Z}_F \right\} \right)    .
\end{align*}
Then, because
\begin{align*}
\norm{\widehat{M}_{1} + \widehat{M}_{rest} - M}_F &\leq
\norm{\widehat{M}_{rest} - (M_2 P_Z^\perp+ P_X^\perp M_3 + M_4 - P_X M_4 P_Z)}_F \\
&\ \ + \norm{\left( \frac{p}{\widehat{p}} - 1\right) P_X  M P_Z  +  \frac{1}{p} P_X (( \Omega - p \boldsymbol{1}\boldsymbol{1}^\top )\circ M) P_Z + \frac{1}{\widehat{p}} P_X (\Omega \circ E) P_Z  }_F \\
&\ \ + \norm{\textit{smoothing error\_1}}_F,
\end{align*}
and
\begin{align*}
&\norm{ \left( \frac{p}{\widehat{p}} - 1\right) P_X  M P_Z  + \frac{1}{p} P_X (( \Omega - p \boldsymbol{1}\boldsymbol{1}^\top )\circ M) P_Z + \frac{1}{\widehat{p}} P_X (\Omega \circ E) P_Z  }_F = O_p\left( \frac{J\sqrt{d_1d_2}}{\sqrt{p}} \right),\\
&\norm{\textit{smoothing error\_1}}_F = O_p\left( \frac{\sqrt{NT} K_1}{J^{\gamma_1^G/d_1}} + \frac{\sqrt{NT} K_1}{J^{\gamma_1^Q/d_2}} + \frac{\sqrt{NT} K_2}{J^{\gamma_2^G/d_1}} + \frac{\sqrt{NT} K_3}{J^{\gamma_2^Q/d_2}} \right),
\end{align*}
as noted in \eqref{eq:smoothing_1}, we have the desired result. $\square$

\subsubsection{Proof of Theorem \ref{thm:convergence_rate_mnar}}

For notational simplicity, denote the subscripts `$tall$' and `$wide$' by $\calT$ and $\calW$. Here, $(U_\calT,D_\calT,V_\calT)$, $(U_\calW,D_\calW,V_\calW)$, and $(U,D,V)$ mean the SVD of `tall', `wide', and `full' matrices, respectively. First, by applying Corollary \ref{cor:singular_vector} to the tall and wide matrices, respectively, we have
$$
\norm{\widehat{U}_{\calW} O_{\calW} - U_{\calW}} = O_p\left( \frac{\calR_{\calW}}{\lambda_{\min,\calW}}\right), \quad \norm{\widehat{U}_{\calT} O_{\calT} - U_{\calT}} = O_p\left( \frac{\calR_{\calT}}{\lambda_{\min,\calT}}\right) .
$$
In addition, by Lemma \ref{lem:tech_MNAR}, we have $ U_{\calW} = U_{(N_0)} H_{\calW}$ where $U_{(N_0)} = [u_1, \cdots, u_{N_0}]^\top$ and $H_{\calW} = (U_{(N_0)}^\top U_{(N_0)})^{-1/2} G_{\calW}$ for some $r \times r$ orthogonal matrix $G_{\calW}$. Then, we have by Lemma \ref{lem:tech_MNAR} that
\begin{gather*}
\norm{\widehat{U}_{\calW} O_{\calW} - U_{\calW}}_F = \norm{\widehat{U}_{\calW} O_{\calW} - U_{(N_0)} H_{\calW}}_F = \norm{\widehat{U}_{\calW}  - U_{(N_0)} Q_{\calW}^{-1} }_F 
= O_p\left( \frac{\calR_{\calW}}{\lambda_{\min,\calW}}\right),\\
\norm{\widehat{U}_{\calW}Q_{\calW}  - U_{(N_0)}}_F = \norm{\widehat{U}_{\calW}  - U_{(N_0)} Q_{\calW}^{-1} }_F \norm{Q_{\calW}} = O_p\left( \frac{\sqrt{N_0}}{\sqrt{N}} \frac{\calR_{\calW}}{\lambda_{\min,\calW}}\right) ,
\end{gather*}
where $Q_{\calW}^{-1} = H_{\calW}O_{\calW}^\top$. Similarly, by Lemma \ref{lem:tech_MNAR}, we have $ U_{\calT} = U H_{\calT}$ where $H_{\calT} = (U^\top U)^{-1/2} G_{\calT}$ for some $r \times r$ orthogonal matrix $G_{\calT}$. Then we have
\begin{gather*}
\norm{\widehat{U}_{\calT} O_{\calT} - U_{\calT}}_F = \norm{\widehat{U}_{\calT} O_{\calT} - U H_{\calT}}_F = \norm{\widehat{U}_{\calT}  - U Q_{\calT}^{-1} }_F 
= O_p\left( \frac{\calR_{\calT}}{\lambda_{\min,\calT}}\right),\\
\norm{\widehat{U}_{\calT} Q_{\calT}  - U }_F = \norm{\widehat{U}_{\calT}  - U Q_{\calT}^{-1} }_F \norm{Q_{\calT}} = O_p\left( \frac{\calR_{\calT}}{\lambda_{\min,\calT}}\right),
\end{gather*}
where $Q_{\calT}^{-1} = H_{\calT} O_{\calT}^\top$. Define $R_1 = \widehat{U}_{\calW} Q_{\calW} - U_{(N_0)}$ and $R_2 = \widehat{U}_{\calT,(N_0)} Q_{\calT} - U_{(N_0)}$ where $U_{\calT,(N_0)} = [u_{\calT,1}, \cdots, u_{\calT,N_0}]^\top$ and $\widehat{U}_{\calT,(N_0)} = [\widehat{u}_{\calT,1}, \cdots, \widehat{u}_{\calT,N_0}]^\top$. Then, we have 
$$
\widehat{U}_{\calW}Q_{\calW} - R_1 = \widehat{U}_{\calT,(N_0)} Q_{\calT} - R_2 \quad \Longrightarrow \quad 
\widehat{U}_{\calW} = \widehat{U}_{\calT,(N_0)} H_{adj} + R_1 Q_{\calW}^{-1} - R_2 Q_{\calW}^{-1},
$$
where $H_{adj} = Q_{\calT} Q_{\calW}^{-1}$. Hence, we have
\begin{align*}
\widehat{H}_{adj} - H_{adj} &= \left( \widehat{U}_{\calT,(N_0)}^\top \widehat{U}_{\calT,(N_0)} \right)^{-1} \widehat{U}_{\calT,(N_0)}^\top \widehat{U}_{\calW}  - H_{adj}\\
& = \left( \widehat{U}_{\calT,(N_0)}^\top \widehat{U}_{\calT,(N_0)} \right)^{-1} \widehat{U}_{\calT,(N_0)}^\top R_1 Q_{\calW}^{-1} -  \left( \widehat{U}_{\calT,(N_0)}^\top \widehat{U}_{\calT,(N_0)} \right)^{-1} \widehat{U}_{\calT,(N_0)}^\top  R_2 Q_{\calW}^{-1}.
\end{align*}
Then, since 
$$
\norm{\left( \widehat{U}_{\calT,(N_0)}^\top \widehat{U}_{\calT,(N_0)} \right)^{-1} \widehat{U}_{\calT,(N_0)}^\top} = \norm{\left(\widehat{U}_{\calT,(N_0)}^\top \widehat{U}_{\calT,(N_0)}\right)^{-1}}^{1/2} = O_p\left(\frac{\sqrt{N}}{\sqrt{N_0}}\right)
$$
by Lemma \ref{lem:tech_MNAR} and $\norm{R_1} = O_p\left( \frac{\sqrt{N_0}}{\sqrt{N}} \frac{\calR_{\calW}}{\lambda_{\min,\calW}}\right)$, $\norm{R_2} = O_p\left( \frac{\calR_{\calT}}{\lambda_{\min,\calT}}\right)$ by the above bounds, we have by Lemma \ref{lem:tech_MNAR} that
\begin{align*}
\norm{\widehat{H}_{adj} - H_{adj}} = O_p\left( \frac{\sqrt{N}}{\sqrt{N_0}} \frac{\calR_{\calW}}{\lambda_{\min,\calW}} +  \frac{N}{N_0}  \frac{\calR_{\calT}}{\lambda_{\min,\calT}} \right) .
\end{align*}
In addition, because $\norm{H_{adj}} = O_p\left( \frac{\sqrt{N}}{\sqrt{N_0}} \right)$ by Lemma \ref{lem:tech_MNAR}, we have
$$
\norm{\widehat{H}_{adj}} \leq \norm{H_{adj}} + \norm{\widehat{H}_{adj} - H_{adj}} = O_p\left( \frac{\sqrt{N}}{\sqrt{N_0}} \right).
$$
Moreover, note that
\begin{align*}
\norm{\widehat{V}_{\calW}\widehat{D}_{\calW} - V_{\calW} D_{\calW} O_{\calW}^\top }_F 
&= \norm{\widehat{D}_{\calW} \widehat{V}_{\calW}^\top - O_{\calW} D_{\calW} V_{\calW}^\top }_F 
= \norm{\widehat{U}_{\calW}^\top \widehat{M}_{\calW} - O_{\calW} 
 U_{\calW}^\top M_{\calW} }_F \\
& = \norm{\widehat{U}_{\calW}^\top \left( \widehat{M}_{\calW} - M_{\calW} \right) + \left( \widehat{U}_{\calW}^\top - O_{\calW} U_{\calW}^\top \right) M_{\calW} }_F \\
&\leq  \norm{\widehat{M}_{\calW} - M_{\calW}}_F + \norm{ \widehat{U}_{\calW}^\top - O_{\calW} U_{\calW}^\top}_F \norm{M_{\calW}} \\
& = O_p\left( \calR_\calW + \frac{\lambda_{\max,\calW}}{\lambda_{\min,\calW}}\calR_\calW \right) .
\end{align*}
Then, because 
$$
D_\calW V_\calW^\top = U_\calW^\top M_\calW =  U_\calW^\top (U_{(N_0)} H_\calW ) H_\calW^{-1} D V^\top =  U_\calW^\top U_\calW H_\calW^{-1} D V^\top = H_\calW^{-1} D V^\top,
$$
we have
\begin{align*}
\norm{\widehat{V}_{\calW}\widehat{D}_{\calW} - V D Q_{\calW}^\top }_F = \norm{\widehat{V}_{\calW}\widehat{D}_{\calW} - V_{\calW} D_{\calW} O_{\calW}^\top }_F =  O_p\left( \calR_\calW + \frac{\lambda_{\max,\calW}}{\lambda_{\min,\calW}}\calR_\calW \right).
\end{align*}
Lastly, we have the following decomposition:
\begin{align*}
\norm{\widehat{M} - M}_F
&= \norm{\widehat{U}_{\calT} \widehat{H}_{adj} \widehat{D}_{\calW} \widehat{V}_{\calW}^\top - UDV^\top}_F \\
&\lesssim \norm{U Q_{\calT}^{-1} \widehat{H}_{adj} \left( \widehat{V}_{\calW}\widehat{D}_{\calW} - V D Q_{\calW}^\top\right)^\top }_F 
+ \norm{\left(\widehat{U}_\calT - U Q_\calT^{-1} \right) \widehat{H}_{adj} Q_{\calW} D V^\top }_F \\
& \ \ + \norm{ U \left( Q_\calT^{-1} \widehat{H}_{adj} Q_\calW - I_r  \right) D V^\top  }_F.
\end{align*}
The first term can be bounded like
\begin{align*}
\norm{U Q_{\calT}^{-1} \widehat{H}_{adj} \left( \widehat{V}_{\calW}\widehat{D}_{\calW} - V D Q_{\calW}^\top\right)^\top }_F 
\leq \norm{Q_{\calT}^{-1}} \norm{\widehat{H}_{adj}} \norm{\widehat{V}_{\calW}\widehat{D}_{\calW} - V D Q_{\calW}^\top}_F 
 = O_p \left( \frac{\sqrt{N}}{\sqrt{N_0}} \kappa_\calW \calR_\calW \right),
\end{align*}
where $\kappa_\calW = \frac{\lambda_{\max,\calW}}{\lambda_{\min,\calW}}$. The second term can be bounded like
\begin{align*}
\norm{\left(\widehat{U}_\calT - U Q_\calT^{-1} \right) \widehat{H}_{adj} Q_{\calW} D V^\top }_F
\leq \norm{\widehat{U}_\calT - U Q_\calT^{-1} } \norm{\widehat{H}_{adj}} \norm{Q_{\calW}} \norm{D} = O_p\left(  \frac{\lambda_{\min}}{\lambda_{\min,\calT}}\kappa \calR_\calT \right).
\end{align*}
In addition, the last term can be bounded like
\begin{align*}
\norm{ U \left( Q_\calT^{-1} \widehat{H}_{adj} Q_\calW - I_r  \right) D V^\top  }_F 
& = \norm{ U Q_\calT^{-1} \left( \widehat{H}_{adj}  - H_{adj} \right) Q_\calW D V^\top  }_F \\
& \leq \norm{Q_\calT^{-1}} \norm{\widehat{H}_{adj}  - H_{adj}} \norm{Q_\calW} \norm{D}_F \\
& = O_p\left( \kappa \frac{\lambda_{\min}}{\lambda_{\min,\calW}} \calR_{\calW} +  \kappa \frac{\sqrt{N}}{\sqrt{N_0}}   \frac{\lambda_{\min}}{\lambda_{\min,\calT}}\calR_{\calT} \right) .
\end{align*}
Moreover, by Lemma F.1 of \cite{choi2024matrix}, we have $\kappa_\calW \lesssim \kappa$, $\frac{\lambda_{\min}}{\lambda_{\min,\calT}} \asymp \frac{\sqrt{T}}{\sqrt{T_0}}$, and $\frac{\lambda_{\min}}{\lambda_{\min,\calW}} \asymp \frac{\sqrt{N}}{\sqrt{N_0}}$. Therefore, to sum up, we have the desired result. $\square$

\subsection{Proofs of results in Appendix}

\subsubsection{Proof of Theorem \ref{thm:convergence_rate_mar_hetero}}

First, note that
\begin{align}\label{eq:M1_hetro}
\widehat{M}_1 &=  P_X \widehat{\Pi}^{-1} (\Omega \circ M) P_Z + P_X \widehat{\Pi}^{-1}(\Omega \circ E) P_Z\\
&\nonumber = P_X \Pi^{-1} (\Omega \circ M) P_Z + P_X ( \widehat{\Pi}^{-1} - \Pi^{-1} ) (\Omega \circ M) P_Z + P_X \Pi^{-1} (\Omega \circ E) P_Z + P_X ( \widehat{\Pi}^{-1} - \Pi^{-1} ) (\Omega \circ E) P_Z\\
&\nonumber = G_1 Q_1^\top + G_2 V_1^\top P_Z + P_X W_1 Q_2^\top +  P_X W_2 V_2^\top P_Z + \textit{smoothing error\_1}\\
&\nonumber \ \   +  P_X \Pi^{-1}(( \Omega -\Pi \boldsymbol{1}\boldsymbol{1}^\top )\circ M) P_Z +  P_X ( \widehat{\Pi}^{-1} - \Pi^{-1} ) (\Omega \circ M) P_Z + P_X \Pi^{-1} (\Omega \circ E) P_Z\\
&\nonumber \ \ + P_X ( \widehat{\Pi}^{-1} - \Pi^{-1} ) (\Omega \circ E) P_Z  ,
\end{align}
where $\textit{smoothing error\_1}$ is defined in \eqref{eq:part_1}. Note that
$$
\norm{ P_X \Pi^{-1}(( \Omega -\Pi \boldsymbol{1}\boldsymbol{1}^\top )\circ M) P_Z }_F \leq  \frac{1}{p} \norm{\Phi (\Phi^\top \Phi)^{-1}} 
\norm{\Phi^\top (\Pi^{-1}(( \Omega -\Pi \boldsymbol{1}\boldsymbol{1}^\top )\circ M)) \Psi}_F 
\norm{\Psi (\Psi^\top \Psi)^{-1}}.
$$
By Assumption \ref{asp:basis}, we know $\norm{\Phi (\Phi^\top \Phi)^{-1}} \lesssim \frac{1}{\sqrt{N}}$ and $\norm{\Psi (\Psi^\top \Psi)^{-1}} \lesssim \frac{1}{\sqrt{T}}$. In addition, because
\begin{align*}
\bbE\left[ \left.\norm{\Phi^\top (\Pi^{-1}(( \Omega -\Pi \boldsymbol{1}\boldsymbol{1}^\top )\circ M)) \Psi}_F^2 \right|  M, X, Z \right] &= 
 \sum_{j_1,j_2} \bbE\left[ \left. \left( \sum_{it} p_i^{-1} (\omega_{it} - p_i) m_{it} \phi_{i,j_1} \psi_{t,j_2} \right)^2  \right|  M, X, Z \right] \\
&= \sum_{j_1,j_2} \sum_{it} p_i^{-2} \bbE[(\omega_{it} - p_i)^2 ] m^2_{it} \phi^2_{i,j_1} \psi^2_{t,j_2} \\
& = p_{\min}^{-1} \sum_{j_1,j_2} \sum_{it} m^2_{it} \phi^2_{i,j_1} \psi^2_{t,j_2}
 = p_{\min}^{-1} \sum_{it} m^2_{it} \norm{\phi_i}^2 \norm{\psi_t}^2 \\
& = O_p(p_{\min}^{-1} N T J^2 d_1 d_2)
\end{align*}
by the assumption that $\bbE[m_{it}^4]$, $\bbE[\phi_{ij}^4]$, and $\bbE[\psi_{tj}^4]$ are bounded, we have 
$$
\norm{ P_X \Pi^{-1}(( \Omega -\Pi \boldsymbol{1}\boldsymbol{1}^\top )\circ M) P_Z }_F = O_p\left( \frac{J\sqrt{d_1d_2}}{\sqrt{p_{\min}}} \right) .
$$
In addition, by the similar token, we have $\norm{P_X \Pi^{-1} (\Omega \circ E) P_Z}_F = O_p\left( \frac{\sigma J\sqrt{d_1d_2}}{\sqrt{p_{\min}}}\right)$.\\
Next, note that
$$
P_X ( \widehat{\Pi}^{-1} - \Pi^{-1} ) (\Omega \circ M) P_Z 
= P_X ( \widehat{\Pi}^{-1} - \Pi^{-1} ) (( \Omega -\Pi \boldsymbol{1}\boldsymbol{1}^\top ) \circ M) P_Z + P_X ( \widehat{\Pi}^{-1} - \Pi^{-1} ) (\Pi \boldsymbol{1}\boldsymbol{1}^\top \circ M) P_Z 
$$
First, by the following relation
$$
\widehat{p}_i^{-1} - p_{i}^{-1} = p_{i}^{-2} \left( \widehat{p}_i -  p_{i} \right) + p_{i}^{-2} \widehat{p}_i^{-1} \left( \widehat{p}_i -  p_{i} \right)^2  ,
$$
the dominating term of $\norm{\Phi^\top ( \widehat{\Pi}^{-1} - \Pi^{-1} ) (( \Omega -\Pi \boldsymbol{1}\boldsymbol{1}^\top ) \circ M)\Psi }_F^2$ is
\begin{align*}
&\sum_{j_1,j_2}\left( \sum_{i,t} p_i^{-2} (\widehat{p_i} - p_i)(\omega_{it} - p_i ) m_{it} \phi_{i,j_1} \psi_{t,j_2} \right)^2 
=\sum_{j_1,j_2}\left( \frac{1}{T}\sum_{i,s,t} p_i^{-2} (\omega_{is} - p_i)(\omega_{it} - p_i ) m_{it} \phi_{i,j_1} \psi_{t,j_2} \right)^2      \\
&\ \  =  \sum_{j_1,j_2}\left( \sum_{i} p_i^{-2}  \phi_{i,j_1} \left( \frac{1}{\sqrt{T}} \sum_{s}(\omega_{is} - p_i)\right) \left( \frac{1}{\sqrt{T}} \sum_{t} (\omega_{it} - p_i ) m_{it} \psi_{t,j_2} \right) \right)^2 .
\end{align*}
Then, by the independence of $\omega_{it}$ across $i$ and $t$ conditioning on $\{X,Y,M\}$, we have
\begin{align*}
\bbE\left[ \left.\norm{\Phi^\top (\Pi^{-1}(( \Omega -\Pi \boldsymbol{1}\boldsymbol{1}^\top )\circ M)) \Psi}_F^2 \right|  M, X, Z \right] = O_p(p_{\min}^{-2} N^2 J^2 d_1 d_2).
\end{align*}
Hence, we have 
$$
\norm{P_X ( \widehat{\Pi}^{-1} - \Pi^{-1} ) (( \Omega -\Pi \boldsymbol{1}\boldsymbol{1}^\top ) \circ M) P_Z }_F = O_p\left(p_{\min}^{-1} \frac{\sqrt{N}}{\sqrt{T}} J \sqrt{d_1 d_2}\right) . 
$$
In addition, note that
\begin{align*}
& \norm{\Phi^\top ( \widehat{\Pi}^{-1} - \Pi^{-1} )( (\Pi \boldsymbol{1}\boldsymbol{1}^\top ) \circ M)\Psi }_F^2 \\
&\ \ \lesssim \sum_{j_1,j_2}\left( \sum_{i,t} p_i^{-1} (\widehat{p_i} - p_i) m_{it} \phi_{i,j_1} \psi_{t,j_2} \right)^2  + \sum_{j_1,j_2}\left( \sum_{i,t} p_i^{-1} \widehat{p}_i^{-1} (\widehat{p_i} - p_i)^2 m_{it} \phi_{i,j_1} \psi_{t,j_2} \right)^2.
\end{align*}
The first term is bounded like
\begin{align*}
&\bbE\left[ \left. \left( \sum_{i,t} p_i^{-1} (\widehat{p_i} - p_i) m_{it} \phi_{i,j_1} \psi_{t,j_2} \right)^2  \right|  M, X, Z \right]  \\
&= \bbE\left[ \left. \left( \sum_{i,s} p_i^{-1} (\omega_{is} - p_i)  \phi_{i,j_1} \left( \frac{1}{T} \sum_{t} m_{it} \psi_{t,j_2} \right) \right)^2  \right|  M, X, Z \right]  \\
&= O_p(p_{\min}^{-1} N T).
\end{align*}
The second term is bounded like
\begin{align*}
&\sum_{j_1,j_2} \left( \sum_{i,t} p_i^{-1} \widehat{p}_i^{-1} (\widehat{p_i} - p_i)^2 m_{it} \phi_{i,j_1} \psi_{t,j_2} \right)^2 \\
& \ \ \leq \sum_{i} p_i^{-2} \left( \frac{1}{\sqrt{T}} \sum_{s} (\omega_{is} - p_i)  \right)^4
\sum_{i} \widehat{p}_i^{-2}  \left( \sum_{j_1}\phi_{i,j_1}^2 \right) \sum_{j_2} \left( \frac{1}{T}\sum_{t} m_{it} \psi_{t,j_2} \right)^2
= O_p(p_{\min}^{-2} N^2 J^2 d_1 d_2).
\end{align*}
Hence, we have 
$$
\norm{P_X( \widehat{\Pi}^{-1} - \Pi^{-1} )( (\Pi \boldsymbol{1}\boldsymbol{1}^\top ) \circ M) P_Z }_F = O_p\left(  p_{\min}^{-1} \frac{\sqrt{N}}{\sqrt{T}} J \sqrt{d_1 d_2} + \frac{ J\sqrt{d_1d_2}}{\sqrt{p_{\min}}} \right) . 
$$
Lastly, note that 
$$
\sum_{j_1,j_2} \bbE\left[ \left. \left( \sum_{i,t}(\widehat{p_i}^{-1} - p_i^{-1}) \omega_{it} \epsilon_{it} \phi_{i,j_1} \psi_{t,j_2} \right)^2  \right|  W, M, X, Z \right] 
= \sigma^2 \sum_{i,t}(\widehat{p_i}^{-1} - p_i^{-1})^2 \omega_{it} \phi_{i,j_1}^2 \psi_{t,j_2}^2 .
$$
In addition, the dominating term of the above part is
\begin{align*}
\sigma^2\sum_{j_1,j_2} \sum_{i,t}p_i^{-2}(\widehat{p_i} - p_i)^2 \omega_{it} \phi_{i,j_1}^2 \psi_{t,j_2}^2 
&= \sigma^2 \sum_{i} p_i^{-2} \left( \frac{1}{\sqrt{T}} \sum_s \omega_{is} - p_i    \right)^2  \sum_{j_1}\phi_{i,j_1}^2  \frac{1}{T}\sum_t \omega_{it} \sum_{j_2} \psi_{t,j_2}^2 \\
&= O_p \left( \sigma^2 N p_{\min}^{-1} J^2 d_1 d_2 \right).
\end{align*}
Hence, we have
$$
 \norm{ P_X ( \widehat{\Pi}^{-1} - \Pi^{-1} ) (\Omega \circ E) P_Z }_F = O_p \left( \frac{\sigma  J \sqrt{d_1 d_2} }{ \sqrt{p_{\min}T} }\right).
$$
In addition, from \eqref{eq:M1_hetro}, we have
\begin{align*}
Y - \widehat{M}_1 &= G_2 V_1^\top P_Z^\perp+ P_X^\perp W_1 V_2^\top + W_2 V_2^\top - P_X W_2 V_2^\top P_Z+ E\\
&\ \  -  P_X \Pi^{-1}(( \Omega -\Pi \boldsymbol{1}\boldsymbol{1}^\top )\circ M) P_Z 
- P_X ( \widehat{\Pi}^{-1} - \Pi^{-1} ) (\Omega \circ M) P_Z 
- P_X \Pi^{-1} (\Omega \circ E) P_Z\\
& \ \ - P_X ( \widehat{\Pi}^{-1} - \Pi^{-1} ) (\Omega \circ E) P_Z  - \textit{smoothing error\_1}.
\end{align*}
Note that, under our assumptions, $\norm{P_X \Pi^{-1}(( \Omega -\Pi \boldsymbol{1}\boldsymbol{1}^\top )\circ M) P_Z }_F$, $\norm{ P_X ( \widehat{\Pi}^{-1} - \Pi^{-1} ) (\Omega \circ M) P_Z }_F$, $\norm{ P_X \Pi^{-1} (\Omega \circ E) P_Z }_F$, $\norm{ P_X ( \widehat{\Pi}^{-1} - \Pi^{-1} ) (\Omega \circ E) P_Z }_F$, $\norm{\textit{smoothing error\_1}}_F$ are $o_p(p_{\max}^{1/2} \sqrt{N+T})$. Hence, we have
$$
\norm{\Omega \circ S'}  \leq \norm{\Omega \circ S'}_F \leq \norm{S'}_F = o_p\left(p_{\max}^{1/2} \sqrt{N+T}\right)
$$
where
\begin{align*}
S' &=  P_X \Pi^{-1}(( \Omega -\Pi \boldsymbol{1}\boldsymbol{1}^\top )\circ M) P_Z 
+ P_X ( \widehat{\Pi}^{-1} - \Pi^{-1} ) (\Omega \circ M) P_Z 
+ P_X \Pi^{-1} (\Omega \circ E) P_Z \\
& \ \ + P_X ( \widehat{\Pi}^{-1} - \Pi^{-1} ) (\Omega \circ E) P_Z  + \textit{smoothing error\_1}
\end{align*}
In addition, we have $\norm{\Omega \circ E} \lesssim \sigma p_{\max}^{1/2}  \sqrt{N+T}$ with high probability by Lemma \ref{lem:orthogonal_oper_norm} since $\norm{\omega_{it} \epsilon_{it}}_{\psi_2} \leq p_{\max}^{1/2} \sigma$. Hence, $2 C_\lambda \norm{\Omega \circ (E - S')} \leq \nu = C p_{\max}^{1/2}  \sqrt{N+T}$ for some large $C > 0$ with high probability, where $C_\lambda > 1$ is in Lemma \ref{lem:nuclear_missing}. Then, by setting $S = E - S'$ and $L = M_2 P_Z^\perp+ P_X^\perp M_3 + M_4 - P_X M_4 P_Z$, we can get from Lemma \ref{lem:nuclear_missing} that
\begin{align*}
&\norm{\widehat{M}_{rest} - (M_2 P_Z^\perp+ P_X^\perp M_3 + M_4 - P_X M_4 P_Z)}_F  \\
&= O_p\left( \sqrt{r^*}   \min \left\{ \frac{\nu}{p_{\min}} + \frac{ p_{\max}^{1/2} \sqrt{N + T}  M_{\max}}{p_{\min}} , \norm{M_2 P_Z^\perp+ P_X^\perp M_3 + M_4 - P_X M_4 P_Z}_F \right\} \right)    .
\end{align*}
Then, because
\begin{align*}
\norm{\widehat{M}_{1} + \widehat{M}_{rest} - M}_F &\leq
\norm{\widehat{M}_{rest} - (M_2 P_Z^\perp+ P_X^\perp M_3 + M_4 - P_X M_4 P_Z)}_F \\
&\ \ + \norm{P_X \Pi^{-1}(( \Omega -\Pi \boldsymbol{1}\boldsymbol{1}^\top )\circ M) P_Z 
- P_X ( \widehat{\Pi}^{-1} - \Pi^{-1} ) (\Omega \circ M) P_Z   }_F \\
&\ \ + \norm{ P_X \Pi^{-1} (\Omega \circ E) P_Z + P_X ( \widehat{\Pi}^{-1} - \Pi^{-1} ) (\Omega \circ E) P_Z }_F  \\
&\ \ + \norm{\textit{smoothing error\_1}}_F,
\end{align*}
and
\begin{align*}
&\norm{P_X \Pi^{-1}(( \Omega -\Pi \boldsymbol{1}\boldsymbol{1}^\top )\circ M) P_Z 
- P_X ( \widehat{\Pi}^{-1} - \Pi^{-1} ) (\Omega \circ M) P_Z   }_F 
= O_p\left( \frac{J\sqrt{d_1d_2}}{\sqrt{p_{\min}}}  +  p_{\min}^{-1} \frac{\sqrt{N}}{\sqrt{T}} J \sqrt{d_1 d_2} \right),\\
&\norm{  P_X \Pi^{-1} (\Omega \circ E) P_Z + P_X ( \widehat{\Pi}^{-1} - \Pi^{-1} ) (\Omega \circ E) P_Z  }_F 
= O_p\left( \frac{\sigma J\sqrt{d_1d_2}}{\sqrt{p_{\min}}}\right),\\
&\norm{\textit{smoothing error\_1}}_F = O_p\left( \frac{\sqrt{NT} K_1}{J^{\gamma_1^G/d_1}} + \frac{\sqrt{NT} K_1}{J^{\gamma_1^Q/d_2}} + \frac{\sqrt{NT} K_2}{J^{\gamma_2^G/d_1}} + \frac{\sqrt{NT} K_3}{J^{\gamma_2^Q/d_2}} \right),
\end{align*}
as noted in \eqref{eq:smoothing_1}, we have the desired result. $\square$

\subsubsection{Proof of Corollaries \ref{corr:convergence_rate_obs_iden} and \ref{corr:convergence_rate_mar_iden}}

In the proofs of Theorems \ref{thm:convergence_rate_obs} and \ref{thm:convergence_rate_mar}, if we apply the following restriction
$$
 P_XW_1=P_XW_2=0,
\qquad
P_ZV_1=P_ZV_2=0,  
$$
we can easily have the desired results. $\square$

\subsubsection{Proof of Theorem \ref{thm:entry_convergence_rate_mar}}

Note that
\begin{align*}
\widehat{M}_1 &= \frac{1}{\widehat{p}} P_X (\Omega \circ M) P_Z + \frac{1}{\widehat{p}} P_X (\Omega \circ E) P_Z\\
\nonumber& = P_X  M P_Z + \left( \frac{p}{\widehat{p}} - 1\right) P_X  M P_Z + \frac{1}{\widehat{p}} P_X (( \Omega - p \boldsymbol{1}\boldsymbol{1}^\top )\circ M) P_Z + \frac{1}{\widehat{p}} P_X (\Omega \circ E) P_Z \\
\nonumber& = G_1 Q_1^\top + \textit{smoothing error\_1}\\
\nonumber&\ \ + \left( \frac{p}{\widehat{p}} - 1\right) P_X  M P_Z  + \frac{1}{p} P_X (( \Omega - p \boldsymbol{1}\boldsymbol{1}^\top )\circ M) P_Z + \frac{1}{\widehat{p}} P_X (\Omega \circ E) P_Z  ,
\end{align*}
where $\textit{smoothing error\_1}$ is defined in \eqref{eq:part_1}. First, because
\begin{align*}
\max_{it} e_i^\top P_X  M P_Z e_t &= \max_{it}  \boldsymbol{\phi}(x_i)^\top (\Phi^\top \Phi)^{-1} \Phi M \Psi (\Psi^\top \Psi)^{-1} \boldsymbol{\psi}(z_t) \\
&\leq \max_{it}  \norm{\Phi (\Phi^\top \Phi)^{-1}}  \norm{\Psi (\Psi^\top \Psi)^{-1}} \norm{\boldsymbol{\phi}(x_i)} \norm{\boldsymbol{\psi}(z_t)} \\
& = O_p \left( J\sqrt{d_1 d_2 \mu_\phi \mu_\psi} \right),
\end{align*}
we have
$$
\max_{it} \left( \frac{p}{\widehat{p}} - 1\right) e_i^\top P_X  M P_Z e_t = O_p \left( \frac{ J\sqrt{d_1 d_2 \mu_\phi \mu_\psi}}{\sqrt{p}  \sqrt{NT}} \right) .
$$
In addition, we have
\begin{align*}
\max_{it} \frac{1}{p} e_i^\top P_X (( \Omega - p \boldsymbol{1}\boldsymbol{1}^\top )\circ M) P_Z e_t 
&= \max_{it} \frac{1}{p} \boldsymbol{\phi}(x_i)^\top (\Phi^\top \Phi)^{-1} \Phi (( \Omega - p \boldsymbol{1}\boldsymbol{1}^\top )\circ M)  \Psi (\Psi^\top \Psi)^{-1} \boldsymbol{\psi}(z_t)  \\
&= \max_{it}  \frac{1}{p} \sum_{js} (\omega_{js} - p) m_{js} (\boldsymbol{\phi}(x_i)^\top (\Phi^\top \Phi)^{-1} \Phi e_j)  (\boldsymbol{\psi}(z_t)^\top (\Psi^\top \Psi)^{-1} \Psi e_s) \\
& = O_p \left( \frac{ J M_{\max} \sqrt{d_1 d_2 \mu_\phi \mu_\psi}}{\sqrt{p}\sqrt{NT}}   \right) 
\end{align*}
and
\begin{align*}
\max_{it} \frac{1}{\widehat{p}} e_i^\top P_X ( \Omega \circ E) P_Z e_t 
&=  \max_{it} \frac{1}{\widehat{p}} \boldsymbol{\phi}(x_i)^\top (\Phi^\top \Phi)^{-1} \Phi ( \Omega \circ E)  \Psi (\Psi^\top \Psi)^{-1} \boldsymbol{\psi}(z_t)  \\
&= \max_{it} \frac{1}{\widehat{p}} \sum_{js} \omega_{js} \epsilon_{js} (\boldsymbol{\phi}(x_i)^\top (\Phi^\top \Phi)^{-1} \Phi e_j)  (\boldsymbol{\psi}(z_t)^\top (\Psi^\top \Psi)^{-1} \Psi e_s) \\
& = O_p \left( \frac{ \sigma J \sqrt{d_1 d_2 \mu_\phi \mu_\psi}}{\sqrt{p}\sqrt{NT}}   \right) .
\end{align*}
Next, in the following equation
\begin{align*}
Y - \widehat{M}_1 &= G_2 V_1^\top +  W_1 V_2^\top + W_2 V_2^\top + E\\
& \ \  - \left( \frac{p}{\widehat{p}} - 1\right) P_X  M P_Z- \frac{1}{p} P_X (( \Omega - p \boldsymbol{1}\boldsymbol{1}^\top )\circ M) P_Z - \frac{1}{\widehat{p}} P_X (\Omega \circ E) P_Z  - \textit{smoothing error\_1},
\end{align*}
we set 
\begin{align*}
&Y^o = Y - \widehat{M}_1, \qquad 
M^o = G_2 V_1^\top +  W_1 V_2^\top + W_2 V_2^\top, \qquad 
E^o = E, \\
& R^o =  - \left( \frac{p}{\widehat{p}} - 1\right) P_X  M P_Z- \frac{1}{p} P_X (( \Omega - p \boldsymbol{1}\boldsymbol{1}^\top )\circ M) P_Z - \frac{1}{\widehat{p}} P_X (\Omega \circ E) P_Z  - \textit{smoothing error\_1} .
\end{align*}
Then, we have by Lemma \ref{lem:ccfmy_mar}
\begin{align*}
&\norm{\widehat{M}_{rest} - (G_2 V_1^\top +  W_1 V_2^\top + W_2 V_2^\top)}_{\infty}  =  O_p \left( \frac{\kappa_{rest}^2 \mu_{rest} r_{rest}^{3/2} \max \{\sqrt{\log N}, \sqrt{\log T} \}}{p^{1/2} \min\{\sqrt{N}, \sqrt{T} \}} \right)   .
\end{align*}
Then, because
\begin{align*}
\norm{\widehat{M}_{1} + \widehat{M}_{rest} - M}_\infty &\leq
\norm{\widehat{M}_{rest} - (G_2 V_1^\top +  W_1 V_2^\top + W_2 V_2^\top)}_\infty \\
&\ \ + \norm{\left( \frac{p}{\widehat{p}} - 1\right) P_X  M P_Z  +  \frac{1}{p} P_X (( \Omega - p \boldsymbol{1}\boldsymbol{1}^\top )\circ M) P_Z + \frac{1}{\widehat{p}} P_X (\Omega \circ E) P_Z  }_\infty \\
&\ \ + \norm{\textit{smoothing error\_1}}_\infty,
\end{align*}
and
\begin{align*}
&\norm{ \left( \frac{p}{\widehat{p}} - 1\right) P_X  M P_Z  + \frac{1}{p} P_X (( \Omega - p \boldsymbol{1}\boldsymbol{1}^\top )\circ M) P_Z + \frac{1}{\widehat{p}} P_X (\Omega \circ E) P_Z  }_\infty \\
&= O_p\left( \frac{ (1 + \norm{M_{rest}}_{\infty}) J \sqrt{d_1 d_2 \mu_\phi \mu_\psi}}{\sqrt{p}\sqrt{NT}} \right),\\
&\norm{\textit{smoothing error\_1}}_\infty = O_p\left( \frac{K_1}{J^{\gamma_1^G/d_1}} + \frac{ K_1}{J^{\gamma_1^Q/d_2}} \right),
\end{align*}
we have the desired result. $\square$

\subsubsection{Proof of Theorem \ref{thm:entry_convergence_rate_mar_weak_signal}}

First, we want to show that, with the probability converging to one,
$$
\norm{\Omega \circ (Y - \widehat{M}_1)} \leq C_o p^{1/2} \sqrt{N + T}
$$
for some constant $C_o>0$ where 
\begin{align*}
Y - \widehat{M}_1 &= G_2 V_1^\top +  W_1 V_2^\top + W_2 V_2^\top + E\\
& \ \  - \left( \frac{p}{\widehat{p}} - 1\right) P_X  M P_Z- \frac{1}{p} P_X (( \Omega - p \boldsymbol{1}\boldsymbol{1}^\top )\circ M) P_Z - \frac{1}{\widehat{p}} P_X (\Omega \circ E) P_Z  - \textit{smoothing error\_1}.
\end{align*}
As we studied in the proof of Theorem \ref{thm:convergence_rate_mar}, we know
\begin{align*}
&\norm{\Omega \circ \left(\left( \frac{p}{\widehat{p}} + 1\right) P_X  M P_Z- \frac{1}{p} P_X (( \Omega - p \boldsymbol{1}\boldsymbol{1}^\top )\circ M) P_Z + \frac{1}{\widehat{p}} P_X (\Omega \circ E) P_Z  + \textit{smoothing error\_1} \right) }\\
&\ll p^{1/2} \sqrt{N + T}   ,
\end{align*}
with the probability converging to one. In addition, 
$$
\norm{\Omega \circ E} \lesssim p^{1/2} \sqrt{N + T}
$$
with the probability converging to one by Lemma \ref{lem:orthogonal_oper_norm} since $\norm{\omega_{it} \epsilon_{it}}_{\psi_2} \leq p^{1/2} \sigma$. Moreover, we have
\begin{align*}
\norm{\Omega \circ M_{rest}} \leq \norm{M_{rest}}_F \leq \Upsilon_{rest,\min} \kappa_{rest} r_{rest}^{1/2} \lesssim p^{1/2} \sqrt{N + T}
\end{align*}
with the probability converging to one under our assumption. Hence, we have
$$
\norm{\Omega \circ (Y - \widehat{M}_1)} \leq C_o p^{1/2} \sqrt{N + T}
$$
for some constant $C_o>0$ with the probability converging to one. Hence, with the probability converging to one,
$$
\norm{\Omega \circ (Y - \widehat{M}_1)}  \leq \frac{\nu}{2}
$$
where $\nu = C p^{1/2} \sqrt{N+T}$ for some large constant $C>0$.

By Theorem 1 and Lemma 5 of \cite{mazumder2010spectral}, we know that $\widehat{M}_{rest}$ is the limit of the soft impute algorithm (Algorithm 1 of \cite{mazumder2010spectral}) with $Y - \widehat{M}_1$ and it is the solution to the following fixed point equation:
$$
\widehat{M}_{rest} = \calS_{\nu/2} \left( \Omega \circ (Y - \widehat{M}_1) + \Omega^{\perp} \circ \widehat{M}_{rest} \right),
$$
where $\Omega^{\perp} =  \textbf{1}_N \textbf{1}_T^\top - \Omega$ and 
$$
\calS_{\nu/2} (W) = U D_{\nu/2} V^\top \quad \text{  with  }  D_{\nu/2} =\diag [(d_1 - \nu/2)_{+} , \dots , (d_r - \nu/2)_{+} ] ,
$$
$UDV^\top$ is the SVD of $W$, $D = \diag [d_1 , \dots , d_r ]$, and $t_+ = \max (t,0)$. Therefore, when $\norm{\Omega \circ (Y - \widehat{M}_1)}  \leq \frac{\nu}{2}$, we have $\widehat{M}_{rest}  = \textbf{0}_{N \times T}$. Hence, with the probability converging to one, we have
$$
\norm{\widehat{M}_{rest} - M_{rest}}_{\infty} = \norm{M_{rest}}_{\infty}.
$$
Then, because
\begin{align*}
\norm{\widehat{M}_{1} + \widehat{M}_{rest} - M}_\infty &\leq
\norm{\widehat{M}_{rest} - M_{rest}}_\infty \\
&\ \ + \norm{\left( \frac{p}{\widehat{p}} - 1\right) P_X  M P_Z  +  \frac{1}{p} P_X (( \Omega - p \boldsymbol{1}\boldsymbol{1}^\top )\circ M) P_Z + \frac{1}{\widehat{p}} P_X (\Omega \circ E) P_Z  }_\infty \\
&\ \ + \norm{\textit{smoothing error\_1}}_\infty,
\end{align*}
and
\begin{align*}
&\norm{ \left( \frac{p}{\widehat{p}} - 1\right) P_X  M P_Z  + \frac{1}{p} P_X (( \Omega - p \boldsymbol{1}\boldsymbol{1}^\top )\circ M) P_Z + \frac{1}{\widehat{p}} P_X (\Omega \circ E) P_Z  }_\infty \\
&= O_p\left( \frac{ (1 + \norm{M_{rest}}_{\infty}) J \sqrt{d_1 d_2 \mu_\phi \mu_\psi}}{\sqrt{p}\sqrt{NT}} \right),\\
&\norm{\textit{smoothing error\_1}}_\infty = O_p\left( \frac{K_1}{J^{\gamma_1^G/d_1}} + \frac{ K_1}{J^{\gamma_1^Q/d_2}} \right),
\end{align*}
we have the desired result. $\square$

\subsubsection{Proof of Theorem \ref{thm:entry_convergence_rate_observed}}

First,
we note that
\begin{align*}
\widehat{M}_1 = P_X G_1 Q_1^\top P_Z  + P_X E P_Z = G_1 Q_1^\top  + P_X E P_Z + \textit{smoothing error\_1},
\end{align*}
where
\begin{align*}
\textit{smoothing error\_1} &= -P_X^\perp R_{G_1} Q_1^\top - G_1 R_{Q_1}^\top P_Z^\perp  + P_X^\perp R_{G_1}R_{Q_1}^\top P_Z^\perp  .
\end{align*}
Then, by Assumption \ref{asp:sieve}, we have 
\begin{gather*}
\norm{\textit{smoothing error\_1}}_\infty = O_p\left( \frac{ K_1}{J^{\gamma_1^G/d_1}} + \frac{K_1}{J^{\gamma_1^Q/d_2}} \right) .    
\end{gather*}
Next, for $\widehat{M}_2$, note that
\begin{align*}
P_X Y P_Z^\perp = P_X G_1 Q_1^\top P_Z^\perp+ P_X G_2 V_1^\top  + P_X E P_Z^\perp  =  M_2 + P_X E P_Z^\perp+ \textit{smoothing error\_2},
\end{align*}
where
$$
\textit{smoothing error\_2} = P_X G_1 R_{Q_1}^\top P_Z^\perp+ -P_X^\perp R_{G_2} V_1^\top  .
$$
By Assumption \ref{asp:sieve}, we have
$$
\norm{\textit{smoothing error\_2}} = O\left( \frac{\sqrt{NT} K_1}{J^{\gamma_1^Q/d_2}}  + \frac{\sqrt{NT} K_2}{J^{\gamma_2^G/d_1}} \right) = o\left(\sqrt{T} \right).
$$
Then, because $\norm{P_X E P_Z^\perp } \lesssim \sigma \sqrt{T}$ with probability converging to one by Lemma \ref{lem:error_oper_norm}, we have 
$$
2  \norm{P_X E P_Z^\perp  + \textit{smoothing error\_2}} \leq \nu_2 = C_2 \sqrt{T}
$$
for some large $C_2>0$ with the probability converging to one.
Denote the SVD of $P_X Y P_Z^\perp$ by $\tilde{U} \tilde{D} \tilde{V}^\top$ and that of $M_2 $ by $U D V^\top$. Then, we have
$$
\widehat{M}_2 = \tilde{U} \left( \tilde{D} - \frac{\nu_2}{2} \cdot I_{N \wedge T} \right)_{+} \tilde{V}^\top .
$$
Denote the $r$-th singular value of a matrix $A$ by $\Upsilon_{r}(A)$. By Weyl's inequality, for $1 \leq r \leq r_2$, with probability converging to one,
$$
\Upsilon_{r}(\tilde{D}) \geq \Upsilon_{r}(D) - \norm{P_X E P_Z^\perp  + \textit{smoothing error\_2}} \geq \frac{1}{2} \Upsilon_{r}(D) > \frac{\nu_2}{2}
$$
because $\Upsilon_{r}(D) \gg \sqrt{T}$. In addition, we have with probability converging to one, for $r > r_2$,
$$
\Upsilon_{r}(\tilde{D}) \leq \Upsilon_{r}(D) + \norm{P_X E P_Z^\perp  + \textit{smoothing error\_2}} \leq \frac{v_2}{2},
$$
since $\Upsilon_{r}(D) = 0$. Hence, with probability converging to one,
$$
\widehat{M}_2 = \tilde{U}_{r_2} \left( \tilde{D}_{r_2} - \frac{\nu_2}{2} \cdot I_{r_2} \right) \tilde{V}_{r_2}^\top,
$$
where $\tilde{D}_{r_2}$ is a diagonal matrix of the top-$r_2$ singular values, $\tilde{U}_{r_2}$, $\tilde{V}_{r_2}$ are the corresponding left and right singular values, respectively.
Then, with probability converging to one, we have
\begin{align*}
\norm{\widehat{M}_2 - M_2}_{\infty} 
&\lesssim \norm{\tilde{U}_{r_2} - U_{r_2} }_{2,\infty} \norm{D_{r_2}} \norm{ V_{r_2} }_{2,\infty} 
+ \norm{U_{r_2} }_{2,\infty} \norm{\tilde{D}_{r_2} - \frac{\nu_2}{2} \cdot I_{N} - D_{r_2}} \norm{ V_{r_2} }_{2,\infty} \\
& \ \ +  \norm{U_{r_2} }_{2,\infty} \norm{D_{r_2}} \norm{\tilde{V}_{r_2} - V_{r_2} }_{2,\infty} .
\end{align*}
Note that $ \norm{ P_X E P_Z^\perp}_1 \leq \norm{ P_X E}_1 + \norm{ P_X E P_Z}_1 $ and the dominating term is $\norm{ P_X E}_1$ where
\begin{align*}
\norm{ P_X E}_1 
&= \max_t \sum_{i} \abs{ \sum_{j} \boldsymbol{\phi}(x_i)^\top \left( \Phi^\top  \Phi \right)^{-1} \boldsymbol{\phi}(x_j) \epsilon_{jt}  } \\
&\lesssim N \max_{it} \abs{ \sum_{j} \boldsymbol{\phi}(x_i)^\top \left( \Phi^\top  \Phi \right)^{-1} \boldsymbol{\phi}(x_j) \epsilon_{jt}  } \\
& = O_p\left( \sigma \sqrt{J d_1 \mu_{\phi}}  \sqrt{N} \sqrt{\log (N + T)}    \right).
\end{align*}
Similarly, we have
$$
\norm{(P_X E)^\top}_{1} =  \max_i \sum_{t} \abs{ \sum_{j} \boldsymbol{\phi}(x_i)^\top \left( \Phi^\top  \Phi \right)^{-1} \boldsymbol{\phi}(x_j) \epsilon_{jt}  } 
= O_p\left( \sigma \sqrt{J d_1 \mu_{\phi}}  \frac{T}{\sqrt{N}} \sqrt{\log (N + T)}    \right) .
$$
Hence, we have
\begin{align*}
\max \{\frac{\sqrt{T}}{\sqrt{N}}\norm{ P_X E P_Z^\perp}_1, \frac{\sqrt{N}}{\sqrt{T}}\norm{ (P_X E P_Z^\perp)^\top }_1  \} = O_p\left(  \sqrt{J d_1 \mu_{\phi}}  \sqrt{ T \log (N + T)}     \right).
\end{align*}
In addition, under our assumption, we also have
$$
\max \{\frac{\sqrt{T}}{\sqrt{N}}\norm{ (\textit{smoothing error\_2}) }_1, \frac{\sqrt{N}}{\sqrt{T}}\norm{ (\textit{smoothing error\_2})^\top }_1  \} 
= O_p\left(  \sqrt{J d_1 \mu_{\phi}}  \sqrt{ T \log (N + T)}     \right).
$$
Then, by Theorem 5 of \cite{fan2018l}, we can have
\begin{align*}
&\norm{\tilde{U}_{r_2} - U_{r_2} }_{2,\infty} = O_p \left( \frac{r_2^4 \mu_2^2  \sqrt{J d_1 \mu_{\phi}}  \sqrt{ T \log (N + T)}}{\sqrt{N} \Upsilon_{r_2}(M_2)}   \right) ,\\
&\norm{\tilde{V}_{r_2} - V_{r_2} }_{2,\infty} = O_p \left( \frac{r_2^4 \mu_2^2  \sqrt{J d_1 \mu_{\phi}}  \sqrt{ \log (N + T)}}{ \Upsilon_{r_2}(M_2)}   \right) .
\end{align*}
In addition, we have
\begin{align*}
\norm{\tilde{D}_{r_2} - \frac{\nu_2}{2} \cdot I_{N} - D_{r_2}} 
&\leq \norm{\tilde{D}_{r_2}  - D_{r_2}} + \frac{\nu_2}{2} \\
&\leq \norm{P_X E P_Z^\perp + \textit{smoothing error\_2}} + \frac{\nu_2}{2}  \\
&= O_p\left( \sqrt{T} \right).
\end{align*}
Hence, we have
\begin{align*}
\norm{\widehat{M}_2 - M_2}_{\infty} 
&= O_p \left(   \frac{r_2^{9/2} \mu_2^{5/2} \kappa_2 \sqrt{J d_1 \mu_{\phi}}  \sqrt{  \log (N + T)}}{\sqrt{N}}  +  \frac{r_2 \mu_2}{\sqrt{N}}  +  \frac{r_2^{9/2} \mu_2^{5/2}  \sqrt{J d_1 \mu_{\phi}}  \sqrt{ \log (N + T)}}{ \sqrt{N} } \right) \\
&= O_p \left(   \frac{r_2^{9/2} \mu_2^{5/2} \kappa_2 \sqrt{J d_1 \mu_{\phi}}  \sqrt{  \log (N + T)}}{\sqrt{N}} \right) .
\end{align*}
Similarly, we can derive
\begin{align*}
\norm{\widehat{M}_3 - M_3}_{\infty} 
= O_p \left(   \frac{r_3^{9/2} \mu_3^{5/2} \kappa_3 \sqrt{J d_2 \mu_{\psi}}  \sqrt{  \log (N + T)}}{\sqrt{T}} \right) .
\end{align*}
Lastly, for $\widehat{M}_4$, note that
\begin{align*}
P_X^\perp Y P_Z^\perp&=   W_2 V_2^\top  + P_X^\perp E P_Z^\perp+ \textit{smoothing error\_4},
\end{align*}
where
\begin{align*}
 \textit{smoothing error\_4} = P_X^\perp R_{G_1} R_{Q_1}^\top P_Z^\perp+ P_X^\perp R_{G_2} V_1^\top   +  W_1 R_{Q_2}^\top P_Z^\perp.   
\end{align*}
we set 
\begin{align*}
&Y^o = P_X^\perp Y P_Z^\perp, \qquad 
M^o = W_2 V_2^\top, \qquad 
E^o = E, \\
& R^o = - P_X E - E P_Z + P_X E P_Z + \textit{smoothing error\_4} .
\end{align*}
Note that 
\begin{align*}
\norm{P_X E}_{\infty} = \max_{it} \abs{\sum_j \boldsymbol{\phi}(x_i)^\top \left( \Phi^\top \Phi  \right)^{-1}   \boldsymbol{\phi}(x_j)\epsilon_{jt}} \lesssim
\frac{\sigma \sqrt{J d_1 \mu_{\phi} } \max \{\sqrt{\log N} , \sqrt{\log T} \}}{\sqrt{N}},\\
\norm{E P_Z}_{\infty} = \max_{it} \abs{\sum_s \boldsymbol{\psi}(z_t)^\top \left( \Phi^\top \Phi  \right)^{-1}  \boldsymbol{\psi}(z_s)\epsilon_{is}} \lesssim
\frac{\sigma \sqrt{J d_2 \mu_{\psi} } \max \{\sqrt{\log N} , \sqrt{\log T} \}}{\sqrt{T}},
\end{align*}
with the probability converging to one. In addition, by Assumption \ref{asp:sieve}, we have
\begin{align*}
&\norm{\textit{smoothing error\_4}}_{\infty} = O_p\left(  \frac{ K_1}{J^{(\gamma_1^G/d_1 + \gamma_1^Q/d_2)}} + \frac{K_2}{J^{\gamma_2^G/d_1}} + \frac{ K_3}{J^{\gamma_2^Q/d_2}} \right) = o_p\left( \frac{1}{\min \{\sqrt{N} ,\sqrt{T} \}} \right),\\
&\norm{\textit{smoothing error\_4}} = O_p\left(  \frac{\sqrt{NT} K_1}{J^{(\gamma_1^G/d_1 + \gamma_1^Q/d_2)}} + \frac{\sqrt{NT} K_2}{J^{\gamma_2^G/d_1}} + \frac{\sqrt{NT} K_3}{J^{\gamma_2^Q/d_2}} \right) = o_p\left(\sqrt{N} + \sqrt{T} \right)   .
\end{align*}
Moreover, $\norm{P_X^\perp E P_Z^\perp} \lesssim \sigma (\sqrt{N} + \sqrt{T})$ with high probability by Lemma \ref{lem:error_oper_norm}, we have 
$$
2\norm{P_X^\perp E P_Z^\perp+ \textit{smoothing error\_4}} \leq \frac{7}{8}\nu_4
$$
for some large $C_4>0$ with high probability. Then, we have by Lemma \ref{lem:ccfmy_observed} that
$$
\norm{\widehat{M}_4 - M_4 }_\infty =  O_p \left( \frac{J^{1/2} \max\{ \sqrt{d_1 \mu_{\phi}},\sqrt{d_2 \mu_{\psi}} \}\kappa_4^2 \mu_4 r_4 \theta \max \{\sqrt{\log N}, \sqrt{\log T} \}}{ \min\{\sqrt{N}, \sqrt{T} \}} \right).
$$
Note that
\begin{align*}
\widehat{M} - M &= \widehat{M}_1 + \widehat{M}_2 + \widehat{M}_3 + \widehat{M}_4 - \left( M_1 + M_2 + M_3 + M_4 \right) \\
&= \widehat{M}_2 - M_2  
+ \widehat{M}_3 -  M_3   + \widehat{M}_4 - M_4 + P_X E P_Z + \textit{smoothing error\_1}.
\end{align*}
To sum up, we have that
\begin{align*}
\norm{\widehat{M} - M}_{\infty}  
&= O_p\left( J\sqrt{d_1d_2} +\frac{r_2^{9/2} \mu_2^{5/2} \kappa_2 \sqrt{J d_1 \mu_{\phi}}  \sqrt{  \log (N + T)}}{\sqrt{N}} +  \frac{r_3^{9/2} \mu_3^{5/2} \kappa_3 \sqrt{J d_2 \mu_{\psi}}  \sqrt{  \log (N + T)}}{\sqrt{T}} \right.  \\
& \qquad  \quad  
+ \frac{J^{1/2} \max\{ \sqrt{d_1 \mu_{\phi}},\sqrt{d_2 \mu_{\psi}} \}\kappa_4^2 \mu_4 r_4 \theta \max \{\sqrt{\log N}, \sqrt{\log T} \}}{ \min\{\sqrt{N}, \sqrt{T} \}}  \\
& \left. \qquad  \quad + \left[ \frac{K_1}{J^{\gamma_1^G/d_1}} + \frac{K_1}{J^{\gamma_1^Q/d_2}} \right]  \right) . \ \ \square
\end{align*}

\subsubsection{Proof of Theorem \ref{thm:entry_convergence_rate_observed_weak_signal}}

Denote the SVD of $P_X Y P_Z^\perp$ by $\tilde{U} \tilde{D} \tilde{V}^\top$. Then, we have
$$
\widehat{M}_2 = \tilde{U} \left( \tilde{D} - \frac{\nu_2}{2} \cdot I_{N \wedge T} \right)_{+} \tilde{V}^\top .
$$
Note that
\begin{align*}
P_X Y P_Z^\perp = M_2 + P_X E P_Z^\perp+ \textit{smoothing error\_2},
\end{align*}
In the proof of Theorem \ref{thm:entry_convergence_rate_observed}, we show that 
$$
2  \norm{P_X E P_Z^\perp  + \textit{smoothing error\_2}} \lesssim \sqrt{T}
$$
with the probability converging to one. In addition, under the assumption that $\kappa_2 \Upsilon_{2,\min}  \lesssim \sqrt{T}$, we have
$$
2\norm{M_2} = 2 \kappa_2 \Upsilon_{2,\min}  \lesssim \sqrt{T}.
$$
Hence, with $v_2 = C_2 \sqrt{T}$ for some large $C_2 >0$, we have
$$
\norm{P_X Y P_Z^\perp} \leq v_2/2 ,
$$
and $\widehat{M}_2 = \textbf{0}_{N \times T}$ with the probability converging to one. Similarly, if $\kappa_3 \Upsilon_{3,\min}  \lesssim \sqrt{N}$, we have $\widehat{M}_3 = \textbf{0}_{N \times T}$, and if $\kappa_4 \Upsilon_{4,\min}  \lesssim \sqrt{N + T}$, we have $\widehat{M}_4 = \textbf{0}_{N \times T}$ with the probability converging to one. Therefore, if $\kappa_2 \Upsilon_{2,\min}  \lesssim \sqrt{T}$,
$$
\norm{\widehat{M}_2 - M_2}_\infty =  \norm{M_2}_\infty,
$$
if $\kappa_3 \Upsilon_{3,\min}  \lesssim \sqrt{N}$,
$$
\norm{\widehat{M}_3 - M_3}_\infty =  \norm{M_3}_\infty,
$$
if $\kappa_4 \Upsilon_{4,\min}  \lesssim \sqrt{N + T}$,
$$
\norm{\widehat{M}_4 - M_4}_\infty =  \norm{M_4}_\infty,
$$
with the probability converging to one. It completes the proof. $\square$

\subsubsection{Proof of Corollaries \ref{corr:convergence_rate_obs_sigma} and \ref{corr:convergence_rate_mar_sigma}}

In the proofs of Theorems \ref{thm:convergence_rate_obs} and \ref{thm:convergence_rate_mar}, if we change the forms of regularization levels, $\nu_2$, $\nu_3$, $\nu_4$, and $\nu$,  we have the desired results.

\noindent (i) Proof of Corollary \ref{corr:convergence_rate_obs_sigma}: For $\iota \in \{ 1, 2\}$, define the sieve approximation error matrices as $R_{G_\iota} = G_\iota(X) - \Phi(X) B_\iota$ and $R_{Q_\iota} = Q_\iota(Z) - \Psi(Z) A_\iota$ where $A_\iota$ and $B_\iota$ are the sieve coefficient matrices consisting of $a_{\iota,k}$ and $b_{\iota,k}$ in Assumption \ref{asp:sieve}, respectively. First,
we note that
\begin{align}\label{eq:part_1}
\widehat{M}_1 &= P_X G_1 Q_1^\top P_Z + P_X G_2 V_1^\top P_Z + P_X W_1 Q_2^\top P_Z + P_X W_2 V_2^\top P_Z + P_X E P_Z \\
\nonumber&= G_1 Q_1^\top + G_2 V_1^\top P_Z + P_X W_1 Q_2^\top +  P_X W_2 V_2^\top P_Z  + P_X E P_Z + \textit{smoothing error\_1},
\end{align}
where
\begin{align*}
\textit{smoothing error\_1} &= -P_X^\perp R_{G_1} Q_1^\top - G_1 R_{Q_1}^\top P_Z^\perp  + P_X^\perp R_{G_1}R_{Q_1}^\top P_Z^\perp  -P_X^\perp R_{G_2} V_1^\top P_Z - P_X W_1 R_{Q_2}^\top P_Z^\perp .
\end{align*}
Then, because, by Assumption \ref{asp:sieve}, we have $\norm{R_{G_1}}_F = O_p \left( \frac{\sqrt{N K_1}}{J^{\gamma_1^G/d_1}}  \right)$, $\norm{R_{G_2}}_F = O_p \left( \frac{\sqrt{N K_2}}{J^{\gamma_1^G/d_1}}  \right)$, $\norm{R_{Q_1}}_F = O_p \left( \frac{\sqrt{T K_1}}{J^{\gamma_1^Q/d_2}}  \right)$, and $\norm{R_{Q_2}}_F = O_p \left( \frac{\sqrt{T K_3}}{J^{\gamma_2^Q/d_2}}  \right)$, we know
\begin{gather}\label{eq:smoothing_1}
\norm{\textit{smoothing error\_1}}_F = O_p\left( \frac{\sqrt{NT} K_1}{J^{\gamma_1^G/d_1}} + \frac{\sqrt{NT} K_1}{J^{\gamma_1^Q/d_2}} + \frac{\sqrt{NT} K_2}{J^{\gamma_2^G/d_1}} + \frac{\sqrt{NT} K_3}{J^{\gamma_2^Q/d_2}} \right) .    
\end{gather}
Next, for $\widehat{M}_2$, note that
\begin{align*}
P_X Y P_Z^\perp&= P_X G_1 Q_1^\top P_Z^\perp+ P_X G_2 V_1^\top P_Z^\perp+ P_X W_1 Q_2^\top P_Z^\perp\\
&\ \ + P_X W_2 V_2^\top P_Z^\perp+ P_X E P_Z^\perp\\
& =  G_2 V_1^\top  P_Z^\perp+  P_X W_2 V_2^\top P_Z^\perp + P_X E P_Z^\perp+ \textit{smoothing error\_2},
\end{align*}
where
$$
\textit{smoothing error\_2} = P_X G_1 R_{Q_1}^\top P_Z^\perp+ -P_X^\perp R_{G_2} V_1^\top P_Z^\perp+ P_X W_1 R_{Q_2}^\top P_Z^\perp.
$$
By Assumption \ref{asp:sieve}, we have
$$
\norm{\textit{smoothing error\_2}} = O \left( \frac{\sqrt{NT} K_1}{J^{\gamma_1^Q/d_2}}  + \frac{\sqrt{NT} K_2}{J^{\gamma_2^G/d_1}} + \frac{\sqrt{NT} K_3}{J^{\gamma_2^Q/d_2}} \right) .
$$
Then, because $\norm{P_X E P_Z^\perp } \lesssim \sigma \sqrt{T}$ with high probability by Lemma \ref{lem:error_oper_norm}, we have 
$$
2 C_\lambda \norm{P_X E P_Z^\perp  + \textit{smoothing error\_2}} \leq \nu_2
$$
with high probability, where $C_\lambda > 1$ is in Lemma \ref{lem:nuclear}. Hence, by setting $S = P_X E P_Z^\perp  + \textit{smoothing error\_2}$ and $L = M_2  P_Z^\perp+  P_X M_4 P_Z^\perp$, we can get by Lemma \ref{lem:nuclear} that
\begin{align}\label{eq:part_2}
\norm{\widehat{M}_2 - M_2  P_Z^\perp-  P_X M_4 P_Z^\perp}_F  = O_p\left( \sqrt{r_2 + r_4} \min\{ \nu_2, \norm{M_2  P_Z^\perp+  P_X M_4 P_Z^\perp}_F \} \right) .
\end{align}
For $\widehat{M}_3$, note that
\begin{align*}
P_X^\perp Y P_Z &=  P_X^\perp W_1 Q_2^\top  +  P_X^\perp W_2 V_2^\top P_Z  + P_X^\perp E P_Z + \textit{smoothing error\_3},
\end{align*}
where
$$
\textit{smoothing error\_3} = P_X^\perp R_{G_1} Q_1^\top P_Z + P_X^\perp R_{G_2} V_1^\top P_Z + P_X^\perp W_1 R_{Q_2}^\top P_Z^\perp .
$$
By Assumption \ref{asp:sieve}, we have
$$
\norm{\textit{smoothing error\_3}} = O\left(  \frac{\sqrt{NT} K_1}{J^{\gamma_1^Q/d_2}} + \frac{\sqrt{NT} K_2}{J^{\gamma_2^G/d_1}} + \frac{\sqrt{NT} K_3}{J^{\gamma_2^Q/d_2}} \right).
$$
Then, because $\norm{P_X^\perp E P_Z} \lesssim \sigma \sqrt{N}$ with high probability by Lemma \ref{lem:error_oper_norm}, we have 
$$
2 C_\lambda \norm{P_X^\perp E P_Z + \textit{smoothing error\_3}} \leq \nu_3
$$
for some large $C_3>0$ with high probability. Hence, by setting $S = P_X^\perp E P_Z + \textit{smoothing error\_3}$ and $L = P_X^\perp M_3  +  P_X^\perp M_4 P_Z $, we can derive by Lemma \ref{lem:nuclear} that
\begin{align}\label{eq:part_3}
\norm{\widehat{M}_3 - P_X^\perp M_3  -  P_X^\perp M_4 P_Z}_F  = O_p\left( \sqrt{r_3 + r_4} \min\{ \nu_3 , \norm{P_X^\perp M_3  +  P_X^\perp M_4 P_Z}_F \} \right) .
\end{align}
Lastly, for $\widehat{M}_4$, note that
\begin{align*}
P_X^\perp Y P_Z^\perp&=  P_X^\perp W_2 V_2^\top P_Z^\perp + P_X^\perp E P_Z^\perp+ \textit{smoothing error\_4},
\end{align*}
where
\begin{align*}
 \textit{smoothing error\_4} = P_X^\perp R_{G_1} R_{Q_1}^\top P_Z^\perp+ P_X^\perp R_{G_2} V_1^\top P_Z^\perp  + P_X^\perp W_1 R_{Q_2}^\top P_Z^\perp.   
\end{align*}
Then, by Assumption \ref{asp:sieve}, we have
$$
\norm{\textit{smoothing error\_4}} = O \left(  \frac{\sqrt{NT} K_1}{J^{(\gamma_1^G/d_1 + \gamma_1^Q/d_2)}} + \frac{\sqrt{NT} K_2}{J^{\gamma_2^G/d_1}} + \frac{\sqrt{NT} K_3}{J^{\gamma_2^Q/d_2}} \right) .
$$
Then, because $\norm{P_X^\perp E P_Z^\perp} \lesssim \sigma (\sqrt{N} + \sqrt{T})$ with high probability by Lemma \ref{lem:error_oper_norm}, we have $2 C_\lambda \norm{P_X^\perp E P_Z^\perp+ \textit{smoothing error\_4}} \leq \nu_4 $ with high probability. Hence, by setting $S = P_X^\perp E P_Z^\perp+ \textit{smoothing error\_4}$ and $L = P_X^\perp M_4 P_Z^\perp$, we can derive by Lemma \ref{lem:nuclear} that
\begin{align}\label{eq:part_4}
\norm{\widehat{M}_4 - P_X^\perp M_4 P_Z^\perp}_F = O_p\left( \sqrt{r_4} \min\{ \nu_4, \norm{ P_X^\perp M_4 P_Z^\perp}_F \} \right) .
\end{align}
From the relation \eqref{eq:part_1}, we have
\begin{align*}
\widehat{M} - M &= \widehat{M}_1 + \widehat{M}_2 + \widehat{M}_3 + \widehat{M}_4 - \left( M_1 + M_2 + M_3 + M_4 \right) \\
&= \widehat{M}_2 - M_2  P_Z^\perp-  P_X M_4 P_Z^\perp
+ \widehat{M}_3 - P_X^\perp M_3  -  P_X^\perp M_4 P_Z \\
&\ \ + \widehat{M}_4 - P_X^\perp M_4 P_Z^\perp
+ P_X E P_Z + \textit{smoothing error\_1}.
\end{align*}
Hence, we have from the bounds \eqref{eq:smoothing_1}, \eqref{eq:part_2}, \eqref{eq:part_3}, and \eqref{eq:part_4} with Lemma \ref{lem:error_oper_norm} that
\begin{align*}
\norm{\widehat{M} - M}_F  
&= O_p\left( \sigma J\sqrt{d_1d_2} + \sqrt{r_2 + r_4}   \min \left\{ \nu_2 , \norm{M_2 P_Z^\perp +  P_X M_4 P_Z^\perp}_F \right\} \right. \\
& \qquad  \quad   + \sqrt{r_3 + r_4}   \min \left\{ \nu_3 ,   \norm{ P_X^\perp M_3  +  P_X^\perp M_4 P_Z }_F \right\}  
+ \sqrt{r_4}   \min \left\{ \nu_4 , \norm{ P_X^\perp M_4 P_Z^\perp }_F  \right\} \\
& \left. \qquad  \quad + \sqrt{NT} \left[ \frac{K_1}{J^{\gamma_1^G/d_1}} + \frac{K_1}{J^{\gamma_1^Q/d_2}} + \frac{K_2}{J^{\gamma_2^G/d_1}} + \frac{K_3}{J^{\gamma_2^Q/d_2}} \right]  \right) .
\end{align*}
Here, we derive
$$
\norm{P_X E P_Z}_F \leq \sqrt{rank(P_X E P_Z)} \norm{P_X E P_Z} \lesssim \sqrt{ J \min\{d_1,d_2 \} } \sigma \sqrt{J\max\{d_1,d_2 \}}
$$
by Lemma \ref{lem:error_oper_norm}. $\square$
\bigskip

\noindent (ii) Proof of Corollary \ref{corr:convergence_rate_mar_sigma}: First, note that
\begin{align}\label{eq:M_1_est}
\widehat{M}_1 &= \frac{1}{\widehat{p}} P_X (\Omega \circ M) P_Z + \frac{1}{\widehat{p}} P_X (\Omega \circ E) P_Z\\
\nonumber& = P_X  M P_Z + \left( \frac{p}{\widehat{p}} - 1\right) P_X  M P_Z + \frac{1}{\widehat{p}} P_X (( \Omega - p \boldsymbol{1}\boldsymbol{1}^\top )\circ M) P_Z + \frac{1}{\widehat{p}} P_X (\Omega \circ E) P_Z \\
\nonumber& = G_1 Q_1^\top + G_2 V_1^\top P_Z + P_X W_1 Q_2^\top +  P_X W_2 V_2^\top P_Z + \textit{smoothing error\_1}\\
\nonumber&\ \ + \left( \frac{p}{\widehat{p}} - 1\right) P_X  M P_Z  + \frac{1}{p} P_X (( \Omega - p \boldsymbol{1}\boldsymbol{1}^\top )\circ M) P_Z + \frac{1}{\widehat{p}} P_X (\Omega \circ E) P_Z  ,
\end{align}
where $\textit{smoothing error\_1}$ is defined in \eqref{eq:part_1}. Because
$$
1 - \frac{p}{\widehat{p}}  = \widehat{p}^{-1} \left( \frac{1}{NT} \sum_{it} \omega_{it} - p \right) = O_p \left( p^{-1} \frac{p^{1/2}}{\sqrt{NT}} \right) = O_p \left( \frac{1}{\sqrt{pNT}} \right)
$$
and $\norm{M}_F = O_p (\sqrt{NT} )$, we have $\norm{\left( \frac{p}{\widehat{p}} - 1\right) P_X  M P_Z }_F =  O_p \left( \frac{1}{\sqrt{p}} \right)$.
Note that
$$
\norm{ \frac{1}{p} P_X (( \Omega - p \boldsymbol{1}\boldsymbol{1}^\top )\circ M) P_Z }_F \leq  \frac{1}{p} \norm{\Phi (\Phi^\top \Phi)^{-1}} \norm{\Phi^\top (( \Omega - p \boldsymbol{1}\boldsymbol{1}^\top )\circ M) \Psi}_F \norm{\Psi (\Psi^\top \Psi)^{-1}}.
$$
By Assumption \ref{asp:basis}, we know $\norm{\Phi (\Phi^\top \Phi)^{-1}} \lesssim \frac{1}{\sqrt{N}}$ and $\norm{\Psi (\Psi^\top \Psi)^{-1}} \lesssim \frac{1}{\sqrt{T}}$. In addition, because
\begin{align*}
\bbE\left[ \left. \norm{\Phi^\top (( \Omega - p \boldsymbol{1}\boldsymbol{1}^\top )\circ M) \Psi}_F^2 \right|  M, X, Z \right] &= 
 \sum_{j_1,j_2} \bbE\left[ \left. \left(\sum_{it}(\omega_{it} - p) m_{it} \phi_{i,j_1} \psi_{t,j_2}\right)^2  \right|  M, X, Z \right] \\
&= \sum_{j_1,j_2} \sum_{it}\bbE[(\omega_{it} - p)^2 ] m^2_{it} \phi^2_{i,j_1} \psi^2_{t,j_2} \\
& \leq p \sum_{j_1,j_2} \sum_{it} m^2_{it} \phi^2_{i,j_1} \psi^2_{t,j_2}
 = p \sum_{it} m^2_{it} \norm{\phi_i}^2 \norm{\psi_t}^2 \\
& = O_p(p N T J^2 d_1 d_2)
\end{align*}
by the assumption that $\bbE[m_{it}^4]$, $\bbE[\phi_{ij}^4]$, and $\bbE[\psi_{tj}^4]$ are bounded, we have $ \norm{\Phi^\top (( \Omega - p \boldsymbol{1}\boldsymbol{1}^\top )\circ M) \Psi}_F = O_p( \sqrt{p N T} J\sqrt{d_1d_2})$. Hence, we have 
$$
\norm{ \frac{1}{p} P_X (( \Omega - p \boldsymbol{1}\boldsymbol{1}^\top )\circ M) P_Z }_F = O_p\left( \frac{J\sqrt{d_1d_2}}{\sqrt{p}} \right) .
$$
Moreover, note that
$$
\norm{ \frac{1}{\widehat{p}} P_X (\Omega \circ E)  P_Z }_F \leq  \frac{1}{\widehat{p}} \norm{\Phi (\Phi^\top \Phi)^{-1}} \norm{\Phi^\top (\Omega \circ E) \Psi}_F \norm{\Psi (\Psi^\top \Psi)^{-1}}.
$$
By Assumption \ref{asp:basis}, we know $\norm{\Phi (\Phi^\top \Phi)^{-1}} \lesssim \frac{1}{\sqrt{N}}$ and $\norm{\Psi (\Psi^\top \Psi)^{-1}} \lesssim \frac{1}{\sqrt{T}}$. In addition, because
\begin{align*}
\bbE\left[ \left. \norm{\Phi^\top (\Omega \circ E) \Psi}_F^2 \right|  M, X, Z \right] &= 
 \sum_{j_1,j_2} \bbE\left[ \left. \left(\sum_{it}\omega_{it}\epsilon_{it}  \phi_{i,j_1} \psi_{t,j_2}\right)^2  \right|  M, X, Z \right] \\
&= \sum_{j_1,j_2} \sum_{it}\bbE[(\omega_{it}\epsilon_{it})^2 ] \phi^2_{i,j_1} \psi^2_{t,j_2} \\
& \leq p \sigma^2 \sum_{j_1,j_2} \sum_{it} \phi^2_{i,j_1} \psi^2_{t,j_2}
 = p  \sigma^2 \sum_{it} \norm{\phi_i}^2 \norm{\psi_t}^2 \\
& = O_p(p\sigma^2  N T J^2 d_1 d_2)
\end{align*}
by the assumption that $\bbE[\phi_{ij}^4]$, and $\bbE[\psi_{tj}^4]$ are bounded, we have $ \norm{\Phi^\top (\Omega \circ E) \Psi}_F = O_p( \sigma \sqrt{p N T} J\sqrt{d_1d_2})$ and 
$$
\norm{ \frac{1}{\widehat{p}} P_X (\Omega \circ E)  P_Z }_F = O_p\left( \frac{ \sigma  J\sqrt{d_1 d_2}}{\sqrt{p}} \right).
$$
In addition, from \eqref{eq:M_1_est}, we have
\begin{align*}
Y - \widehat{M}_1 &= G_2 V_1^\top P_Z^\perp+ P_X^\perp W_1 V_2^\top + W_2 V_2^\top - P_X W_2 V_2^\top P_Z+ E\\
& \ \  - \left( \frac{p}{\widehat{p}} - 1\right) P_X  M P_Z- \frac{1}{p} P_X (( \Omega - p \boldsymbol{1}\boldsymbol{1}^\top )\circ M) P_Z - \frac{1}{\widehat{p}} P_X (\Omega \circ E) P_Z  - \textit{smoothing error\_1}.
\end{align*}
Define
$$
S' = \left( \frac{p}{\widehat{p}} - 1\right) P_X  M P_Z + \frac{1}{p} P_X (( \Omega - p \boldsymbol{1}\boldsymbol{1}^\top )\circ M) P_Z + \frac{1}{\widehat{p}} P_X (\Omega \circ E) P_Z  + \textit{smoothing error\_1} .
$$
Then, by the above analysis, we have
\begin{align*}
 & \norm{\Omega \circ S'}  \leq \norm{\Omega \circ S'}_F   \leq \norm{S'}_F \\
 &\lesssim  (\log J)^{\frac{1}{10}} \left( \frac{ \sigma  J\sqrt{d_1 d_2}}{\sqrt{p}}  + \frac{  J\sqrt{d_1 d_2}}{\sqrt{p}} \right) + \frac{\sqrt{NT} K_1}{J^{\gamma_1^G/d_1}} + \frac{\sqrt{NT} K_1}{J^{\gamma_1^Q/d_2}} + \frac{\sqrt{NT} K_2}{J^{\gamma_2^G/d_1}} + \frac{\sqrt{NT} K_3}{J^{\gamma_2^Q/d_2}}  , 
\end{align*}
with high probability.
In addition, we have $\norm{\Omega \circ E} \lesssim \sigma p^{1/2}  \sqrt{N+T}$ with high probability by Lemma \ref{lem:orthogonal_oper_norm} since $\norm{\omega_{it} \epsilon_{it}}_{\psi_2} \leq p^{1/2} \sigma$. Hence, 
$$
2 C_\lambda \norm{\Omega \circ (E - S')} \leq \nu 
$$ 
with high probability, where $C_\lambda > 1$ is in Lemma \ref{lem:nuclear_missing}. Then, by setting $S = E - S'$ and $L = M_2 P_Z^\perp+ P_X^\perp M_3 + M_4 - P_X M_4 P_Z$, we can get from Lemma \ref{lem:nuclear_missing} that
\begin{align*}
&\norm{\widehat{M}_{rest} - (M_2 P_Z^\perp+ P_X^\perp M_3 + M_4 - P_X M_4 P_Z)}_F  \\
&= O_p\left( \sqrt{r^*}   \min \left\{ \frac{\nu}{p} + \frac{\sqrt{N + T} M_{\max}}{\sqrt{p}} , \norm{M_2 P_Z^\perp+ P_X^\perp M_3 + M_4 - P_X M_4 P_Z}_F \right\} \right)    .
\end{align*}
Then, because
\begin{align*}
\norm{\widehat{M}_{1} + \widehat{M}_{rest} - M}_F &\leq
\norm{\widehat{M}_{rest} - (M_2 P_Z^\perp+ P_X^\perp M_3 + M_4 - P_X M_4 P_Z)}_F \\
&\ \ + \norm{\left( \frac{p}{\widehat{p}} - 1\right) P_X  M P_Z  +  \frac{1}{p} P_X (( \Omega - p \boldsymbol{1}\boldsymbol{1}^\top )\circ M) P_Z + \frac{1}{\widehat{p}} P_X (\Omega \circ E) P_Z  }_F \\
&\ \ + \norm{\textit{smoothing error\_1}}_F,
\end{align*}
and
\begin{align*}
&\norm{ \left( \frac{p}{\widehat{p}} - 1\right) P_X  M P_Z  + \frac{1}{p} P_X (( \Omega - p \boldsymbol{1}\boldsymbol{1}^\top )\circ M) P_Z + \frac{1}{\widehat{p}} P_X (\Omega \circ E) P_Z  }_F = O_p\left(\frac{ \sigma  J\sqrt{d_1 d_2}}{\sqrt{p}}  + \frac{  J\sqrt{d_1 d_2}}{\sqrt{p}} \right),\\
&\norm{\textit{smoothing error\_1}}_F = O_p\left( \frac{\sqrt{NT} K_1}{J^{\gamma_1^G/d_1}} + \frac{\sqrt{NT} K_1}{J^{\gamma_1^Q/d_2}} + \frac{\sqrt{NT} K_2}{J^{\gamma_2^G/d_1}} + \frac{\sqrt{NT} K_3}{J^{\gamma_2^Q/d_2}} \right),
\end{align*}
as noted in \eqref{eq:smoothing_1}, we have the desired result. $\square$

\subsubsection{Proof of Corollaries \ref{corr:convergence_rate_obs_sigma_2} and \ref{corr:convergence_rate_mar_sigma_2}}

(i) Proof of Corollary \ref{corr:convergence_rate_obs_sigma_2}: It is basically the same as that of Theorem \ref{thm:convergence_rate_obs}. For $\widehat{M}_2$, we set $S = P_X E P_Z^\perp$ and $L= M_2P_Z^\perp + P_X M_4 P_Z^\perp + \textit{smoothing error\_2}$ to use Lemma \ref{lem:nuclear} and follow the proof method of Theorem \ref{thm:convergence_rate_obs}. Similarly, for $\widehat{M}_3$, we set $S = P_X E P_Z^\perp$ and $L= P_X^\perp M_3 + P_X^\perp M_4 P_Z + \textit{smoothing error\_3}$, and for $\widehat{M}_4$, we set $S = P_X^\perp E P_Z^\perp$ and $L= P_X^\perp M_4 P_Z^\perp + \textit{smoothing error\_4}$ to use Lemma \ref{lem:nuclear}. Then, we have the desired result.\\
(ii) Proof of Corollary \ref{corr:convergence_rate_mar_sigma_2}: It is basically the same as that of Theorem \ref{thm:convergence_rate_mar}. For $\widehat{M}_{rest}$, we set $S = E$ and $L = M_2 P_Z^\perp+ P_X^\perp M_3 + M_4 - P_X M_4 P_Z - S'$ to use Lemma \ref{lem:nuclear_missing} and follow the proof method of Theorem \ref{thm:convergence_rate_mar}. $\square$

\subsection{Auxiliary lemmas}

Consider the following generic model, $Z = L + S$, where $rank(L) = r_L$. Denote the nuclear norm penalized estimator by
$$
\widehat{L} \coloneqq \arg\min_A \norm{Z - A}_F^2 + \lambda \norm{A}_*.
$$
Then, we have the following bound for the estimator.
\begin{lemma}\label{lem:nuclear}
Let $\lambda/2 \geq C_\lambda \norm{S}$ for some constant $C_\lambda > 1$. Then, we have
$$
\norm{\widehat{L} - L}_F \lesssim \min\left\{ \sqrt{r_L}\lambda , \sqrt{r_L} \norm{L}_F \right\}.
$$
In addition, if $L = 0$, then $\widehat{L} = 0$.
\end{lemma}

\paragraph{Proof.}

Let $\Delta = \widehat{L} - L$. Then, we have
$$
\norm{Z - \widehat{L}}_F^2 = \norm{S - \Delta}_F^2 = \norm{S}_F^2 + \norm{\Delta}_F^2 - 2 tr(\Delta^\top S).
$$
In addition, for some constant $0 < c < 1$, we have
$$
\abs{2 tr(\Delta^\top S)} \leq 2 \norm{\Delta}_* \norm{S} \leq (1-c) \lambda \norm{\Delta}_* 
$$
since $\lambda \geq \frac{2}{1-c} \norm{S} $. Then, we have
\begin{align}\label{eq:key}
\nonumber &\norm{Z - \widehat{L}}_F^2 + \lambda \norm{\widehat{L}}_* \leq \norm{Z - L}_F^2 + \lambda \norm{L}_* , \\
\nonumber & \norm{\Delta}_F^2 - 2 tr(\Delta^\top S) + \lambda \norm{\widehat{L}}_* \leq \lambda \norm{L}_* , \\
& \norm{\Delta}_F^2 - (1-c) \lambda \norm{\Delta}_*  + \lambda \norm{\widehat{L}}_* \leq \lambda \norm{L}_* .
\end{align}
(i) When $L=0$.\\
Since $\widehat{L} = \Delta $ and $\norm{L}_* = 0$, we have by \eqref{eq:key} that
$$
\norm{\Delta}_F^2 + 2 c \lambda \norm{\Delta}_*  \leq 0.
$$
Since $c > 0$, we have $\norm{\Delta}_F = 0$.\\
(ii) When $L \neq 0$.\\
Note that 
$$
\norm{\widehat{L}}_* = \norm{\Delta + L}_* \geq \norm{\Delta}_* - \norm{L}_* .
$$
Hence, we have by \eqref{eq:key} that
\begin{align*}
 &\norm{\Delta}_F^2 - (1-c) \lambda \norm{\Delta}_* +\lambda \norm{\Delta}_* - \lambda \norm{L}_*  \leq \norm{\Delta}_F^2 - (1-c) \lambda \norm{\Delta}_*  + \lambda \norm{\widehat{L}}_* \leq \lambda \norm{L}_*,\\
 & \norm{\Delta}_F^2 + c \lambda \norm{\Delta}_*  - \lambda \norm{L}_*  \leq \lambda  \norm{L}_*,\\
 & \norm{\Delta}_F^2 + c \lambda \norm{\Delta}_* \leq 2 \lambda \norm{L}_* .
\end{align*}
So, we have $\norm{\Delta}_* \leq \frac{2}{c} \norm{L}_*$. Then, we have $\norm{\Delta}_F \leq \norm{\Delta}_* \leq \frac{2}{c} \norm{L}_* \leq \frac{2}{c} \sqrt{r_L} \norm{L}_F$.\\
Next, we derive the bound of $\sqrt{r_L}\lambda$. Denote the singular value decomposition of $L$ by $L = UDV^\top$ where $U = (U_o, U_c)$ and $V = (V_o, V_c)$. Here, $(U_c,V_c)$ are the columns of $U$, $V$ that correspond to the zero singular values, while $(U_o,V_o)$ denote the columns of $U$, $V$ associated with the nonzero singular values. In addition, let
$$
\calP(A) =  U_c U_c^\top A V_c V_c^\top , \qquad  \calM(A) =  A - \calP(A).
$$
Note that
\begin{align}\label{eq:key_both}
\norm{\widehat{L}}_* &= \norm{L + \Delta}_* = \norm{L + \calP(\Delta) + \calM(\Delta) }_* \\
\nonumber & \geq \norm{L + \calP(\Delta)  }_* - \norm{\calM(\Delta)}_* 
= \norm{L}_* + \norm{\calP(\Delta)  }_* - \norm{\calM(\Delta)}_* .
\end{align} 
So, using this relation with \eqref{eq:key} and the fact that $(1-c) \lambda \norm{\Delta}_* \leq (1-c) \lambda \norm{\calP(\Delta)}_* + (1-c) \lambda \norm{\calM(\Delta)}_*$, we have $ \norm{\Delta}_F^2 + c \lambda \norm{\calP(\Delta)}_* \leq (2-c) \lambda \norm{\calM(\Delta)}_*$.
Therefore, we have
\begin{align*}
\norm{\Delta}_F^2 &\leq (2-c) \lambda \norm{\calM(\Delta)}_* 
\leq \lambda \norm{\calM(\Delta)}_F \sqrt{2r_L}
\leq \lambda \norm{\Delta}_F \sqrt{2r_L}. \ \ \square
\end{align*}
\smallskip

On the other hand, if we can only observe $\Omega \circ Z$ instead of $Z$ where $\Omega = (\omega_{it})_{i \leq N, t \leq T}$ and $\omega_{it} = 1\{z_{it} \text{ is observed} \}$, then the nuclear norm penalized estimator becomes
$$
\widehat{L} \coloneqq \arg\min_{A \in \calA} \norm{\Omega \circ (Z - A)}_F^2 + \lambda \norm{A}_*,
$$
where $\calA = \{A: \norm{A}_{\infty} \leq L_{\max}\}$. Then, we have the following bound for the estimator.
\begin{lemma}\label{lem:nuclear_missing}
Let $\lambda/2 \geq C_\lambda \norm{\Omega \circ S}$ for some constant $C_\lambda >1$. Then, if $\norm{L}_{\infty} \leq L_{\max}$, with probability converging to $1$, we have
$$
\norm{\widehat{L} - L}_F \lesssim \min\left\{ \frac{\sqrt{r_L}\lambda}{p_{\min}} + \frac{\sqrt{r_L \theta_p (N+T)}L_{\max}}{\sqrt{p_{\min}}} , \sqrt{r_L} \norm{L}_F \right\},
$$
where $p_{\min} \leq p_i \leq p_{\max}$ and $\theta_p =  p_{\max}/p_{\min}$.
\end{lemma}

\noindent \textbf{Proof.} (i) Let $\Delta = \widehat{L} - L$. Then, we have
$$
\norm{\Omega \circ (Z - \widehat{L})}_F^2  = \norm{\Omega \circ (S - \Delta)}_F^2 = \norm{\Omega \circ S}_F^2 + \norm{\Omega \circ\Delta}_F^2 - 2 tr((\Omega \circ\Delta)^\top (\Omega \circ S)).
$$
In addition, for some constant $0 < c < 1$, we have
$$
\abs{ 2 tr((\Omega \circ\Delta)^\top (\Omega \circ S))} = \abs{ 2 tr( \Delta^\top (\Omega \circ S))}   \leq 2 \norm{\Delta}_* \norm{\Omega \circ S} \leq (1-c) \lambda \norm{\Delta}_* 
$$
since $\lambda \geq \frac{2}{1-c} \norm{\Omega \circ S} $. Then, we have
\begin{align}\label{eq:key_missing}
\nonumber &\norm{\Omega \circ (Z - \widehat{L})}_F^2 + \lambda \norm{\widehat{L}}_* \leq \norm{\Omega \circ (Z - L) }_F^2 + \lambda \norm{L}_* , \\
\nonumber & \norm{\Omega \circ \Delta}_F^2 - 2 tr((\Omega \circ \Delta)^\top(\Omega \circ  S)) + \lambda \norm{\widehat{L}}_* \leq \lambda \norm{L}_* , \\
& \norm{\Omega \circ \Delta}_F^2 - (1-c) \lambda \norm{\Delta}_*  + \lambda \norm{\widehat{L}}_* \leq \lambda \norm{L}_* .
\end{align}
Note that 
$$
\norm{\widehat{L}}_* = \norm{\Delta + L}_* \geq \norm{\Delta}_* - \norm{L}_* .
$$
Hence, we have by \eqref{eq:key_missing} that
\begin{align*}
 &\norm{\Omega \circ \Delta}_F^2 - (1-c) \lambda \norm{\Delta}_* +\lambda \norm{\Delta}_* - \lambda \norm{L}_*  \leq \norm{\Omega \circ \Delta}_F^2 - (1-c) \lambda \norm{\Delta}_*  + \lambda \norm{\widehat{L}}_* \leq \lambda \norm{L}_*,\\
 & \norm{\Omega \circ \Delta}_F^2 + c \lambda \norm{\Delta}_* \leq 2 \lambda \norm{L}_* .
\end{align*}
So, we have $\norm{\Delta}_* \leq \frac{2}{c} \norm{L}_*$. Then, we have $\norm{\Delta}_F \leq \norm{\Delta}_* \leq \frac{2}{c} \norm{L}_* \leq \frac{2}{c} \sqrt{r_L} \norm{L}_F$.\\
(ii) Using the relations \eqref{eq:key_both} and \eqref{eq:key_missing} with the fact that $(1-c) \lambda \norm{\Delta}_* \leq (1-c) \lambda \norm{\calP(\Delta)}_* + (1-c) \lambda \norm{\calM(\Delta)}_*$, we have 
$$ 
c \lambda \norm{\calP(\Delta)}_*  \leq \norm{\Omega \circ \Delta}_F^2 + c \lambda \norm{\calP(\Delta)}_* \leq (2-c) \lambda \norm{\calM(\Delta)}_* .
$$
Hence, we have $\norm{\calP(\Delta)}_* \leq \frac{2-c}{c} \norm{\calM(\Delta)}_*$. In addition, we know $\norm{\Delta}_\infty \leq 2 L_{\max}$. Set $\bar{L} = 2 L_{\max}$ and $C_1 = \frac{2-c}{c}$. If $\norm{\Delta}_F^2 > 2 B \frac{\bar{L}^2}{p_{\min}}\sqrt{NT}$ for some sufficiently large $B >0$, then $\Delta \in \calC \left(C_1 , 2 B \frac{\bar{L}^2}{p_{\min}} \right)$ and we have with high probability that
$$
p_{\min} \norm{\Delta}_F^2 < 2 \norm{\Omega \circ \Delta}_F^2 + 8 B r_L \theta_p (N+T) L_{\max}^2
$$
by Lemma \ref{lem:rsc}. Then, since $\norm{\Omega \circ \Delta}_F^2 \leq (2-c) \lambda \norm{\calM(\Delta)}_* $, we have
\begin{align*}
p_{\min} \norm{\Delta}_F^2 &< 2(2-c)  \lambda \norm{\calM(\Delta)}_* + 8 B r_L \theta_p (N+T) L_{\max}^2   \\
& < 2(2-c)  \lambda \sqrt{2r_L} \norm{\Delta}_F + 8 B r_L \theta_p (N+T) L_{\max}^2 .
\end{align*}
If the first term dominates the second term, we have
$$
p_{\min} \norm{\Delta}_F^2 < 4(2-c)  \lambda \sqrt{2r_L} \norm{\Delta}_F \quad \Longrightarrow \quad
\norm{\Delta}_F < 4(2-c)  \lambda \sqrt{2r_L} / p_{\min}.
$$
If the second term dominates the first term, we have
$$
p_{\min} \norm{\Delta}_F^2 < 16 B r_L \theta_p (N+T) L_{\max}^2 \quad \Longrightarrow \quad
\norm{\Delta}_F < 4 B^{1/2} \sqrt{r_L \theta_p (N+T)} L_{\max} / \sqrt{p_{\min}}.
$$
In addition, if $\norm{\Delta}_F^2 \leq 2 B \frac{\bar{L}^2}{p_{\min}}\sqrt{NT}$, we have
$$
\norm{\Delta}_F \leq 2 \sqrt{2} B^{1/2} \frac{L_{\max}\sqrt{N+T}}{\sqrt{p_{\min}}} .
$$
Hence, with probability converging to 1, we have 
$$
\norm{\widehat{L} - L}_F \lesssim  \frac{\sqrt{r_L}\lambda}{p_{\min}} + \frac{\sqrt{r_L \theta_p (N+T)}L_{\max}}{\sqrt{p_{\min}}}  . \ \ \square
$$
\smallskip

\begin{lemma}\label{lem:orthogonal_oper_norm}
Let $E$ be a $N \times T$ matrix of independent sub-Gaussian entries such that $\norm{\epsilon_{it}}_{\psi_2} \leq \sigma $. In addition, let $L$ and $R$ be $N \times J_1$ and $T \times J_2$ orthonormal matrices, respectively. Then, we have
$$
\norm{L^\top E R} \lesssim \sigma \sqrt{J_1 + J_2 + \log(\max\{N,T\})} ,
$$
with probability at least $1 - O(\max\{N,T\}^{-5})$.
\end{lemma}

\noindent \textbf{Proof.}
First, by Corollary 4.2.13 of \cite{vershynin2018high}, we can find an $1/4$-net $\calN$ of the unit sphere $S^{J_1 -1}$ and $1/4$-net $\calM$ of the unit sphere $S^{J_2 -1}$ with cardinalities
$$
\abs{\calN} \leq 9^{J_1}, \quad \abs{\calM} \leq 9^{J_2}.
$$
Note that 
$$
\norm{L^\top E R} \leq 2 \max_{a \in \calN, b \in \calM} \sum_{i = 1}^N \sum_{t = 1}^T \epsilon_{it} (a^\top L_i) (b^\top R_t)
$$
where $L_i^\top$ is the $i$-th row of $L$ and $R_t^\top$ is the $t$-th row of $R$ (see, exercise 4.4.3 of \cite{vershynin2018high}). Fix $a_o \in \calN$ and $b_o \in \calM$. Then, by Hoeffding's inequality with the independent sub-Gaussian assumption, we have with probability at least $1 - u$,
$$
\sum_{i = 1}^N \sum_{t = 1}^T \epsilon_{it} (a_o^\top L_i) (b_o^\top R_t) \lesssim \sigma \norm{L a_o}_2 \norm{R b_o}_2 \sqrt{\log(u^{-1})} = \sigma \sqrt{\log(u^{-1})},
$$
and by setting $u = \max\{N,T\}^{-5} 9^{-J_1 - J_2}$, we have
with probability at least $1 - \max\{N,T\}^{-5} 9^{-J_1 - J_2}$,
$$
\sum_{i = 1}^N \sum_{t = 1}^T \epsilon_{it} (a_o^\top L_i) (b_o^\top R_t) \lesssim \sigma \sqrt{J_1 + J_2 + \log(\max\{N,T\})},
$$
Then, because
\begin{align*}
P\left(  \max_{a \in \calN, b \in \calM} \sum_{i = 1}^N \sum_{t = 1}^T \epsilon_{it} (a^\top L_i) (b^\top R_t) > c \right) &\leq \sum_{a \in \calN , b \in \calM} P\left(  \sum_{i = 1}^N \sum_{t = 1}^T \epsilon_{it} (a^\top L_i) (b^\top R_t) > c \right)  \\
& \leq 9^{J_1 + J_2} P\left(  \sum_{i = 1}^N \sum_{t = 1}^T \epsilon_{it} (a^\top L_i) (b^\top R_t) > c \right),
\end{align*}
by setting $c = \sigma \sqrt{J_1 + J_2 + \log(\max\{N,T\})}$, we have
\begin{align*}
P\left(  \max_{a \in \calN, b \in \calM} \sum_{i = 1}^N \sum_{t = 1}^T \epsilon_{it} (a^\top L_i) (b^\top R_t) >  \sigma \sqrt{J_1 + J_2 + \log(\max\{N,T\})} \right) \lesssim \max\{N,T\}^{-5}. \ \ \square
\end{align*}
\smallskip

\begin{lemma}\label{lem:error_oper_norm}
With probability converging to $1$, we have\\
(i) $\norm{P_X E P_Z} \lesssim \sigma \sqrt{J(d_1 + d_2)}$; (ii) $\norm{P_X E} \lesssim \sigma \sqrt{T}$; (iii) $\norm{E P_Z} \lesssim \sigma \sqrt{N}$; (iv) $\norm{E } \lesssim \sigma \sqrt{N + T}$.
\end{lemma}

\noindent \textbf{Proof.} (i) By Assumption \ref{asp:basis}, with probability converging to $1$, $\norm{\Phi (\Phi^\top \Phi)^{-1/2}}$ and $\norm{\Psi (\Psi^\top \Psi)^{-1/2}}$ are bounded. Hence, by Lemma \ref{lem:orthogonal_oper_norm}, with probability converging to $1$, we have
$$
\norm{P_X E P_Z}  \leq \norm{\Phi (\Phi^\top \Phi)^{-1/2}} \norm{((\Phi^\top \Phi)^{-1/2} \Phi^\top) E (\Psi (\Psi^\top \Psi)^{-1/2})} \norm{(\Psi^\top \Psi)^{-1/2} \Psi^\top} 
\lesssim \sigma \sqrt{J (d_1 + d_2)}   
$$
because $\Phi (\Phi^\top \Phi)^{-1/2}$ and $\Psi (\Psi^\top \Psi)^{-1/2}$ are $N \times Jd_1$ and $T \times Jd_2$ orthogonal matrices, respectively, and $\log(\max\{N,T\}) \lesssim J(d_1 + d_2)$.\\
(ii) By Lemma \ref{lem:orthogonal_oper_norm}, with probability converging to $1$, we have
$$
\norm{P_X E}  \leq \norm{\Phi (\Phi^\top \Phi)^{-1/2}} \norm{((\Phi^\top \Phi)^{-1/2} \Phi^\top) E I_T} \lesssim \sigma \sqrt{T}   
$$
because $\Phi (\Phi^\top \Phi)^{-1/2}$ and $I_T$ are $N \times Jd_1$ and $T \times T$ orthogonal matrices, respectively, and  $J d_1 \lesssim T$.\\
(iii) The proof is symmetric to that of (ii).\\
(iv) It follows from Lemma \ref{lem:orthogonal_oper_norm} with $L = I_N$ and $R = I_T$. $\square$
\bigskip

\begin{lemma}\label{lem:rsc}
Define the restricted set of directions as
$$
\calC(c_1,c_2) = \left\{ A \in \calA^* : \norm{\calP(A)}_* \leq c_1 \norm{\calM(A)}_*, \norm{A}_F^2 > c_2 \sqrt{NT} \right\},
$$
where $\calA^* = \{A \in \bbR^{N\times T} : \norm{A}_{\max} \leq \bar{L} \}$. Then, for any $C_1 > 0$ and sufficiently large $B >0$, we have, with probability converging to $1$, that uniformly for $A \in \calC \left(C_1, 2B\frac{\bar{L}^2}{p_{\min}} \right)$, 
$$
\norm{\Omega \circ A}_F^2   > 0.5 \sum_{it} p_i A_{it}^2 - B r_L  \theta_p (N + T) \bar{L}^2  \geq  0.5 p_{\min} \norm{A}_F^2 - B r_L  \theta_p (N + T) \bar{L}^2 .
$$

\end{lemma}

\noindent \textbf{Proof.} It is an extension of Lemma A.2 of \cite{chernozhukov2023inference}. Here, we condition on $X$ and denote the conditional expectation and conditional probability by $\bbE_{| X}$ and $P_{| X}$, respectively. First, let $\Omega(A) = \norm{\Omega \circ A}_F^2 = \sum_{it} \omega_{it}^2 A_{it}^2$. Then, we have $\bbE_{| X} \Omega(A) \geq p_{\min} \sum_{it} A_{it}^2 \geq p_{\min} \norm{A}_F^2$. In addition, define
$$
\calE(A) = \left\{ \abs{\Omega(A) - \bbE_{| X} \Omega(A) } > 0.5 \cdot \bbE_{| X} \Omega(A) + B r_L \theta_p (N + T) \bar{L}^2 \right\}.
$$ 
Then, we want to show that $P_{| X} \left( \exists A \in \calC \left(C_1, 2B\frac{\bar{L}^2}{p_{\min}} \right) : \calE(A) \text{ holds} \right) \rightarrow 0$. To use the standard peeling argument, define
$$
\Gamma_l = \left\{ A \in \calC \left(C_1, 2B \frac{\bar{L}^2}{p_{\min}} \right) : 2^l v_n \leq \bbE_{| X} \Omega(A) \leq 2^{l+1} v_n \right\}
$$
where $v_n = B \bar{L}^2 \sqrt{NT}$ and $l \in \bbN$.

\paragraph{Part 1.} We want to show $\calC \left(C_1, 2B\frac{\bar{L}^2}{p_{\min}} \right) \subset \cup_{l=1}^\infty \Gamma_l$. If $A \in \calC \left(C_1, 2B\frac{\bar{L}^2}{p_{\min}} \right)$, we have
$$
\bbE_{| X} \Omega(A) \geq p_{\min} \norm{A}_F^2 \geq  2 B \bar{L}^2 \sqrt{NT}  = 2 v_n .
$$
Hence, there is $l \in \bbN$ such that $A \in \Gamma_l$ for any $A \in \calC \left(C_1, 2B\frac{\bar{L}^2}{p_{\min}} \right)$.

\paragraph{Part 2.} Let
\begin{align*}
&\calD(x) = \left\{ A \in \calC \left(C_1, 2B \frac{\bar{L}^2}{p_{\min}} \right) : \norm{A}_F^2 \leq x \right\},\\
&\calF(A) = \left\{ \abs{\Omega(A) - \bbE_{| X} \Omega(A) } - B r_L \theta_p (N + T) \bar{L}^2 > 0.25 \cdot 2^{l+1} v_n  \right\}.
\end{align*}
We want to show that if $A \in \Gamma_l$ and $\calE(A)$ holds, than $A \in \calD(x_l)$ where $x_l = p_{\min}^{-1} 2^{l+1} v_n$ and $\calF(A)$ holds. This is because
$$
\abs{\Omega(A) - \bbE_{| X} \Omega(A) } - B r_L  \theta_p (N + T) \bar{L}^2  > 0.5 \cdot \bbE_{| X} \Omega(A) \geq 0.25 \cdot 2^{l+1} v_n,
$$
and $\norm{A}_F^2 \leq p_{\min}^{-1} \bbE_{| X} \Omega(A) \leq  p_{\min}^{-1} 2^{l+1} v_n$.

\paragraph{Part 3.} Let
$$
Q(x) = \sup_{A \in \calD(x)} \abs{ \frac{1}{NT} \sum_{it} \omega_{it}^2 A_{it}^2 - \bbE_{| X}\omega_{it}^2 A_{it}^2 }.
$$
We bound $\bbE_{| X} Q(x)$. First, note that for any $A \in \calD(x) \subset \calC \left(C_1, 2B \frac{\bar{L}^2}{p_{\min}} \right) $, we have 
\begin{align*}
\norm{A}_*  &=\norm{\calP(A) + \calM(A)}_* \leq (1 + C_1)\norm{\calM(A)}_* \leq (1 + C_1) \sqrt{r_L} \norm{\calM(A)}_F \\
&\leq (1 + C_1) \sqrt{r_L} \norm{A}_F \leq (1 + C_1) \sqrt{r_L x}.
\end{align*}
Let $u_{it}$ be an i.i.d. Rademacher random variable. Then, $\bbE_{| X} \norm{\Omega_u} \lesssim p_{\max}^{1/2} \sqrt{N + T}$ where $\Omega_u = (\omega_{it} u_{it})_{N \times T}$. Hence, by using the symmetrization argument with the concentration inequality (e.g., (2.3) of \cite{koltchinskii2011oracle}), we have 
\begin{align*}
\bbE_{| X} Q(x) &\leq 2 \bbE_{| X} \sup_{A \in \calD(x)} \abs{ \frac{1}{NT} \sum_{it} \omega_{it}^2 A_{it}^2 u_{it} } \leq c_{3}
\bar{L} \bbE_{| X}  \sup_{A \in \calD(x)} \abs{ \frac{1}{NT} \sum_{it} \omega_{it} A_{it} u_{it} } \\
& = c_{3} \bar{L} \bbE_{| X}  \sup_{A \in \calD(x)} \abs{ \frac{1}{NT} tr(\Omega_u A^\top) } 
\leq c_{3} \bar{L} \bbE_{| X}  \sup_{A \in \calD(x)}  \frac{1}{NT} \norm{\Omega_u} \norm{A}_*  \\
& \leq c_{4} p_{\max}^{1/2} \frac{\sqrt{N+T}}{NT} \bar{L} \sup_{A \in \calD(x)} \norm{A}_* 
\leq c_{5}  p_{\max}^{1/2} \frac{\sqrt{N+T}}{NT} \bar{L} \sqrt{r_L x} \\
& =  2 c_{5} \sqrt{8} \bar{L} \frac{\sqrt{r_L \theta_p (N+T)}}{\sqrt{NT}} \times \sqrt{\frac{p_{\min}}{32 NT}x} 
\leq  \frac{p_{\min}}{32NT}x + 32 c_{5}^2 \frac{\bar{L}^2 r_L \theta_p (N+T)}{NT} \\
& \leq \frac{p_{\min}}{32NT}x + B \frac{\bar{L}^2 r_L \theta_p (N+T)}{NT},
\end{align*}
for sufficiently large $B > 0$. Here, we use the fact that $\omega_{it}$ is bounded by $1$ and $A_{it}$ is bounded by $\bar{L}$.

\paragraph{Part 4.} Next, we bound the tail probability of $Q(x) - \bbE_{| X} Q(x)$. Since $A_{it}^2/\bar{L}^2$ is bounded, we can use the Massart inequality (e.g., Theorem 14.2 of \cite{buhlmann2011statistics}) to have
$$
P_{| X} \left(\frac{1}{\bar{L}^2}Q(x) > \frac{1}{\bar{L}^2} \bbE_{| X} Q(x) + t \right) \leq \exp (- c_{6} NT t^2).
$$
Set $t = \frac{7xp_{\min}}{32 \bar{L}^2 NT}$. Then, because 
$$
\frac{\bbE_{| X} Q(x)}{\bar{L}^2} \leq \frac{p_{\min}}{32\bar{L}^2 NT}x + B \frac{r_L \theta_p (N+T)}{NT},
$$ 
we have
\begin{align*}
P_{| X} \left(Q(x) > B \frac{r_L \theta_p (N+T)}{NT}\bar{L}^2 + 0.25 \cdot \frac{xp_{\min}}{ NT} \right) &= P_{| X} \left(\frac{1}{\bar{L}^2}Q(x) > B \frac{r_L  \theta_p (N+T)}{NT} + 0.25 \cdot \frac{xp_{\min}}{\bar{L}^2 NT} \right) \\
&\leq \exp \left(- c_{7} \frac{p_{\min}^2 x^2}{\bar{L}^4 N T} \right) .
\end{align*}

\paragraph{Part 5.} Finally, we use the pealing argument. Note that
\begin{align*}
&P_{| X} \left( \exists A \in \calC \left(C_1, 2B\frac{\bar{L}^2}{p_{\min}} \right) : \calE(A) \text{ holds} \right) \leq \sum_{l=1}^\infty 
P_{| X} \left( \exists A \in \Gamma_l : \calE(A) \text{ holds} \right) \\
& \leq  \sum_{l=1}^\infty  P_{| X} \left( \exists A \in \calD(x_l) : \calF(A) \text{ holds} \right) \\
& \leq \sum_{l=1}^\infty P_{| X} \left( \sup_{A \in \calD(x_l)} \abs{\Omega(A) - \bbE_{| X} \Omega(A) } > B r_L \theta_p (N + T) \bar{L}^2 + 0.25 p_{\min} x_l \right) \\
& = \sum_{l=1}^\infty P_{| X} \left( Q(x_l) > B \frac{r_L \theta_p (N + T)}{NT}\bar{L}^2  + 0.25 \frac{p_{\min} x_l}{NT} \right)  \leq  \sum_{l=1}^\infty  \exp \left(- c_{7} \frac{p_{\min}^2 x_l^2}{\bar{L}^4 N T} \right) \\
& = \sum_{l=1}^\infty  \exp \left(- c_{7} 4^{l+1} B^2 \right)
\leq \frac{\exp(- 16 c_{7} B^2)}{1 - \exp(- 16 c_{7} B^2)} < \varepsilon
\end{align*}
for any $\varepsilon > 0$ and sufficiently large $B$. Here, we use the relations $x_l^2 = p_{\min}^{-2} 4^{l+1} v_n^2$ and $v_n^2 = B^2 \bar{L}^4 NT$. Therefore, we have, with probability converging to 1, that uniformly for all $A \in \calC \left(C_1, 2B\frac{\bar{L}^2}{p_{\min}} \right)$, 
$$
\abs{\Omega(A) - \bbE_{| X} \Omega(A) } \leq 0.5 \cdot \bbE_{| X} \Omega(A) + B r_L \theta_p (N + T) \bar{L}^2 ,
$$
which means that
$$
\Omega(A) \geq 0.5 \cdot \bbE_{| X} \Omega(A) - B r_L \theta_p (N + T) \bar{L}^2 .
$$
Then, the desired result follows from the definition of $\Omega(A)$ and $\bbE_{| X} \Omega(A)$. $\square$
\bigskip

\begin{lemma}\label{lem:tech_MNAR}
(i) We have $ U_{\calW} = U_{(N_0)} H_{\calW}$ where $U_{(N_0)} = [u_1, \cdots, u_{N_0}]^\top$ and $H_{\calW} = (U_{(N_0)}^\top U_{(N_0)})^{-1/2} G_{\calW}$ for some $r \times r$ orthogonal matrix $G_{\calW}$; (ii) We have $ U_{\calT} = U H_{\calT}$ where $H_{\calT} = (U^\top U)^{-1/2} G_{\calT}$ for some $r \times r$ orthogonal matrix $G_{\calT}$; (iii) $\norm{H_\calW} \lesssim \frac{\sqrt{N}}{\sqrt{N_0}} $, $\norm{H_\calW^{-1}} \lesssim \frac{\sqrt{N_0}}{\sqrt{N}} $, $\norm{Q_{\calW}} \lesssim \frac{\sqrt{N_0}}{\sqrt{N}}$, $\norm{Q_{\calW}^{-1}} \lesssim \frac{\sqrt{N}}{\sqrt{N_0}}$ with probability converging to 1. In addition, $\norm{H_\calT}$, $\norm{H_\calT^{-1}}$,  $\norm{Q_{\calT}}$, $\norm{Q_{\calT}^{-1}}$ are bounded; (iv) $||\left(\widehat{U}_{\calT,(N_0)}^\top \widehat{U}_{\calT,(N_0)}\right)^{-1}|| = O_p\left(\frac{N}{N_0}\right)$.
\end{lemma}

\noindent \textbf{Proof.} 
(i) Let $\Omega_{\calW} = (U_{(N_0)}^\top U_{(N_0)})^{1/2}  D^2  (U_{(N_0)}^\top U_{(N_0)})^{1/2}$ and $G_{\calW}$ be a $r \times r$ matrix whose columns are the eigenvectors of $\Omega_{\calW}$ such that $\Lambda_\calW = G_\calW^\top \Omega_\calW G_\calW$ is the descending order diagonal matrix of the eigenvalues of $\Omega_\calW$. Define $H_\calW = (U_{(N_0)}^\top U_{(N_0)})^{-1/2} G_{\calW}$. Then, we have
\begin{align*}
&(U_{(N_0)} D^2 U_{(N_0)}^\top)U_{(N_0)} H_\calW  \\
&= U_{(N_0)} \left( U_{(N_0)}^{\top}U_{(N_0)} \right)^{-1/2} \left( U_{(N_0)}^{\top} U_{(N_0)} \right)^{1/2} D^2 \left( U_{(N_0)}^{\top} U_{(N_0)} \right)^{1/2} \left( U_{(N_0)}^{\top} U_{(N_0)} \right)^{1/2} H_\calW \\
& = U_{(N_0)} \left( U_{(N_0)}^{\top} U_{(N_0)} \right)^{-1/2}  \left[ \left( U_{(N_0)}^{\top} U_{(N_0)} \right)^{1/2} D^2 \left( U_{(N_0)}^{\top} U_{(N_0)} \right)^{1/2} 
 G_\calW \right] \\
& = U_{(N_0)} \left( U_{(N_0)}^{\top} U_{(N_0)} \right)^{-1/2}  \Omega_\calW G_\calW
 = U_{(N_0)} \left( U_{(N_0)}^{\top} U_{(N_0)} \right)^{-1/2} G_\calW   \Lambda_\calW \\
 & = U_{(N_0)} H_\calW  \Lambda_\calW .
\end{align*}
In addition, note that $\left( U_{(N_0)} H_\calW \right)^\top \left( U_{(N_0)} H_\calW \right) = H_\calW^{\top} U_{(N_0)}^{ \top} U_{(N_0)} H_\calW = G_\calW^\top G_\calW = I_r$. Therefore, the columns of $ U_{(N_0)} H_\calW $ are the eigenvectors of $U_{(N_0)} D^2 U_{(N_0)}^{\top}$ and the left singular vectors of $U_{(N_0)} D V^\top$. \\
(ii) The proof is the same as that of the above.\\
(iii) Since $\norm{H_\calW} \leq \norm{\left(  U_{(N_0)}^{\top} U_{(N_0)}\right)^{-1/2}}\norm{G_\calW}$ and $G_\calW$ is an eigenvector matrix, we have $\norm{H_\calW} \lesssim \frac{\sqrt{N}}{\sqrt{N_0}} $ and $\norm{H_\calW^{-1}} \lesssim \frac{\sqrt{N_0}}{\sqrt{N}} $ with high probability by Assumption \ref{asp:incoherence}. Similarly, $\norm{H_\calT}$ and $\norm{H_\calT^{-1}}$ are bounded since $U^\top U = I_r$. In addition, because $Q_{\calW}^{-1} = H_{\calW}O_{\calW}^\top$, we have $\norm{Q_{\calW}} \lesssim \frac{\sqrt{N_0}}{\sqrt{N}}$ and $\norm{Q_{\calW}^{-1}} \lesssim \frac{\sqrt{N}}{\sqrt{N_0}}$ with high probability. Because $Q_{\calT}^{-1} = H_{\calT} O_{\calT}^\top$, we have $\norm{Q_{\calT}}$ and $\norm{Q_{\calT}^{-1}}$ are bounded.\\
(iv) First, note that
\begin{align*}
\norm{\widehat{U}_{\calT,(N_0)}^\top \widehat{U}_{\calT,(N_0)} - Q_{\calT}^{-\top} U_{N_0}^\top U_{N_0} Q_{\calT}^{-1} }
\lesssim \norm{U_{N_0} Q_{\calT}^{-1}} \norm{\widehat{U}_{(N_0),\calT}  - U_{(N_0)} Q_{\calT}^{-1} }_F &= O_p\left( \frac{\sqrt{N_0}}{\sqrt{N}}  \frac{\calR_{\calT}}{\lambda_{\min,\calT}} \right) \\
&= o_p\left( \frac{N_0}{N} \right)
\end{align*}
since $\norm{U_{N_0}} = O_p\left( \frac{\sqrt{N_0}}{\sqrt{N}}\right)$ and $\norm{\widehat{U}_{\calT}  - U Q_{\calT}^{-1} }_F = O_p\left( \frac{\calR_{\calT}}{\lambda_{\min,\calT}}\right)$. Then, because 
$$
\lambda_{\min} (Q_{\calT}^{-\top} U_{N_0}^\top U_{N_0} Q_{\calT}^{-1}) \geq \lambda^2_{\min}(Q_{\calT}^{-1}) \lambda_{\min} (U_{N_0}^\top U_{N_0}) \geq c \frac{N_0}{N},
$$
for some constant $c > 0$, we have with probability converging to 1 that
\begin{align*}
\lambda_{\min} \left( \widehat{U}_{\calT,(N_0)}^\top \widehat{U}_{\calT,(N_0)} \right) 
&\geq \lambda_{\min} (Q_{\calT}^{-\top} U_{N_0}^\top U_{N_0} Q_{\calT}^{-1}) - \norm{\widehat{U}_{\calT,(N_0)}^\top \widehat{U}_{\calT,(N_0)} - Q_{\calT}^{-\top} U_{N_0}^\top U_{N_0} Q_{\calT}^{-1} } \\
&\geq \frac{c}{2} \frac{N_0}{N}.
\end{align*}
Hence, $||\left(\widehat{U}_{\calT,(N_0)}^\top \widehat{U}_{\calT,(N_0)}\right)^{-1}|| = O_p\left(\frac{N}{N_0}\right)$. $\square$
\bigskip

Consider the model
$$
Y^o = M^o + R^o + E^o.
$$
Here, the noises $(\epsilon^o_{it})_{i \leq N, t \leq T}$ are independent, mean-zero, sub-Gaussian and satisfy $\bbE[(\epsilon^o_{it})^2] \leq C$ for some constant $C>0$. In addition, $E^o$ is independent of $\{M^o,X,Z\}$. $R^o$ is a matrix of negligible errors. Denote the condition number, incoherence parameter, and the rank of $M^o$ by $\kappa$, $\mu$, and $r$, respectively, and the smallest non-zero singular value of $M^o$ by $\Upsilon_{\min}$.

\begin{lemma}[MAR case]\label{lem:ccfmy_mar}
Denote the observation probability by $\bbE[\omega_{it}] = p$. If
\begin{align*}
& r \kappa^2 \mu^{1/2} \max \{ N \sqrt{\log N}, T \sqrt{\log T} \} \ll p^{1/2} \Upsilon_{\min} \min \{ \sqrt{N} , \sqrt{T} \} ,      \\
& r^2 \kappa^4 \mu^2 \max \{ N \log^3 N, T \log^3 T \} \ll p \min \{ N^2 , T^2 \},\\
& \max_{it} \abs{r^o_{it}}  \ll \frac{1}{\max\{ \sqrt{N}, \sqrt{T} \}},
\end{align*}
with the probability tending to one, then the corresponding nuclear norm penalized estimator of $M^{o}$ with the regularization parameter $\nu = C p^{1/2} \sqrt{N + T}$ for some large constant $C>0$ has the following bound:
$$
\max_{it} \abs{\widehat{m}_{it}^o - m_{it}^o } = O_p \left( \frac{\kappa^2 \mu r^{3/2} \max \{\sqrt{\log N}, \sqrt{\log T} \}}{p^{1/2} \min\{\sqrt{N}, \sqrt{T} \}} \right) .
$$
\end{lemma}

\noindent \textbf{Proof.} 
This result extends Theorem~2 of \cite{chen2020noisy} by accounting for the negligible error terms \(R^o\) and allowing \(N\) and \(T\) to
differ as in \cite{choi2024matrix,choi2024inference}. The proof easily follows by \cite{choi2024matrix,choi2024inference}. Hence, we omit it. $\square$
\bigskip

\begin{lemma}[Fully observed case]\label{lem:ccfmy_observed}
If we have the following relations with the probability tending to one,
\begin{align*}
& \max_{it} \abs{r^o_{it}}  \lesssim \frac{J^{1/2} \max\{ \sqrt{d_1 \mu_{\phi}},\sqrt{d_2 \mu_{\psi}} \} \max\{ \sqrt{ \log N} , \sqrt{ \log T}\} }{ \min \{ \sqrt{N} , \sqrt{T}  \}},\\
& 2 \norm{E^o + R^o} \leq \frac{7}{8}  \nu,
\end{align*}
with some regularization parameter $\nu = C \sqrt{N + T}$ for a large constant $C>0$, and 
\begin{align*}
& J^{1/2} \max\{ \sqrt{d_1 \mu_{\phi}},\sqrt{d_2 \mu_{\psi}} \} r \kappa^2 \mu^{1/2} \max \{ N \sqrt{\log N}, T \sqrt{\log T} \} \ll \Upsilon_{\min} \min \{ \sqrt{N} , \sqrt{T} \} ,      \\
& r^2 \kappa^4 \mu^2 \max \{ N \log^3 N, T \log^3 T \} \ll  \min \{ N^2 , T^2 \},
\end{align*}
then the corresponding nuclear norm penalized estimator of $M^{o}$ has the following bound:
$$
\max_{it} \abs{\widehat{m}_{it}^o - m_{it}^o } = O_p \left( \frac{J^{1/2} \max\{ \sqrt{d_1 \mu_{\phi}},\sqrt{d_2 \mu_{\psi}} \}\kappa^2 \mu r \theta \max \{\sqrt{\log N}, \sqrt{\log T} \}}{ \min\{\sqrt{N}, \sqrt{T} \}} \right),
$$
where $\theta = \max \{N,T \} / \min \{ N , T\}$.
\end{lemma}

\noindent \textbf{Proof.} 
This result extends Lemma \ref{lem:ccfmy_mar} by accounting for the slightly larger negligible error terms \(R^o\). The proof is similar to that of \cite{chen2020noisy,choi2024matrix}. The main difference is the part of bounding $\norm{D_1}$ and $\norm{D_2}$ in Lemma J.6 of \cite{choi2024matrix}, which is an extension of Claim 7 of \cite{chen2020noisy}. By using our bound of $\max_{it} \abs{r^o_{it}}$, we can derive new bounds for $\norm{D_1}$, $\norm{D_2}$, and it gives
\begin{align*}
&\max_{1\leq m \leq N_o+ T_o} \norm{   \calF_o^{\tau + 1} H_o^{ \tau + 1}  -    \calF_o^{\tau + 1, (m)} Q_o^{ \tau + 1  , (m)} } \\
&\ \ \lesssim
\frac{\sigma \kappa \theta^{1/2} J^{1/2} \max\{ \sqrt{d_1 \mu_{\phi}},\sqrt{d_2 \mu_{\psi}} \}   \max \{  \sqrt{N \log N}, \sqrt{T \log T} \}}{\Upsilon_{\min}} \norm{\calF_o}_{2,\infty}
\end{align*}
in Lemma J.6 of \cite{choi2024matrix}. Other parts are the same as those in \cite{choi2024matrix} with the assumption that $ \norm{E^o + R^o} \leq \frac{7}{16} \nu$. Here, our tuning parameter $\nu$ is $2 \lambda$ in \cite{chen2020noisy,choi2024matrix}. Then, we can have the desired results. Here, we omit the details for brevity. $\square$

\end{document}